\documentclass[aps,pra,onecolumn, superscriptaddress, longbibliography,nofootinbib]{revtex4-2}
\usepackage[letterpaper, left=1in, right=1in, top=1in, bottom=1in]{geometry}

\usepackage{cmap} 
\usepackage[utf8]{inputenc}
\usepackage[english]{babel}
\usepackage[T1]{fontenc}
\usepackage{appendix}
\usepackage{braket}
\usepackage{url, amsfonts, tikz, physics, listings, xcolor, graphicx, float}
\usepackage[shortlabels]{enumitem}
\usepackage{dsfont}
\usepackage{amsmath, amssymb, amstext, amscd, amsthm, makeidx, graphicx, url, mathrsfs, mathtools, longdivision, polynom, bbm, complexity}
\usepackage{nicefrac}
\usepackage[linesnumbered,ruled,vlined]{algorithm2e}
\usepackage{xcolor}
\usepackage[normalem]{ulem}
\usepackage{booktabs}
\usepackage{setspace}
\definecolor{blueviolet}{rgb}{0.2, 0.2, 0.6}
\definecolor{webgreen}{rgb}{0,.5,0}
\definecolor{webbrown}{rgb}{.6,0,0}
\usepackage[pdftex,
  bookmarks=false,
  colorlinks=true, 
  urlcolor=webbrown,
  linkcolor=blueviolet, 
  citecolor=webgreen,
  pdfstartpage=1,
  pdfstartview={FitH},  
  bookmarksopen=false
  ]{hyperref}
\usepackage{cleveref}

\makeatletter
\newcommand\RedeclareMathOperator{%
  \@ifstar{\def\rmo@s{m}\rmo@redeclare}{\def\rmo@s{o}\rmo@redeclare}%
}

\newcommand\rmo@redeclare[2]{%
  \begingroup \escapechar\m@ne\xdef\@gtempa{{\string#1}}\endgroup
  \expandafter\@ifundefined\@gtempa
     {\@latex@error{\noexpand#1undefined}\@ehc}%
     \relax
  \expandafter\rmo@declmathop\rmo@s{#1}{#2}}
\newcommand\rmo@declmathop[3]{%
  \DeclareRobustCommand{#2}{\qopname\newmcodes@#1{#3}}%
}
\@onlypreamble\RedeclareMathOperator
\makeatother

\allowdisplaybreaks

\newtheorem{definition}{Definition}
\newtheorem{theorem}{Theorem}
\newtheorem{lemma}{Lemma}
\newtheorem{conjecture}{Conjecture}
\newtheorem{corollary}{Corollary}

\newtheorem{alg}{Algorithm}

\newtheorem{formula}{Formula}
\newtheorem{assumption}{Assumption}

\newcommand{\paramv}{\vect{\Theta}}
\newcommand{\av}{\vect{a}}
\newcommand{\bv}{\vect{b}}

\newcommand{\xv}{\vect{x}}

\newcommand{\sv}{\vect{s}}

\newcommand{\zv}{\vect{z}}
\newcommand{\nv}{\vect{n}}
\newcommand{\mv}{\vect{m}}
\newcommand{\cA}{\mathcal{A}}

\RedeclareMathOperator*{\E}{{\mathbb{E}}}

\DeclareMathOperator*{\argmax}{arg\,max}
\DeclareMathOperator*{\argmin}{arg\,min}

\makeatletter
\DeclareFontFamily{OMX}{MnSymbolE}{}
\DeclareSymbolFont{MnLargeSymbols}{OMX}{MnSymbolE}{m}{n}
\SetSymbolFont{MnLargeSymbols}{bold}{OMX}{MnSymbolE}{b}{n}
\DeclareFontShape{OMX}{MnSymbolE}{m}{n}{
    <-6>  MnSymbolE5
   <6-7>  MnSymbolE6
   <7-8>  MnSymbolE7
   <8-9>  MnSymbolE8
   <9-10> MnSymbolE9
  <10-12> MnSymbolE10
  <12->   MnSymbolE12
}{}
\DeclareFontShape{OMX}{MnSymbolE}{b}{n}{
    <-6>  MnSymbolE-Bold5
   <6-7>  MnSymbolE-Bold6
   <7-8>  MnSymbolE-Bold7
   <8-9>  MnSymbolE-Bold8
   <9-10> MnSymbolE-Bold9
  <10-12> MnSymbolE-Bold10
  <12->   MnSymbolE-Bold12
}{}

\let\llangle\@undefined
\let\rrangle\@undefined
\DeclareMathDelimiter{\llangle}{\mathopen}%
                     {MnLargeSymbols}{'164}{MnLargeSymbols}{'164}
\DeclareMathDelimiter{\rrangle}{\mathclose}%
                     {MnLargeSymbols}{'171}{MnLargeSymbols}{'171}
\makeatother

\numberwithin{equation}{section}

\newcommand{\stkout}[1]{\ifmmode\text{\sout{\ensuremath{#1}}}\else\sout{#1}\fi}
\newif\ifverbose
\verbosefalse

\newcommand{\nocontentsline}[3]{}

\newcommand{\tocless}[2]{\bgroup\let\addcontentsline=\nocontentsline#1{#2}\egroup}

\newcommand{\Av}{\vect{\cA}}

\newcommand{\vect}[1]{\boldsymbol{#1}}

\usepackage{xargs}

\usepackage[colorinlistoftodos,prependcaption,textsize=tiny]{todonotes}

\newcommandx{\ltodo}[2][1=]{\todo[linecolor=red,backgroundcolor=red!25,bordercolor=red,#1]{L: #2}}

\newcommandx{\itodo}[2][1=]{\todo[linecolor=green,backgroundcolor=green!25,bordercolor=green,#1]{I: #2}}

\usepackage{mathtools}

\DeclareMathOperator{\EV}{\mathbb{E}}

\newcommand{\iid}{{i.i.d.}}

\newcommand{\cD}{\mathcal{D}}

\newcommand{\bbS}{\mathbb{S}}

\newcommand{\HamQAOA}{{\textnormal{HamQAOA}}}

\begin{document}

\title{A Quantum Approximate Optimization Algorithm \\ for Local Hamiltonian Problems}

\author{Ishaan Kannan}
\email{ishaan@alumni.caltech.edu}
\affiliation{California Institute of Technology, Pasadena, CA 91125}

\author{Robbie King}
\email{wking@caltech.edu}
\affiliation{California Institute of Technology, Pasadena, CA 91125}

\author{Leo Zhou}
\email{leoxzhou@ucla.edu}
\affiliation{California Institute of Technology, Pasadena, CA 91125}
\affiliation{University of California, Los Angeles, CA 90095}

\begin{abstract}
    Local Hamiltonian Problems (LHPs) are important problems that are computationally QMA-complete and physically relevant for many-body quantum systems. Quantum MaxCut (QMC), which equates to finding ground states of the quantum Heisenberg model, is the canonical LHP for which various algorithms have been proposed, including semidefinite programs \cite{parekh_et_al:LIPIcs.ICALP.2021.102, lee2024improvedquantummaxcut, King_2023} and variational quantum algorithms \cite{ vqcs,kagome,agm_beyond}. We propose and analyze a quantum approximation algorithm which we call the Hamiltonian Quantum Approximate Optimization Algorithm (HamQAOA), which builds on the well-known scheme~\cite{farhi2014quantum} for combinatorial optimization and is suitable for implementations on near-term hardware. We establish rigorous performance guarantees of the HamQAOA for QMC on high-girth regular graphs, and our result provides bounds on the ground energy density for quantum Heisenberg spin glasses in the infinite size limit that improve with depth. Furthermore, we develop heuristic strategies with which to efficiently obtain good HamQAOA parameters. Through numerical simulations, we show that the HamQAOA empirically outperforms prior algorithms on a wide variety of QMC instances. In particular, our results indicate that the linear-depth HamQAOA can deterministically prepare exact ground states of 1-dimensional antiferromagnetic Heisenberg spin chains described by the Bethe ansatz, in contrast to the exponential depths required in previous protocols for preparing Bethe states \cite{Van_Dyke_2021, Li_2022, Sopena_2022, raveh2024deterministicbethestatepreparation}. 
    
\end{abstract}
\nocite{*}
\maketitle
\singlespacing
\vspace{-2em}
\fontsize{9}{9.5}\selectfont
{\tableofcontents}
\normalsize

\vspace{-1em}
\section{Introduction}
Many problems in computer science ranging from combinatorial optimization, clustering, constraint satisfication, and beyond are studied through the lens of approximation algorithms due to complexity-theoretic restrictions on finding their exact solutions. The approximability of such problems, such as the ubiquitous MaxCut problem, is well-understood \cite{khot2004, goemans_williamson}; however, quantum physics brings us a plethora of computational problems that are even more difficult for classical computers, and whose approximability by quantum computers remains elusive. The canonical family of such problems are called Local Hamiltonian Problems (LHPs), in which one must find the ground state energy of a Hamiltonian composed of terms acting only on a few sites of a many-body quantum state \cite{kempe2005complexity}. LHPs are studied as important models of emergent quantum behavior, such as high-temperature superconductivity \cite{Zhanghtc, Hirschhtc, andersonhtc}, quantum spin liquids \cite{Kitaev_2006, Kitaev_2003}, and a wide variety of other theories spanning condensed matter physics, quantum error-correction, quantum chemistry, quantum field theories, and gravity \cite{Kitaev_2006, Dennis_2002, Almheiri_2015,witten2003chernsimonsgaugetheorystring, Rangamani_2017}. A frequently used benchmark task for approximation algorithms is a particular LHP with only 2-local terms known as Quantum MaxCut (QMC), a natural quantization of classical MaxCut. Solving QMC equates to finding the ground state of the antiferromagnetic (AFM) quantum Heisenberg model, a fundamental model of quantum magnetism.

Most classical approaches to solving QMC have focused on finding product-state approximations to its ground state. However, it is known that product states can only achieve an approximation ratio (the ratio between the energies of the output state and the true ground state) of 0.5 in the worst case due to the often highly-entangled nature of these quantum ground states \cite{parekh_prod_states}. One family of algorithms which output entangled quantum states round the solutions of classical semidefinite programs \cite{parekh_et_al:LIPIcs.ICALP.2021.102,lee2024improvedquantummaxcut, King_2023}, achieving worst-case approximation ratios of 0.595, 0.582 and 0.533 respectively. Moreover, \cite{agm_beyond} proposes a variational algorithm finding a shallow quantum circuit that provably outperforms any product-state solution. Despite these analytical results, however, the empirical performance of these techniques on LHPs of interest is not well-studied.

Alternatively, quantum algorithms are better suited to efficiently prepare entangled ground states. Algorithms such as quantum phase estimation \cite{kitaev1995quantummeasurementsabelianstabilizer} or thermal state preparation \cite{chen2023quantumthermalstatepreparation} are prominent approaches to ground-state preparation; however, the former requires an input state sufficiently close to the ground state to succeed, while the latter is challenging to implement at scale on near-term quantum devices. Due to these restrictions, variational quantum algorithms have been an often-utilized approach to prepare nontrivial quantum ground states, and have been tailored to specific problems such as the AFM Heisenberg model on 1d chains and the Kagome lattice \cite{kagome}, or GHZ states and 1d transverse field Ising model \cite{vqcs}, realizing either exact target states or order-1 fidelity approximations with relatively low-depth circuits. However, these tailored algorithms rely on known properties of the desired ground states and assumptions of the problem geometry that apply only to highly structured problem instances, and lack rigorous performance guarantees when applied to different LHPs at arbitrary depth.

In this work, we propose a flexible and near-term friendly variational quantum algorithm for LHPs that we call the Hamiltonian-QAOA (HamQAOA). Our algorithm is a generalization of the Quantum Approximate Optimization Algorithm (QAOA), a variational quantum approximation algorithm designed for classical combinatorial optimization problems \cite{farhi2014quantum}. The QAOA circuit consists of $p$ layers of gates, with a constant number of parametrized unitary drivers to optimize at each layer, and has inspired the well-studied variational quantum eigensolver (VQE) ansatz \cite{Peruzzo_2014} known as the Hamiltonian variational ansatz (HVA) \cite{HVA}. Since the QMC Hamiltonian is invariant under $SO(3)$ spin-rotations, any quantum circuit preserving the total angular momentum of the initial state can only produce approximations of the ground state that lie within the same angular momentum sector as the initial state, severely limiting the performance of the original QAOA and HVA 
when the angular momentum of the ground state is not known. This observation motivates our new ansatz, which utilizes an inhomogeneous rotation that breaks the rotational symmetry to overcome this limitation. (The specific differences between the HamQAOA and previous variational frameworks are discussed in Section \ref{sec:background}.) The HamQAOA thus has a similar circuit implementation to prior near-term-friendly algorithms, but can be flexibly applied to different problem Hamiltonians with improved performance.

Through theoretical analysis, we provide various performance guarantees for the HamQAOA. First, we prove that the $p=1$ HamQAOA performs at least as well as the variational algorithm from \cite{agm_beyond} on any graph, and the SDP-relaxation algorithm from \cite{King_2023} on any edge-transitive graphs. The performance of the HamQAOA monotonically improves with depth $p$, and we prove that it achieves the exact QMC ground state in the $p\rightarrow \infty$ limit on bipartite graphs (and conjecturally on general graphs). We also give average-case performance guarantees for the HamQAOA on large random regular graphs at any finite depth through an iterative formula that computes the expected energy achieved by the algorithm as a function of its parameters.
In the thermodynamic limit, our formula serves as a novel algorithmic upper bound on the ground state energy density of quantum Heisenberg spin glass models. A better understanding of our formula may yield an educated guess of the true ground energy of this model with all-to-all connectivity. This quantity is a quantum analogue to the Parisi value, generalizing the ground energy density of the classical Sherrington-Kirkpatrick (SK) spin-glass model \cite{Talagrand2006ThePF, parisi_inf} to quantum Heisenberg spin glasses. Furthermore, our formula can be used to obtain good HamQAOA parameter initializations on locally tree-like graphs. 

Empirically, we provide three heuristic strategies with which to efficiently optimize the HamQAOA's parameters and use these methods to benchmark its performance against other prior approaches on various problems. We perform this comparison on a large array of dense Erdos-Renyi graphs and sparse random 3-regular graphs, where we find that HamQAOA of depth only $2$ convincingly beats the previous best QMC-specific algorithms. In the case of periodic 1-dimensional Heisenberg chains, whose QMC ground energy density is known via the Bethe ansatz \cite{1931.Bethe.ZP.71} to be $-2\ln2\approx -1.386$ in the infinite-size limit, the HamQAOA of only depth $p=5$ achieves energy density $-1.3717$ and approaches the ground energy with excitation-energy density scaling roughly as $O(1/p)$. In contrast, a prior state-of-the-art algorithm based on semi-definite programming relaxation in Ref.~\cite{King_2023} does no better than $-1.3090$.  Moreover, at $p=2N$ (where $N$ is the system size), the HamQAOA shows empirical evidence of preparing states that have near-unity fidelity with the true ground state.
We further observe that the $p=4$ and $7$ HamQAOA successfully prepare the exact Bethe ground states (up to floating-point precision) of rings of size $N=4$ and $6$ respectively, a feat not achieved by any other variational quantum algorithm to our knowledge. 

These observations lead us to conjecture that the HamQAOA can prepare the exact Bethe ground states of 1-dimensional Heisenberg chains using a quantum circuit of $\Theta(N)$ depth.
We remark that these Bethe states are not only relevant in condensed matter physics but also encode challenging computational problems \cite{crichigno2024quantumspinchainssymmetric}, and existing circuit constructions for deterministically preparing such Bethe states require $\exp(N)$ depths \cite{Van_Dyke_2021, Li_2022, Sopena_2022, raveh2024deterministicbethestatepreparation}.
We repeat this case study of 1d ring graphs for the quantum XY model and Heisenberg XXZ model in an external field lying in a Luttinger-liquid phase of matter, demonstrating the HamQAOA's versatility for different LHPs.

\vspace{-1em}
\section{Background}\label{sec:background}
\subsection{Quantum MaxCut}\label{sec:QMC}
\begin{definition}[Quantum MaxCut]
\label{Quantum MaxCut}
Given a graph $G = (V, E)$ whose vertices are qubits, the \emph{Quantum MaxCut (QMC) problem} refers to the preparation of the quantum state corresponding to the maximum eigenvalue of the Hamiltonian 
\begin{equation}
\label{QMCHam}
    H_{\rm QMC} = \sum_{i \sim j} \frac{1}{2}(\mathbbm{1} - X_i X_j - Y_i Y_j - Z_i Z_j)
\end{equation}
where $i\sim j$ denotes $(i, j) \in E$ and $X, Y, Z$ are the standard Pauli operators.
\end{definition}
Note that QMC is the natural quantum generalization of the classical MaxCut problem, whose cost function is given by
\begin{equation}
    H_{\rm MaxCut} = \sum_{i\sim j} \frac{1}{2} (1-Z_i Z_j)
\end{equation}
which is a diagonal matrix in the standard basis, where we may treat $Z_i\in \{\pm 1\}$ as classical bit variables.

In a more physical context, the 2-local Hamiltonian $H_{\rm QMC}$ is known as the AFM Heisenberg XXX model with no transverse field. In one dimension, this problem is exactly solvable via the Bethe Ansatz \cite{1931.Bethe.ZP.71}. More generally, the study of Heisenberg model ground states is a longstanding problem of interest in condensed matter theory and quantum chemistry \cite{Suzuki_2021,Deumal2021InsightsIT, 1931.Bethe.ZP.71, Zhanghtc, Hirschhtc, andersonhtc, Kitaev_2006, Kitaev_2003}.
Furthermore, the weighted variant of the QMC problem is not only QMA-complete~\cite{CubittMontanaro2016} but also a universal family of Hamiltonians that can encode all other local Hamiltonians in the low-lying part of its spectrum~\cite{CMP2018universal, ZhouAharonov2021}.
Even approximating the unweighted $H_{\rm QMC}$ to higher than a 0.956 ratio is believed to be computationally hard \cite{hwang2022unique} assuming the Unique-Games conjecture, which ultimately limits the potential of purely classical algorithms (though none exist which saturate the known classical upper bound \cite{Briet_2014}). Thus, we turn to the quantum setting in the hopes of finding an empirically better-performing, near-term solution.

Minimizing the subtracted part of $H_{\rm QMC}$ is equivalent to finding approximate ground states of the AFM Heisenberg XXX model and manifestly maximizes the energy of the QMC Hamiltonian. Throughout the rest of this paper, we discuss energy minimization rather than maximization due to this equivalence.
 
\vspace{-2em}\subsection{Quantum Approximate Optimization Algorithm and Hamiltonian Variational Ansatz}\label{sec:QAOA_intro}

The QAOA is an algorithm proposed by \cite{farhi2014quantum} for approximately solving classical combinatorial optimization problems. It guarantees monotonically improving performance in the number of layers applied (commonly denoted by $p$), with a number of variational parameters scaling linearly with $p$.
The input to the QAOA includes a classical problem whose cost function is encoded in a classical Hamiltonian (i.e., a diagonal operator) $A$, and another ``mixing'' driver Hamiltonian $B$ whose ground state $\ket{g_B}$ is known and efficiently preparable. Then we define $2p$ unitary evolution operators $e^{-i\alpha_j A}, e^{-i\beta_j B}$ for $j \in \{1, ..., p\}$, with all $\alpha_j, \beta_j \in \mathbb{R}$. The QAOA prepares the following state 
\begin{equation}
   \ket{\psi_{\vect\alpha,\vect\beta}} = e^{-i \beta_p B}e^{-i \alpha_p A}...e^{i \beta_1 B}e^{-i \alpha_1 A}\ket{g_B}
\end{equation}
with the output energy (i.e., the cost function value) given by $\braket{\psi_{\vect\alpha,\vect\beta}}{A|\psi_{\vect\alpha,\vect\beta}}$. The variational parameters are tuned to minimize this energy.

The QAOA is attractive partly due to an asymptotic guarantee that it can find a ground state of the classical Hamilonitan $A$ as $p\to\infty$. To see this, we use the QAOA to simulate a discretization of the Quantum Adiabatic Algorithm (QAA) in which a system is initialized in the known ground state of a simple Hamiltonian $B$ and time evolved to reach the unknown ground state of a complex Hamiltonian $A$ \cite{farhi2000quantumcomputationadiabaticevolution}. This is done by evolving the system over a time $T$ with a time-dependent $H(t) = (t/T) A + (1-t/T) B$, and converges to the exact ground state after sufficiently large $T$ if and only if there is a nonzero spectral gap between the smallest two eigenvalues of $H(t)$ for all $t \in [0, T)$, a guarantee due to the adiabatic theorem. As a concrete example, one can take the classical MaxCut objective Hamiltonian $A=H_{\rm MaxCut} = \sum_{j\sim k} (1-Z_j Z_k)/2$ and mixing Hamiltonian $B  = \sum_i X_i$ (both of which can be readily implemented on most quantum hardware platforms), and apply the Perron-Frobenius theorem to demonstrate that $H(t)$ has a gapped spectrum for all $t\in [0, T)$, yielding the guarantee for the QAA to prepare a ground state of $A$ as $T\to\infty$. A more careful argument, e.g., Theorem 12 of \cite{binkowski2023elementary}, demonstrates that the QAOA with drivers $A$ and $B$ can guarantee the same convergence in the $p\rightarrow \infty$ limit when $H(t)$ is gapped. The argument relies on the fact that as $p$ grows, the QAOA circuit can better simulate a Trotter approximation of the evolution under $H(t)$, and converges to the adiabatic algorithm in the $p\to\infty$ limit. In this work, we use this technique to prove convergence of the HamQAOA for QMC at large depth on bipartite graphs (see Theorem \ref{Thm:bipartite_guarantee} in Section~\ref{sec:guarantees}). Even at depth $p$\,=\,1, the QAOA and HamQAOA are not efficiently classically simulable~\cite{farhi2016supremacy} and exhibits nontrivial performance guarantees~\cite{farhi2014quantum}, making it an appealing framework for finding quantum advantage for LHPs. 

The HVA \cite{HVA} has a similar setup to the QAOA. To target the ground state of a Hamiltonian $H$, one first decomposes $H$ into a sum of $q$ terms, $H = \sum_{j=1}^q H_j$, then chooses a depth $p$ and an easily-prepared initial state $\ket{\psi_0}$. Given these choices, the HVA state is 
\begin{equation}
    \ket{\psi_{\vect\theta}} = \left(\prod_{i=1}^p e^{-i\theta_{i, q}H_q}\hdots e^{-i\theta_{i, 1}H_1}\right)\ket{\psi_0}
\end{equation}
where the $p\cdot q$ variational parameters $\theta_{i,j}$ are chosen to minimize the energy $\braket{\psi_{\vect\theta}}{H|\psi_{\vect\theta}}$. A key difference between the HVA and QAOA is that the QAOA often uses drivers that are not terms in the Hamiltonian. Leveraging this crucial distinction allows our HamQAOA to make significant improvements over other VQE frameworks for LHPs.

\section{Hamiltonian Quantum Approximate Optimization Algorithm}

\begin{alg}[General Hamiltonian QAOA]
\label{GHamQAOA}
    Given an instance of any $2$-local Hamiltonian 
    \begin{equation}
        H = \sum_{i\sim j} J_{ij}\sigma_i\sigma_j, \ J_{ij}\in \mathbb{R}, \ \sigma_i, \sigma_j \in \{I, X, Y, Z\}
    \end{equation} 
    with interaction graph $G = (V, E)$, the \emph{General Hamiltonian QAOA} is defined as follows. Select collections of unit vectors $\nv = \{\nv_v\}_{v\in V}, \mv = \{\mv_v\}_{v\in V}$ where each $\nv_v,\mv_v$ lying on the 3-dimensional unit sphere, and a positive integer $p$. Define four driver Hamiltonians:
    \begin{equation}
    \begin{split}
    A = \sum_{u \sim v} Z_u Z_v, \qquad
    B = \sum_v X_v, \qquad
    C = \sum_v Z_v, \qquad
    D = \sum_v \nv_v \cdot \vect{\sigma}_v,
    \end{split}
    \end{equation}
    where $(u, v) \in E$ and each operator $X_v, Y_v, Z_v$ is a local Pauli acting nontrivially only on the qubit at vertex $v$, with $\vect{\sigma}_v = (X_v, Y_v, Z_v)$.
    Let $\ket{\mv_v}$ be the single-qubit state on the Bloch sphere corresponding to the unit vector $\mv_v$. Given parameters $\paramv = (\vect{\alpha}, \vect{\beta}, \vect{\gamma}, \vect{\delta})$, where each $\vect{\alpha}, \vect{\beta}, \vect{\gamma}, \vect{\delta} \in (-\pi/2, \pi/2]^p$, we define the \emph{HamQAOA state (HQS)} as
    
    \begin{equation}  \label{eq:hamqaoa_circuit}
        \ket{HQS_p\left(\paramv, G, \nv, \mv\right)} = e^{-i\delta_p D} e^{-i\gamma_p C} e^{-i\beta_p B} e^{-i\alpha_p A} \dots e^{-i\delta_1 D} e^{-i\gamma_1 C} e^{-i\beta_1 B} e^{-i\alpha_1 A} \bigg( \bigotimes_v |\mv_v\rangle \bigg).
    \end{equation}
    With fixed $G, \nv$ and $\mv$, let $\ket{\paramv} = \ket{HQS_p\left(\paramv, G, \nv, \mv\right)}$ . Then the algorithm output is 
    \begin{equation}
        {\HamQAOA}_p(G, \nv, \mv) = \ket{\paramv^*}, \  \text{where} \ \   \paramv^* = \argmin_{\paramv} \langle\paramv|H|\paramv\rangle
    \end{equation} 
    In practice, the global minimum $\paramv^*$ may be difficult to find; one can use the optimization heuristics defined in Section \ref{subsec:prac_use} to find a local minimum that provides a good approximation.
\end{alg}

The intuition behind this generalized choice of drivers is as follows. One can view a single layer of the HamQAOA as the composition of commuting interaction terms and an arbitrary single-qubit rotation on each qubit, specified by the choices of vectors $\nv_v$. Furthermore, the HamQAOA can Trotter-approximate $e^{-i \theta H_{\rm QMC}}$ for any $\theta$, allowing us to recover asymptotic performance guarantees via reduction to the QAA in the same vein as the original QAOA. 

Crucially, a central challenge in solving QMC is the problem's rotational symmetry. Observe that the QMC Hamiltonian \eqref{QMCHam} is invariant under global spin-rotation of all qubits,  which implies that the ground state of QMC has a fixed total angular momentum. If the QMC ground state lies in a different angular momentum sector of the Hilbert space than the initial state, then any algorithm which acts homogeneously on each qubit at every circuit layer (such as the original QAOA) will preserve the total angular momentum of the initial state, and thus can never find the QMC state. This would make an asymptotic convergence guarantee to the ground state impossible.

Prior variational approaches have been developed for problems where the angular momentum sector of the ground state is known---such as the 1d periodic chain and Kagome lattice\cite{vqcs, kagome}---so that the initial state could be prepared accordingly in the correct sector. In contrast, the $D$ driver in our algorithm can be used to break this problematic symmetry by choosing the vectors $\nv_v$ to be spatially varying, so that $D$ does not commute with the total angular momentum operator. Our HamQAOA ansatz thus allows the algorithm to access any angular momentum sector of the Hilbert space to tackle arbitrary problem geometries and Hamiltonians.

The freedom to choose $\nv_v$ and $\mv_v$, each of which can be described by two polar angles, yields a large total of $4|V|$ tunable parameters in selecting the algorithm ansatz. However, one of the key strengths of the ``vanilla'' QAOA~\cite{farhi2014quantum} was having a number of variational parameters per layer that does not scale with system size. This raises the question of how to find the best, simplified choice of HamQAOA ansatz for the QMC problem. We find that restricting each $\nv_v$ to the $\hat{\xv}$-axis and fixing every $\mv_v = \nv_v$ yield both strong empirical results and a connection to the classical MaxCut problem. The choice of setting $\mv_v = \nv_v$ also recovers the large-$p$ asymptotic reduction to the adiabatic algorithm from the vanilla QAOA, since the initial state then coincides with the ground state of the $D$ driver. This brings us to the simplified HamQAOA algorithm: 
\begin{alg}[Simplified Hamiltonian QAOA]
\label{HamQAOA}
 Given an LHP with interaction graph $G = (V, E)$, a bitstring $\vect{s} \in \{\pm 1\}^{|V|}$ and $p\in\mathbb{N}$, let $\nv'_v = \mv'_v = (s_v, 0, 0)$. Then the \emph{simplified HamQAOA} is defined as
 \begin{equation}
     \HamQAOA_p(G, \vect{s}) \coloneqq \HamQAOA_p(G, \nv', \mv')
 \end{equation}
 where $\nv', \mv'$ are the collections of all $\nv_v', \mv_v'$ as in Algorithm \ref{GHamQAOA}.
\end{alg}
In short, the simplified Hamiltonian QAOA retains the same drivers $A, B, C$ from Algorithm \ref{HamQAOA}, but with 
\begin{equation}
     D = \sum_v s_v X_v
\end{equation}

For most of this paper, we will focus on applying Algorithm \ref{HamQAOA}, but we compare its performance to the more general Algorithm~\ref{GHamQAOA} in Appendix \ref{sec:gen_ansatze}. There, we find that optimizing the general ansatz over choice of vectors $\nv_v$ and $\mv_v$ is computationally difficult for QMC since the number of parameters increases with system size, and that due to the hardness of finding good vectors the simplified HamQAOA performs just as well for practical purposes. 
Moreover, we find that when fixing $\nv = \mv$ and computing the average algorithm performance where these vectors are sampled either (a) uniformly from the unit sphere, as the general ansatz allows, versus (b) equiprobably on the $\pm \hat{\xv}$-axis, as in the simplified algorithm, the simplified HamQAOA wins out. We also find that the simple Algorithm \ref{HamQAOA} is highly effective for LHPs beyond the Heseinberg XXX model, as detailed in Section \ref{sec:otherlhp}. Whether the general ansatz can be tailored differently for specific LHPs is an interesting consideration for practical many-body problems.

We remark that the generalization of our HamQAOA algorithm to Hamiltonians of an arbitrary locality $k$ is immediate: we simply replace driver $A$ with a sum over $k$-local Z operators matching the connectivity of the target Hamiltonian.

\subsection{Practical Usage of HamQAOA}\label{subsec:prac_use}

To approximately solve real QMC instances using the simplified HamQAOA, one must choose the values $s_v$ and a strategy with which to find good parameters $\paramv$.

\begin{itemize}
    \item \textbf{Choosing $s_v$}: The values $s_v\in \{\pm1\}$ should coincide as closely as possible with a bitstring yieiding the classical MaxCut of the interaction graph (see Appendix \ref{sec:D_Driver_choice}).
    \begin{itemize}
        \item When the average vertex degree is $o(N)$ for system size $N$ and the exact MaxCut is not known, the $s_v$ should be chosen to give the largest cut found by an approximation algorithm, e.g., \cite{goemans_williamson}.
        \item When the average degree is $\Theta(N)$, then $s_v$ can be chosen uniformly at random in the large $N$ limit. This is because for large vertex degrees, random cuts become asymptotically close to the true MaxCut in the large-size limit.
    \end{itemize}
    \item \textbf{Parameter Optimization Strategies:}
    \begin{enumerate}
        \item Greedy Iterative (GI) strategy (heuristic simplification of \cite{Sack_2023}): First optimize at $p=1$, sampling initial parameters from $[-\frac{\pi}{2}, \frac{\pi}{2})^4$ 
        uniformly at random many times. Take the set of optimized parameters which yield the best energy and append a layer of 0's for the second layer of the HamQAOA. Iterate this procedure until depth $p$. Further adjustments to this strategy are detailed in Appendix \ref{sec:GI}. The GI strategy is particularly useful when the interaction graph is far from regular (where the following strategy outperforms) and $p$ is not so large that this greedy method gets stuck in bad local minima.
        \item Iterative Formula Parameter strategy: For a graph with average degree $d$ and HamQAOA of depth $p$, use the best parameters obtained from the finite-degree formula \ref{alg:fin_iter} with the same $d$ and $p$ as the parameter initialization. This strategy is useful when $p$ is small enough to numerically optimize the formula  and $d$ is large enough that random cuts are close to the MaxCut, and performs particularly well when the graph is nearly regular.
        \item Random Initialization strategy: For HamQAOA depth $p$, sample all parameters independently and uniformly at random from $[-\frac{\pi}{2}, \frac{\pi}{2})$ and optimize many times. This strategy is useful when the HamQAOA is highly overparametrized, such as when $p \gg$ the diameter of the graph, because we find in this regime that the variational landscape is populated with many easily-found local minima close to the global minimum. Alternatively, if a large $p$ is required to solve a problem, where the GI strategy gets stuck in bad minima and the iteration is intractable, one can use random initialization.
    \end{enumerate}
    
\end{itemize}

\subsection{Performance Guarantees}
\label{sec:guarantees}

Using this construction, we prove the following asymptotic performance guarantee:
\begin{theorem}\label{Thm:bipartite_guarantee}
    Given $H_{\rm QMC}$ on bipartite graph $G = (V, E)$, select $\vect{s} \in \{\pm 1\}^{|V|}$ to coincide with the exact classical MaxCut of G. For $p\geq 1$ define $\ket{\paramv_p^*} = \HamQAOA_p(G, \vect{s})$ for the objective Hamiltonian $H=-H_{\rm QMC}$. Then
    \begin{align}
        \lim_{p\to\infty} \braket{\paramv_p^*}{H_{\rm QMC} |\paramv_p^*} = \lambda_{\max} (H_{\rm QMC})
    \end{align}
    where $\lambda_{\max}(H_{\rm QMC})$ denotes the largest eigenvalue of the QMC Hamiltonian.
   
\end{theorem}
Intuitively, the proof of the above theorem uses the fact that the HamQAOA can be used to simulate an adiabatic evolution to the QMC eigenstate which has an avoided eigenstate crossing. While this performance guarantee is only proven for bipartite graphs (which include many physically relevant cases such as square and hexagonal lattices), the rotational symmetry-breaking of our ansatz makes eigenstate crossings in the adiabatic evolution unlikely even on general graphs. Hence, we expect this asymptotic guarantee to hold more generally.

Furthermore, our ansatz provably encapsulates some of the best prior approaches:
\begin{lemma}\label{lemma:contains_prev}
    The $p=1$ HamQAOA performs no worse than the Anshu-Gosset-Morenz algorithm \cite{Anshu_2021} on any graph. Furthermore, the $p=1$ HamQAOA performs no worse than the algorithm from \cite{King_2023} on edge-transitive graphs, that is, graphs for which for any pair of edges $e_1, e_2$ there exists an automorphism of the graph mapping $e_1$ to $e_2$.
\end{lemma}
We defer the proofs of these performance guarantees to Appendix \ref{sec:proof_guarantees}.

\section{Analysis of HamQAOA on High-Girth Regular Graphs}
\label{sec:Iteration}

While analysis of the HamQAOA is difficult on general graphs, we are able to derive an exact formula for the average-case performance of our algorithm on high-girth regular graphs.
Specifically, we give an iterative formula that computes the energy per edge of the output state of the HamQAOA of any depth $p$ on regular graphs of degree $d+1$ and girth $> 2p + 1$, averaged over a uniformly random choice of signs $s_v$ in the $D$ driver of our Algorithm~\ref{HamQAOA}.
Similar to Ref.~\cite{Leo_MaxCut}, our result on the high-girth regular graphs also applies to evaluate the algorithm's performance on random regular graphs since the latter are locally treelike.
Moreover, the parameters obtained from optimizing this formula allow us to efficiently optimize other locally treelike regular graphs and more general weighted problems (see e.g.,~\cite{Sureshbabu2024parametersetting}).
    
As an additional consequence, our formula yields an upper bound to the ground energy density of the (diluted) quantum Heisenberg spin glass.
To our knowledge, there is currently no rigorous understanding of ensemble-averaged ground state energy densities for the quantum Heisenberg spin glass in the same way that we have a Parisi value for the classical Sherrington-Kirkpatrick (SK) model \cite{parisi_inf, Talagrand2006ThePF}. In the classical setting, previous works~\cite{farhi2019SK, Leo_MaxCut} has identified that the energy achieved by the QAOA on MaxCut appears to provide a tight bound on the Parisi value. Therefore, our work takes a step towards identifying an analogous ``quantum Parisi value'' for the Heisenberg model.

Note that our formalism allows the evaluation of our algorithm's average-case performance for any 2-local objective Hamiltonian composed of identical interaction terms (possibly with varying coupling strengths) across all edges on a high-girth regular interaction graph, and applies to a broad array of choices of HamQAOA drivers, distributions over the choices of single-qubit rotation axes $\nv_v$ and initial product states $\ket{\mv_v}$ whose tensor product forms the HamQAOA input state.

Given this flexibility, our iterative formula can be used to compare the performance of different HamQAOA ansatze, to benchmark the HamQAOA on other 2-local LHPs, or obtain bounds on the ground energy of different many-body physical models.

\subsection{Finite-Degree Iterative Formula}
The goal of our analysis is to compute a 2-local observable expectation value $\langle{\sigma_L \sigma_R}\rangle$ for some Pauli operators $\sigma_{i}\in\{I,X,Y,Z\}$ acting on two neighboring qubits $L,R$.
The key observation is that high-girth regular graphs are particularly amenable to analysis since the Heisenberg picture operator
\begin{equation}
    \hat{O}_{\paramv}(\sigma_L, \sigma_R) = e^{i\delta_p D} e^{i\gamma_p C} ... e^{i\beta_1 B} e^{i\alpha_1 A}\sigma_L\sigma_R e^{-i\delta_p D} e^{-i\gamma_p C} ... e^{-i\beta_1 B} e^{-i\alpha_1 A}
\end{equation}
acts nontrivially only on vertices who are $p$ or less edges far from $L$ and $R$.
Hence, for a regular graph with degree $d+1$ and girth $> 2p+1$, the lightcone of local operators seen the HamQAOA is supported on a $p$-level $d$-ary tree subgraph, with each vertex having $d$ children (see Figure \ref{fig:D_reg_tree} in Appendix~\ref{sec:proof_iter}).

The following iterative formula can be evaluated in $\tilde{O}(16^p)$ time using $O(4^p)$ memory for the General HamQAOA of depth $p$ with randomized choices of all $\mv_v, \nv_v$, and evaluates the expectation of any two-local Pauli observable on a single edge of a regular graph of girth $> 2p+1$, written as 
\begin{equation}
    \bra{\paramv}\sigma_L\sigma_R\ket{\paramv} = \bra{\vect{m}} \hat{O}(\sigma_L, \sigma_R) \ket{\vect{m}}
\end{equation}
such that the Heisenberg expectation on a single edge of the QMC Hamiltonian~\eqref{QMCHam} can be computed as
\begin{equation}
    \frac12 - \frac{1}{2}\mathbb{E}_{\mv, \nv}[\bra{\vect{m}} \hat{O}(X,X) +  \hat{O}(Y,Y) + \hat{O}(Z,Z) \ket{\vect{m}}].
\end{equation}
The expectation $\EV$ is taken over $\mv = \{\mv_v\}_{v\in V}$ and $\nv = \{\nv_v\}_{v\in V}$ within the HamQAOA $D$ driver for each of $N$ vertices. 

Notationally, we index a vector $\vect{x}$ of length $2p+2$ as
\begin{equation}
    \vect{x} = (x_1,  x_2, ... x_p, x_{p+1}, x_{-(p+1)}, x_{-p}, ..., x_{-2}, x_{-1})
    \label{idx_convention}
\end{equation}
and the element-wise product of vectors $\vect{x}$ and $\vect{y}$ as $\vect{xy}$. We define 
\begin{align}
    &\vect{\mathcal{A}} = (\alpha_1, ... \alpha_p, 0, 0, -\alpha_p, ..., \alpha_1)\\
    &E_j(\beta_j, \gamma_j, \delta_j) =\begin{cases}
         e^{i\beta_j B}e^{i\gamma_j C}e^{i\delta_j D}, \qquad &j>0\\
         E_{|j|}^\dagger,\qquad &j<0
    \end{cases}
\end{align}
and work in the computational basis, where $\ket{1}, \ket{-1}$ are up and down along the $\hat\zv$-axis respectively. Though $E_j$ is a function of HamQAOA parameters, we omit its arguments for brevity in the following equations.
\begin{formula}[Finite-Degree Iteration for General HamQAOA]\label{alg:fin_iter}

Given any HamQAOA depth $p\in\mathbb{N}$, consider a 2-local LHP on a $(d+1)$-regular interaction graph $G$ of girth $> 2p+1$. The iterative formula takes the following inputs: (1) a Pauli interaction term $\sigma_L\sigma_R$ acting on one edge of G, (2) a collection of $4p$ real-valued HamQAOA parameters denoted by $\paramv = (\vect{\alpha}, \vect{\beta}, \vect{\gamma}, \vect{\delta})$, and (3) distributions $\cD_n$ and $\cD_m$ over $\bbS^2$ from which the HamQAOA initial product state $\ket{\mv} = \bigotimes_v \ket{\mv_v}$ and the $D$-driver rotation axes on each vertex $\nv = \{\nv_v\}_{v\in V}$ are sampled. The formula then computes the expected energy of the HamQAOA output state with respect to $\sigma_L\sigma_R$.

With these inputs, define
\begin{align}\label{fin_iter_f_func}
    f^{\sigma}_{\mv_v,\nv_v}(\vect{z}) = \braket{\mv_v}{z_1}\bra{z_1}E_1\ket{z_2}...\bra{z_{p-1}}&E_{p-1}\ket{z_p}     \bra{z_p}E_p\ket{z_{p+1}}\bra{z_{p+1}}\sigma\ket{z_{-(p+1)}} \nonumber \\
    &\bra{z_{-(p+1)}}E_{-p}\ket{z_{-p}}...\bra{z_{-2}}E_{-1}\ket{z_1}\braket{z_{-1}}{\mv_v}
\end{align}
for $\vect{z} \in \{\pm 1\}^{2p+2}$. Note that the $\nv_v$-dependence is implicitly included through the $D$-driver in each unitary $E_i$. Then we define
\begin{equation} 
\label{eq:sphere_avg}
    \bar{f}^\sigma(\zv) = \mathbb{E}_{\mv_v, \nv_v} [f^{\sigma}_{\mv_v, \nv_v}(\zv)], \qquad \text{where} \quad \ 
    \mv_v \sim_\iid \cD_m \quad \text{and} \quad \nv_v \sim_\iid \cD_n.
\end{equation}
We next define $H_d^{(k)}(\vect{z}): \{\pm 1\}^{2p+2} \rightarrow \mathbb{C}$ iteratively for  $k=0,1,\ldots,p$:
\begin{align}
H_d^{(0)}(\vect{z}) &= 1, \\
H_d^{(k)}(\vect{z}) &= \left( \sum_{\substack{\vect{x} \\ x^{(p+1)}=x^{-(p+1)}}} \exp( -i \vect{\mathcal{A}} \cdot (\vect{x} \vect{z})) H_d^{(k-1)}(\vect{x}) \bar{f}^I(\vect{x}) \right)^d
\end{align}
where the sum is over bitstrings $\vect{x}\in\{\pm1\}^{2p+2}$. The formula output is
\begin{equation}
    \mathbb{E}_{\mv,\nv} [\langle\vect{m}|\hat{O}(\sigma_L, \sigma_R)|\vect{m}\rangle] = \sum_{\vect{z}_L, \vect{z}_R} \exp( -i \vect{\mathcal{A}} \cdot (\vect{z}_L \vect{z}_R)) H_d^{(p)}(\vect{z}_L) H_d^{(p)}(\vect{z}_R) \bar{f}^{\sigma_L}(\vect{z}_L) \bar{f}^{\sigma_R}(\vect{z}_R).
\end{equation}
\end{formula}

If $\mv_v = \nv_v$, we use the notation $f^{\sigma}_{\nv_v}$ for the function defined in \eqref{fin_iter_f_func}.
One possible choice of inputs to Formula \ref{alg:fin_iter} might be to set $\mv_v = \nv_v$ and $\cD_n=\cD_m=\text{Unif}(\bbS^2)$. As discussed in Appendix~\ref{sec:gen_ansatze}, the expectation in equation \eqref{eq:sphere_avg} over this distribution can be computed by taking the mean of $2p+1$ values of $f_{\nv}^\sigma(\zv)$, where $\nv$ are taken from a spherical-$(2p+1)$ design.

Another possible choice is to consider the simplified HamQAOA (Algorithm \ref{HamQAOA}) and choose each $s_v\in \{\pm 1\}$ uniformly at random. In this case, one needs only define 
\begin{equation}
    \bar{f}^\sigma(\zv) = \frac{1}{2}\left(f^{\sigma}_1(\zv) + f^{\sigma}_{-1}(\zv)\right)
\end{equation}
since the vectors $\nv_v = s_v$ are restricted to the $\pm X$ axis. In this case, the expectation over $\mv, \nv$ would collapse to an expectation over only $\sv = \{s_v\}_{v\in V}$, where each $s_v \in \{\pm 1\}$ and the initial state is $\ket{\sv} = \bigotimes_{v\in V} \ket{s_v^X}$ where each single-qubit state is an $X$-basis eigenstate. This is because the initial state in Algorithm \ref{HamQAOA} is fixed upon choosing the $D$ driver. Later in Section \ref{sec:iter_numerics}, we compute Formula~\ref{alg:fin_iter} to assess the performance of the simplified HamQAOA for QMC at different values of depth $p$ and vertex degree $d$, and plot the results in Figure~\ref{fig:iter_scaling}.

\vspace{-1em}
\subsection{The Infinite-Degree Limit}
\label{sec:inf_d_iter}

From our numerical analysis of Formula \ref{alg:fin_iter} on QMC on high-girth regular graphs at various depth $p$ and degree $d$ (see Figure \ref{fig:iter_scaling} in Section~\ref{sec:iter_numerics}), we observe that the optimal function value at fixed $p$ scales as $1/\sqrt{d}$. Based on this observation, we wish to extend Formula $\ref{alg:fin_iter}$ to calculate the asymptotic energy scaling to leading order in $d$, allowing us to approximate the formula output well at large finite $d$.

We remark that the optimal HamQAOA parameters found from our numerical optimizations of Algorithm \ref{HamQAOA} take all $\beta$ values to be multiples of $\frac{\pi}{4}$ and $\gamma = 0$ for all tested $d>3$.
By restricting the parameter space to these values, the initial state to a product over $X$-eigenstates, and making a technical assumption to complete our proof (Assumption \ref{lem:fx_H=0}), we are able to simplify our iteration in the $d\to\infty$ limit. We find numerically that the optimal energy returned by this infinite-$d$ formula follows the scaling expected from our finite-$d$ formula as $d$ is increased (Figure \ref{fig:iter_scaling}). This simplified formula can be computed asymptotically in $\tilde{O}( 4^p)$ time and $O(p^2)$ memory, improving over the finite-$d$ case with quadratically better time and exponentially better memory usage. Under Assumption $\ref{lem:fx_H=0}$, this formula provides a rigorous upper bound on the ground energy density of Heisenberg spin glasses.

The infinite and finite-$d$ formulas require the same inputs: a 2-local Pauli interaction term $\sigma_L\sigma_R$ on neighboring sites of high-girth regular graph $G$, $4p$ parameters $\paramv$, and distributions $\cD_m, \cD_n$ over the input state and $D$-driver rotation axes. Note that the support of $\cD_m$ is not $\bbS^2$, but is constrained to the $\hat{\vect{x}}$ axis. For convenience, we define $\nu_{p,d}(\paramv, \sigma_L, \sigma_R)$ as follows: 
\begin{equation*}
    \nu_{p,d}(\paramv, \sigma_i, \sigma_j) = -\frac{\sqrt{d}}{2}\mathbb{E}_{\mv, \nv}\langle\paramv|\sigma_i\sigma_j|\paramv\rangle
\end{equation*}
where $\ket{\paramv}$ is the HQS defined in Algorithm \ref{GHamQAOA}.
This definition will be convenient for obtaining an expression free of $d$ when we go to the large-degree limit, as made apparent by the scaling of $\mathbb{E}\langle\paramv|X_iX_j + Y_iY_j + Z_iZ_j|\paramv\rangle$ as $1/\sqrt{d}$ in Figure \ref{fig:iter_scaling}.
 
Then 
\begin{equation}
    \nu_{p,d} = -\frac{\sqrt{d}}{2}\mathbb{E}\langle\paramv|\sigma_i\sigma_j|\paramv\rangle = -\frac{\sqrt{d}}{2}\sum_{\vect{z}_L, \vect{z}_R} \exp( -i \vect{\mathcal{A}} \cdot (\vect{z}_L \vect{z}_R)) H_d^{(p)}(\vect{z}_L) H_d^{(p)}(\vect{z}_R) \Bar{f}^{\sigma_i}(\vect{z}_L) \Bar{f}^{\sigma_j}(\vect{z}_R).
    \label{40}
\end{equation}
Furthermore, our numerical optimizations also indicate that the magnitude of optimal $\alpha$ scale as $\Theta\left(1/\sqrt{d}\right)$. 
Motivated by this observation, we define $\tilde{\alpha} \vcentcolon= {\alpha}{\sqrt{d}}$ and $\Tilde{\vect{\mathcal{A}}}\vcentcolon= \vect{\mathcal{A}}\sqrt{d}$. With these definitions, our objective function across an edge $ij$ in the infinite-degree limit becomes
\begin{equation}
    \lim_{d\rightarrow\infty} \left[\nu_{p,d}(\paramv, X_i, X_j) + \nu_{p,d}(\paramv, Y_i, Y_j) + \nu_{p,d}(\paramv, Z_i, Z_j)\right]
\end{equation} 
whose maximization corresponds to maximization of the QMC cost function. We now give a more efficiently computable expression to compute this $\nu_{p,d}$.

\begin{formula}[Infinite-Degree Iterative Formula for HamQAOA]\label{alg:d_inf_iter}
We define inputs $p, \paramv, \sigma_L,\sigma_R, \cD_m, \cD_n$ as in Formula \ref{alg:fin_iter}. Restrict all values of $\gamma$ to $0$ and $\beta$ to $k\pi/4$ for any integer $k$. Moreover, restrict the support of $\cD_m$ to the $\hat{\vect{x}}$ axis on $\bbS^2$. Now define matrices $G^{(m)} \in \mathbb{C}^{2p+2}\times\mathbb{C}^{2p+2}$ for $0 \leq m \leq p-1$ as follows: For $j, k \in \{1, 2, ... , p, p+1, -(p+1), -p, .... , -2, -1\}$, let
\begin{equation}
    G^{(p+1)}_{j, k} = \sum_{\zv} \bar{f}^I(\zv) z_j z_k
\end{equation}
where $\bar{f}^\sigma(\zv) = \bar{f}(\zv)h^\sigma(\zv)$, and
\begin{equation}
    G^{(m)}_{j,k} = \sum_{\zv}\bar{f}^I(\zv)z_j z_k\ \exp\left(-\frac{1}{2}\sum_{j',k'=-p}^p G^{(m-1)}_{j',k'}\tilde\alpha_{j'}\tilde\alpha_{k'}z_{j'}z_{k'}\right)
\end{equation}
for $1\leq m \leq p-1$.
Now define the vector $K^\sigma \in \mathbb{C}^{2p+2}$ such that
\begin{equation}
    K^{\sigma}_{i} = \sum_{\zv}\bar{f}^\sigma(\zv)z_i\exp\left(-\frac{1}{2}\sum_{j',k'=-p}^p G^{(p-1)}_{j',k'}\tilde\alpha_{j'}\tilde\alpha_{k'}z_{j'}z_{k'}\right)
\end{equation}
Then assuming a technical Assumption~\ref{lem:fx_H=0}
\begin{equation}
\label{eq:nu_p}
    \nu_{p}(\paramv, \sigma_L,\sigma_R) := \lim_{d\rightarrow\infty} \nu_{p,d}(\paramv, \sigma_L, \sigma_R) = \begin{cases}
        \frac{i}{2}\sum_{j=-p}^p \tilde\alpha_j (K^{\sigma}_o)(K^{\sigma}_j),\qquad &\sigma_L \textnormal{ and } \sigma_R \in {Y, Z} \\
        0, \qquad & \sigma_L \textnormal{ or } \sigma_R = X
    \end{cases}
\end{equation}
\end{formula}

We remark that Assumption \ref{lem:fx_H=0}, on which the correctness of the $d\to\infty$ formula relies, has been verified by us via exhaustive numerical search. Moreover, we find that the optimal energy found using the infinite-degree formula follows the $1/\sqrt{d}$ scaling from the finite-degree case and matches the finite-degree values when rescaled by the appropriate factor, as shown in Figure \ref{fig:iter_scaling}.
The improved efficiency of evaluating this formula makes analysis feasible at higher depth, and could potentially provide parameters that are good initial values for optimization for the more computationally intensive finite-degree formula and for real high-degree graphs. 

\section{Empirical Results}

\subsection{Numerical Evaluation of Iterative Formulae and Agreement on Specific Graphs}
\label{sec:iter_numerics}

We first evaluate the performance of HamQAOA on the Quantum MaxCut problem on high-girth regular graphs. To do this, we optimize the HamQAOA parameters of Formula \ref{alg:fin_iter} for Algorithm \ref{HamQAOA} with cost function $X_LX_R + Y_LY_R + Z_LZ_R$, sampling initial parameters uniformly at random from $[-\pi/2, \pi/2]^{4p}$ and each $\mv_v=\nv_v$ uniformly from $\pm \hat{\vect{x}}$.
As in the case of classical MaxCut (e.g. \cite{Leo_MaxCut}), we find that the optimal cost function value scales as $1/{\sqrt{d}}$.
This scaling also serves as a numerical consistency check on our infinite-$d$ formula, as demonstrated in Figure \ref{fig:iter_scaling}.

\begin{figure}[h!]
    \centering
    \includegraphics[width=\textwidth]{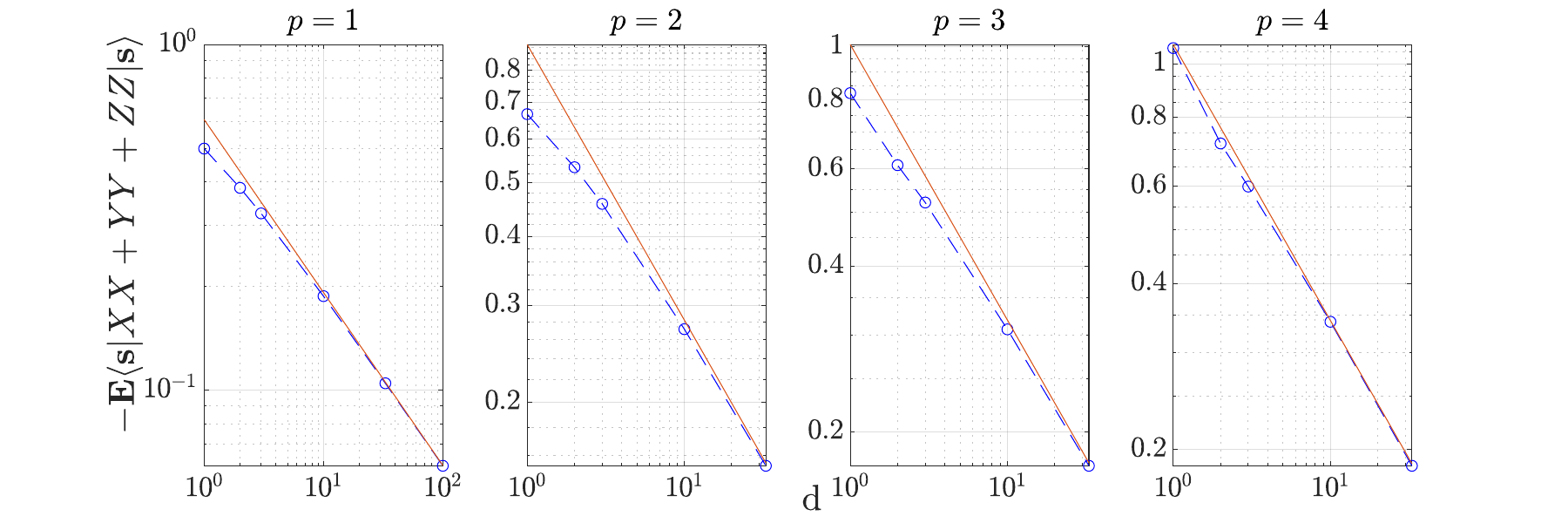}
    \caption{Comparison of the optimized finite-$d$ energies from Formula \ref{alg:fin_iter} (dash blue curve) to the optimized and rescaled infinite-$d$ formula value $2\nu_{p}/\sqrt{d}$ (solid red line). We observe that the red curve is an increasingly good estimate for the finite-$d$ energy as $d$ grows large, demonstrating that the energy from Formula \ref{alg:fin_iter} scales as $1/\sqrt{d}$ and that the infinite-$d$ formula provides a good estimate to energies at large $d$. 
    }
    \label{fig:iter_scaling}
\end{figure}

We also find that the optimal parameters given by the formula yield theoretical energies consistent with the HamQAOA ground energy approximations on specific instances of regular graphs, even with small size and girth. Specifically, we evaluate the HamQAOA with optimal parameters from the $d=2$ (i.e., vertex degree $3$) formula on an ensemble of 3-regular graphs. We do the same on ring graphs, where for each $p$ we choose a ring of size $2p+2$ and use the formula parameters at $d=1$. The energies obtained on these graphs, averaged over many choices of $D$ driver, are displayed in Figure \ref{fig:Theory-experiment-agreement}.
 
\begin{figure}[h!]
    \centering
    \includegraphics[width=0.5\textwidth]{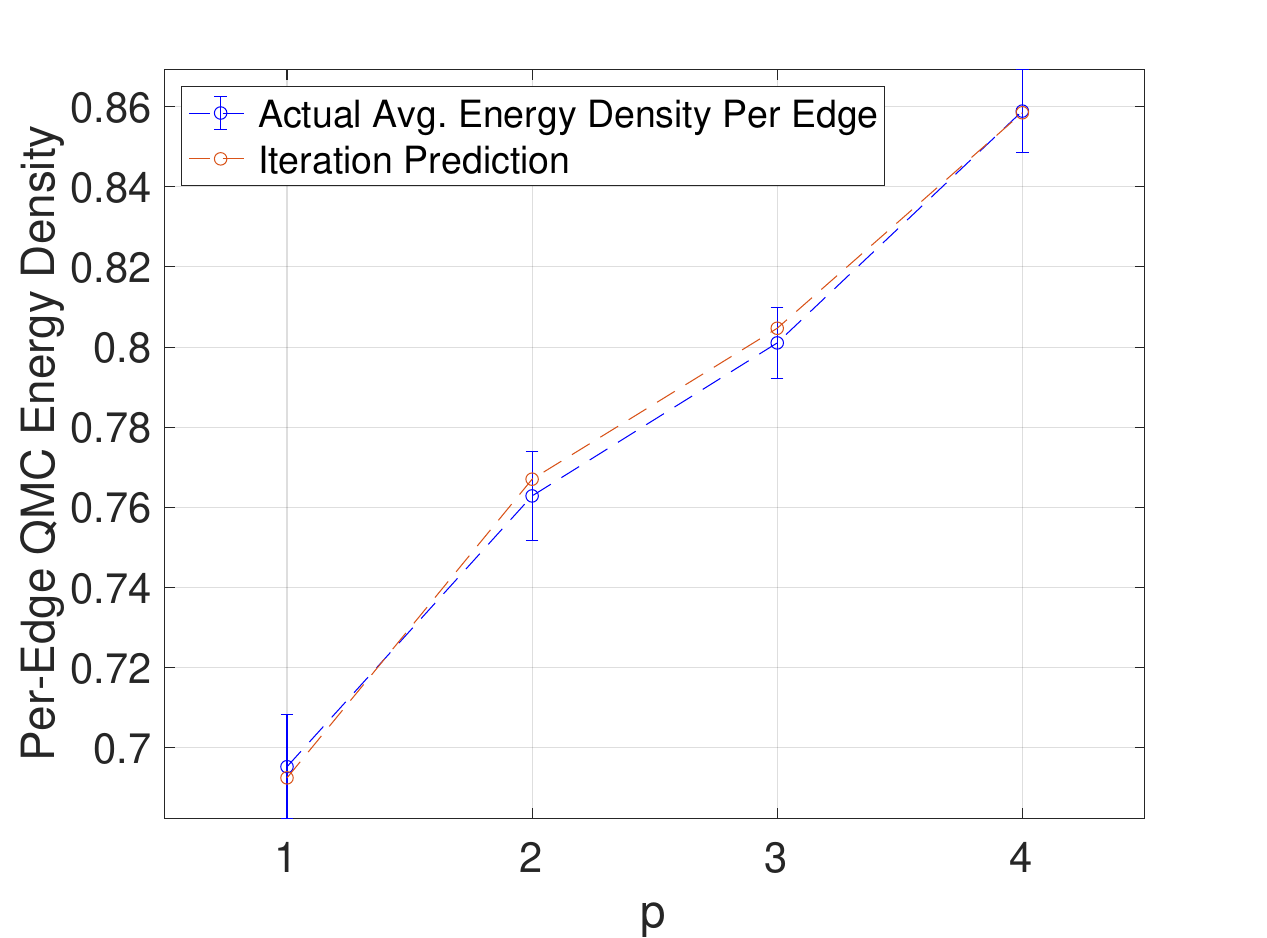}\hfill
    \includegraphics[width=0.5\textwidth]{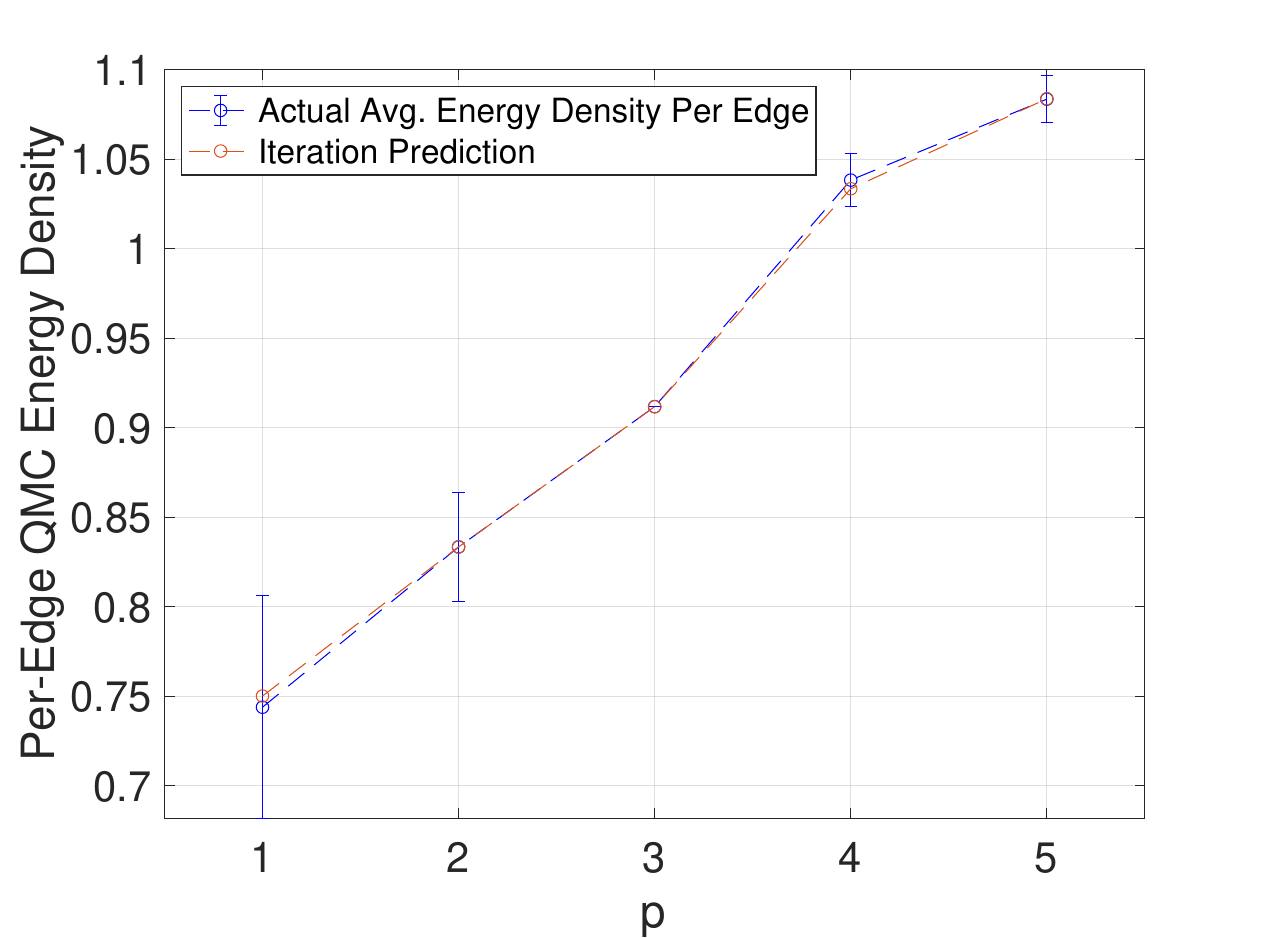}
    \caption{Agreement of theoretical average-case prediction with empirical graphs. Left: 3-regular graphs of size 14. At each $p \leq 4$ we sample $20$ strings from $\{\pm 1\}^{14}$ and, using optimal iterative formula parameters at $d=2$, we compute the average performance over $8$ graphs. Right: Ring graphs of size $2p+2$. At each $p\leq 5$ we repeat the above process, only using one ring graph of size $2p+2$ such that the high-girth assumption of the formula is satisfied, and with parameters at $d=1$. For these graphs, we sample $100$ strings from $\{\pm 1\}^{2p+2}$. In both cases, we see that the empirical data agrees with the formula prediction up to statistical uncertainty. 
    }
    \label{fig:Theory-experiment-agreement}
\end{figure}

We remark that the energy values we obtained in Formula~\ref{alg:fin_iter} can serve as rigorous lower bounds on the true maximum QMC value, and correspondingly, upper bounds on the Heisenberg ground energy.
In particular, we note that the solution for the infinite-size ring is known to be about $1.386$, whereas our $p=5$ value in Figure \ref{fig:Theory-experiment-agreement} yields $1.08$, which is a loose lower bound.
This looseness is because the HamQAOA's performance is strongly dependent on the choice of the rotation axes $\{\nv_v\}$ or $\{s_v\}$ that parameterize the $D$ driver. Note that for the simplified HamQAOA, the classical bitstring $\sv$ that defines the $D$ driver can be interpreted as a partition on the vertices of the interaction graph, and defines a cut on the edges of the graph. As we show in Figure \ref{fig:D_Driver_Perf} in Appendix \ref{sec:D_Driver_choice}, the simplified HamQAOA performs significantly better when this partition defines a cut as close as possible to the graph's true MaxCut. Thus, we expect that the iterative formula lower bounds will be tighter for graphs with large vertex degree $d$, where a cut chosen uniformly at random will cut half the edges in expectation, a fraction that becomes asymptotically close to the MaxCut value on random regular graphs~\cite{Dembo_2017} as $d$ goes to infinity. We also provide numerical evidence that the average-over-$\sv$ driver performance improves with degree $d$ in Figure \ref{fig:rand_vs_mc_d}. Therefore, Formula~\ref{alg:fin_iter} provides a new strategy to bound ground energies of LHPs, particularly those with large vertex degrees.

Given that these parameters are optimal for regular graphs when averaging uniformly over choice of $\sv$, it is interesting to study their performance as parameter initializations in the best case, when the choice of $\sv$ coincides with the classical MaxCut. Empirically, this works well for nearly-regular graphs, even of low girth, which is discussed in Appendix \ref{sec:iter_strategy}.
\vspace{-2em}
\subsection{Algorithm Performance on Specific Graph Ensembles}\label{sec:numerics}
In this section we optimize the HamQAOA for the QMC problem on ensembles of three families of graphs, and compare the results against some of the best previous algorithms for QMC. For benchmarking, we consider the SDP algorithm from \cite{King_2023}, which at the time that this work was underway achieved the best worst-case guarantee for the QMC problem, and the variational quantum algorithm by Anshu, Gosset and Morenz \cite{Anshu_2021}, which we refer to as the AGM algorithm.
As detailed in Appendix \ref{sec:D_Driver_choice}, the HamQAOA performs best when the choice of $D$-driver coincides with the classical MaxCut of the graph, with performance positively correlated with the cut fraction. The AGM algorithm also depends on a classical assignment of $\pm 1$ to each vertex, and its performance is similarly correlated with the classical cut fraction. For an appropriate comparison, we give both the HamQAOA and AGM algorithms access to the classical MaxCut of each graph, such that the following numerical results represent the best achievable performance of each algorithm. 

For all families of graphs we consider, we find that the HamQAOA at only $p=2$ layers substantially outperforms both the SDP and AGM algorithms. 
Note that both the SDP and AGM algorithms employ single-layer variational circuits, similar to the $p=1$ HamQAOA; the HamQAOA allows us to go beyond the shallowest depths and to achieve monotonically improving performance at each layer, with asymptotic convergence guarantees.

We then focus specifically on 1-dimensional chains with periodic boundary conditions, i.e. ring graphs. The ground states of these models are exactly solvable via the Bethe ansatz \cite{1931.Bethe.ZP.71}, and many works have presented quantum circuit constructions to prepare these Bethe states \cite{Van_Dyke_2021, Li_2022, Sopena_2022, raveh2024deterministicbethestatepreparation}, making rings an ideal testbed for further analysis of the HamQAOA's capabilities. We present numerical evidence that the HamQAOA may be able to prepare Bethe states of AFM Heisenberg rings of length $N$ in $p=O(N)$, which would constitute an exponential improvement in circuit depth over the previous best algorithms.

\vspace{-1em}
\subsubsection{Performance on Erdos-Renyi and Random 3-Regular Graphs}\label{sec:ER_3Reg_Perf}

\begin{figure}[h!]
    \centering
    \includegraphics[width=0.49\textwidth]{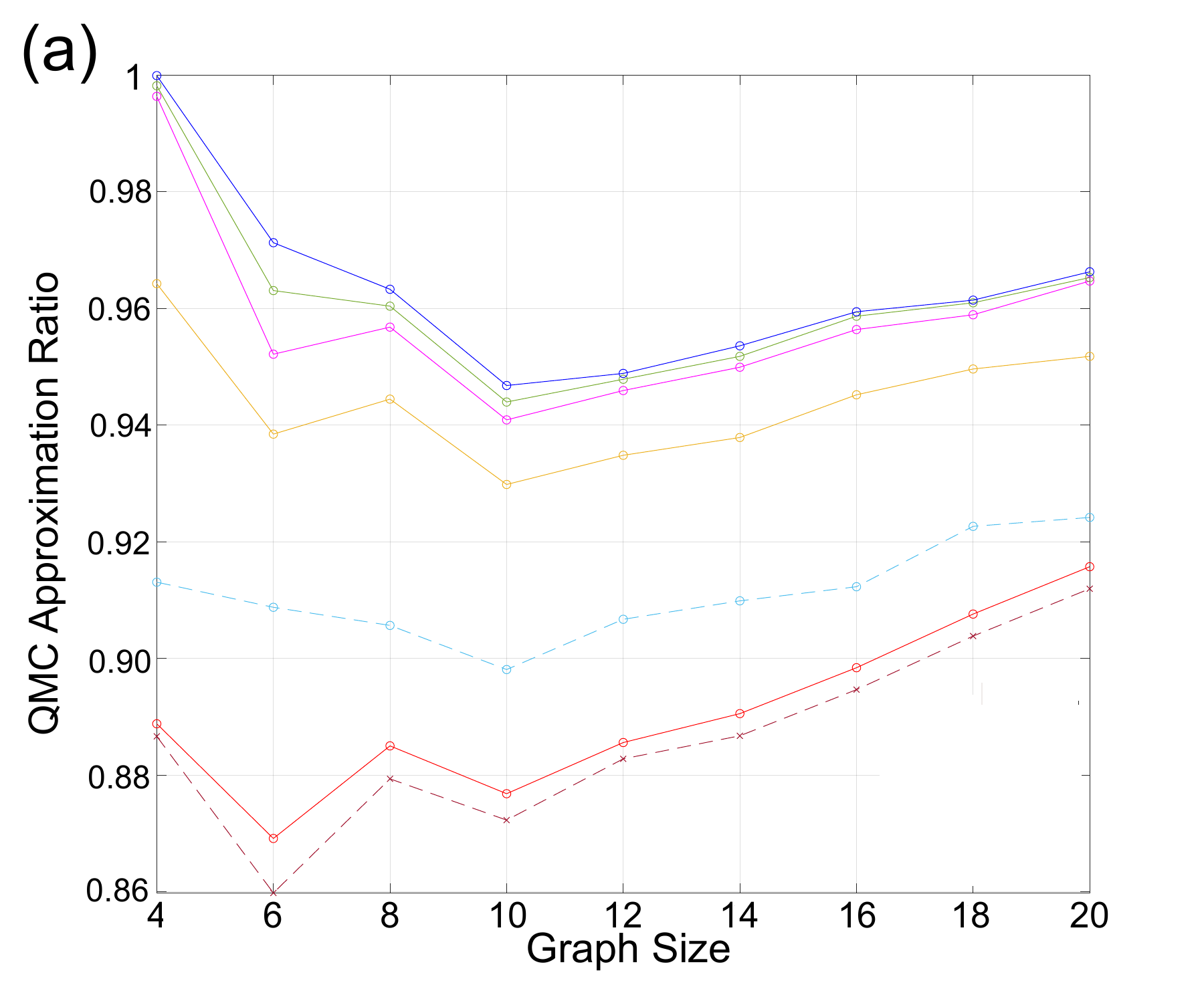}\hfill
    \includegraphics[width=0.49\textwidth]{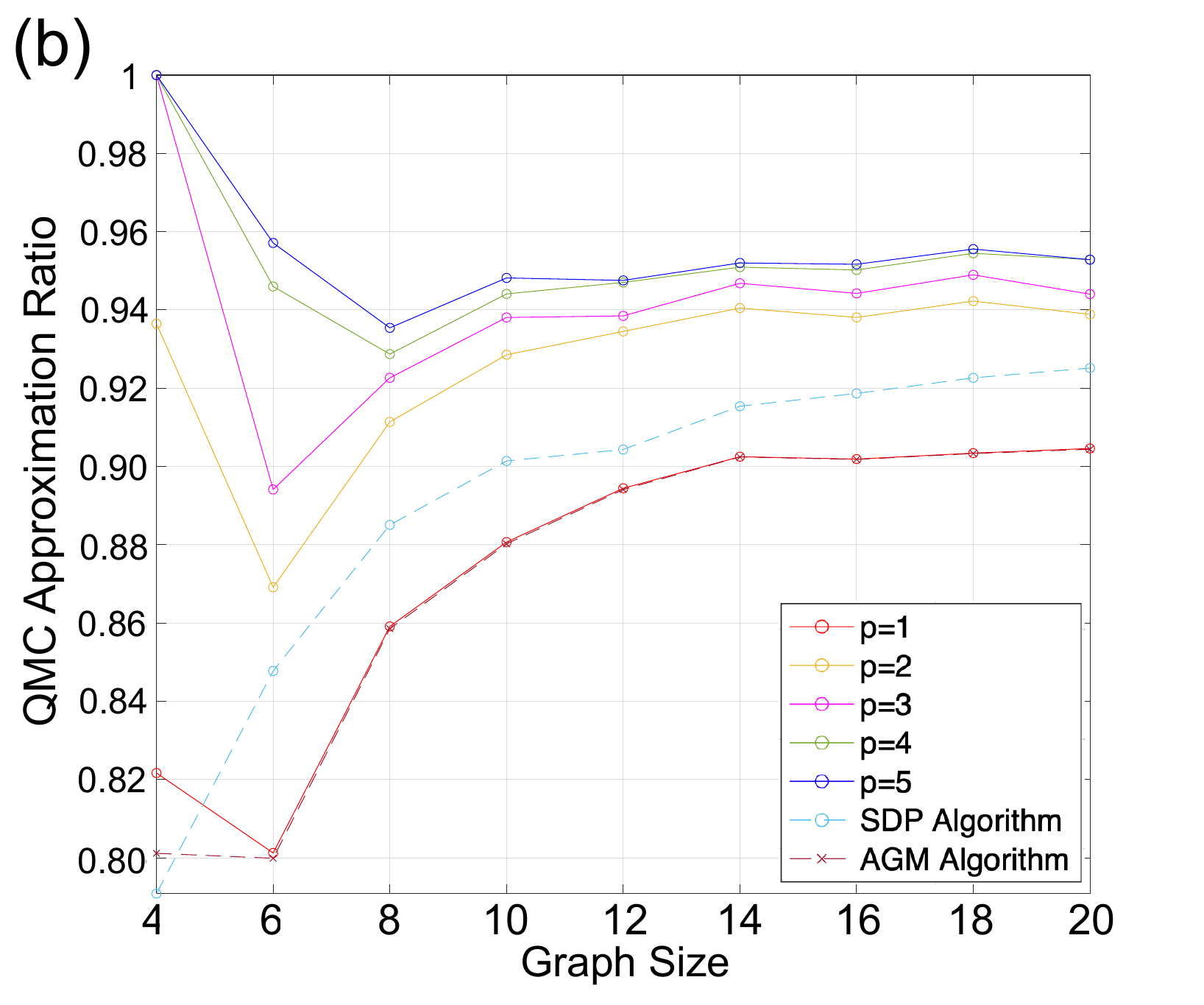}
    \caption{Performance of the HamQAOA versus the AGM and SDP algorithms on (a) Erdos-Renyi graphs with edge probability 0.5 and (b) random 3-regular graphs. As expected from Lemma \ref{lemma:contains_prev}, the HamQAOA does no worse than the AGM algorithm, and for $p > 2$ the HamQAOA also significantly outperforms the SDP algorithm.}
    \label{fig:er_3reg_best}
\end{figure}
First, we benchmark the HamQAOA on random graph instances. For each system size in Figure \ref{fig:er_3reg_best}, we generate 8 Erdos-Renyi random graphs with edge probability 0.5 and 8 random 3-regular graphs, allowing us to assess HamQAOA performance on both dense and sparse random graphs.
We remark that the AGM algorithm is known to achieve an energy on 3-regular graphs that provably surpasses any product state, and our $p>1$ HamQAOA strictly outperforms the AGM algorithm. Moreover, since there is no known method to compute the total angular momentum of the ground states on these generic non-bipartite graphs, variational schemes without symmetry-breaking drivers will be intrinsically limited.

\subsubsection{Performance on Rings, and Bethe State Conjecture}
\label{sec:rings}

For ring graphs, we find that the performance of the SDP algorithm, the AGM algorithm, and the $p=1$ HamQAOA all coincide at every system size, and that the HamQAOA at higher depths significantly outperform the other algorithms. While the HamQAOA's outperformance is expected from Lemma \ref{lemma:contains_prev}, Figure \ref{fig:Ring_Best_p[1-5]} (left) shows that for rings, the two benchmark algorithms (which we prove to be equivalent on edge-transitive graphs) exactly match the $p=1$ HamQAOA.

\begin{figure}[h!]
    \centering
    \includegraphics[width=0.49\textwidth]{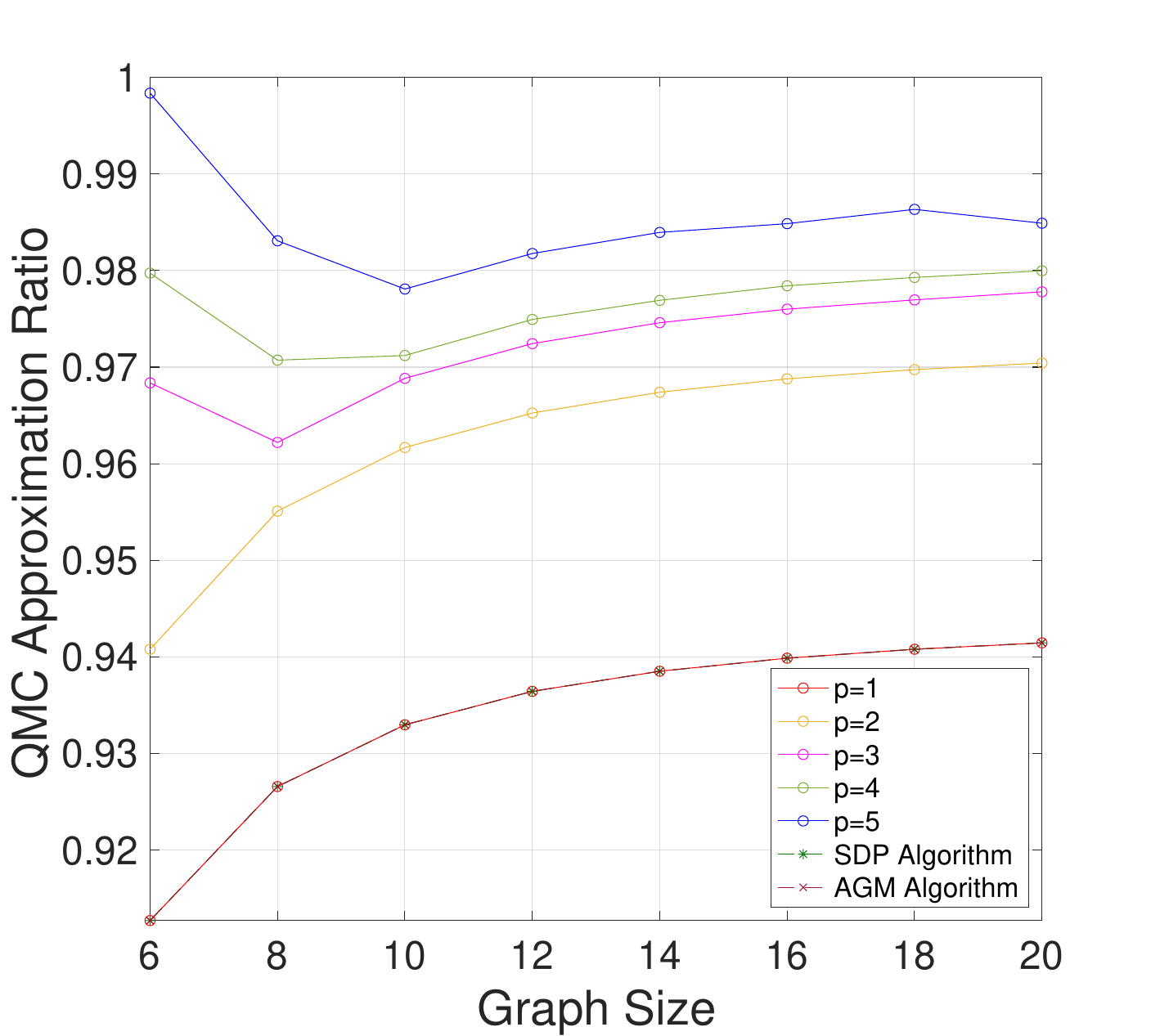} \includegraphics[width=0.49\textwidth]{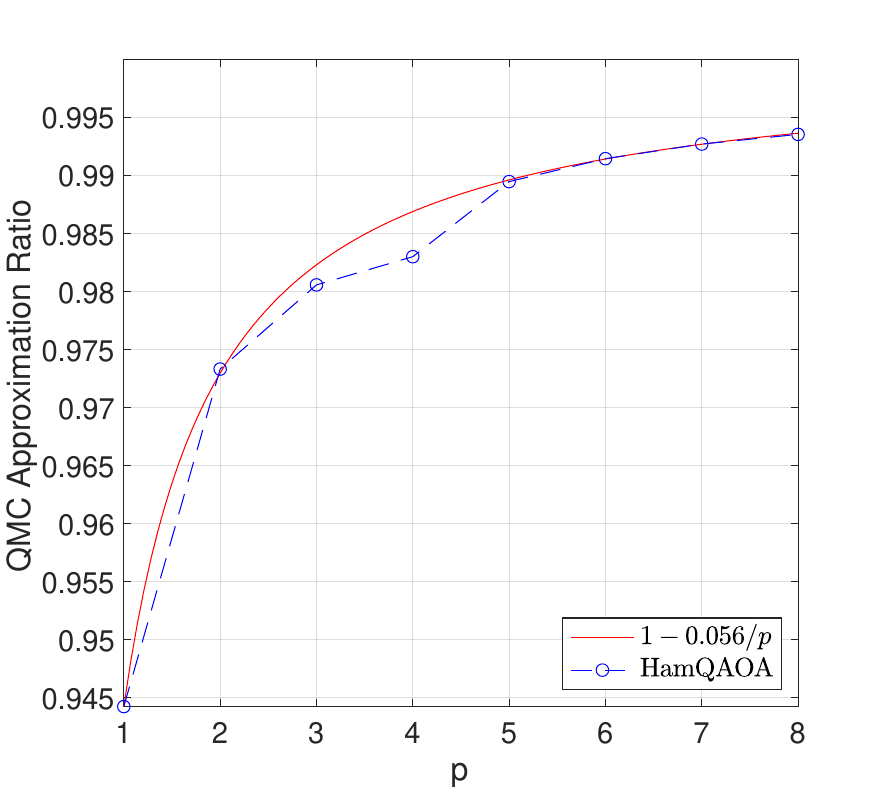}
    \caption{\textit{Left}: Performance of the HamQAOA at $1\le p\le 5$ and two benchmark algorithms as a function of system size on even-sized ring graphs. The QMC approximation ratio is the output state energy $\langle{H_{\rm QMC}}\rangle$ divided by the maximum eigenvalue of the $H_{\rm QMC}$. Note that the AGM, SDP and $p=1$ HamQAOA curves lie exactly on top of one another. 
    \textit{Right}: Scaling of HamQAOA approximation ratio on infinite-size chain as a function of depth $p$. The exact ground state energy density is known from the Bethe ansatz. The approximation ratio achieved scales as $1-1/p$ with fluctations at $p=3, 4$; here we fit the curve $I - C/p^j$ for parameters $I, C, j$ and obtain $I=1, j=1$. The kinks at $p=3, 4$ also appear in the left plot, where the spacing between $p=4$ and $5$ is larger than that between $3$ and $4$.}
\label{fig:Ring_Best_p[1-5]}
\end{figure}

Recall that for a HamQAOA circuit of depth $p$, the light cone of each qubit contains only vertices $\leq p$ edges away.
On ring graphs, the energy contribution from one edge only considers $2p+2$ neighboring edges, and the HamQAOA does not see the whole ring graph whenever $p\leq (N-2)/2$.
The energy density of states prepared in this regime is therefore identical to that prepared on an infinite-length chain.
In this regime, we study the approximation ratio achieved by the HamQAOA relative to the Bethe ansatz ground energy, and find that it scales roughly as $1-O(1/p)$ in Figure \ref{fig:Ring_Best_p[1-5]} (right). 
Conversely, when $p \gg N/2$ (e.g., $p=2N$), the HamQAOA has depth far greater than the minimum necessary to see the whole graph.
At this linear depth, we find it easy to prepare states that have very high overlap with the true ground state of the system which remains constant as system size increases. This is shown in Figure \ref{fig:ring_GSpop}.

\begin{figure}[h!]
    \centering
\includegraphics[width=0.9\textwidth]{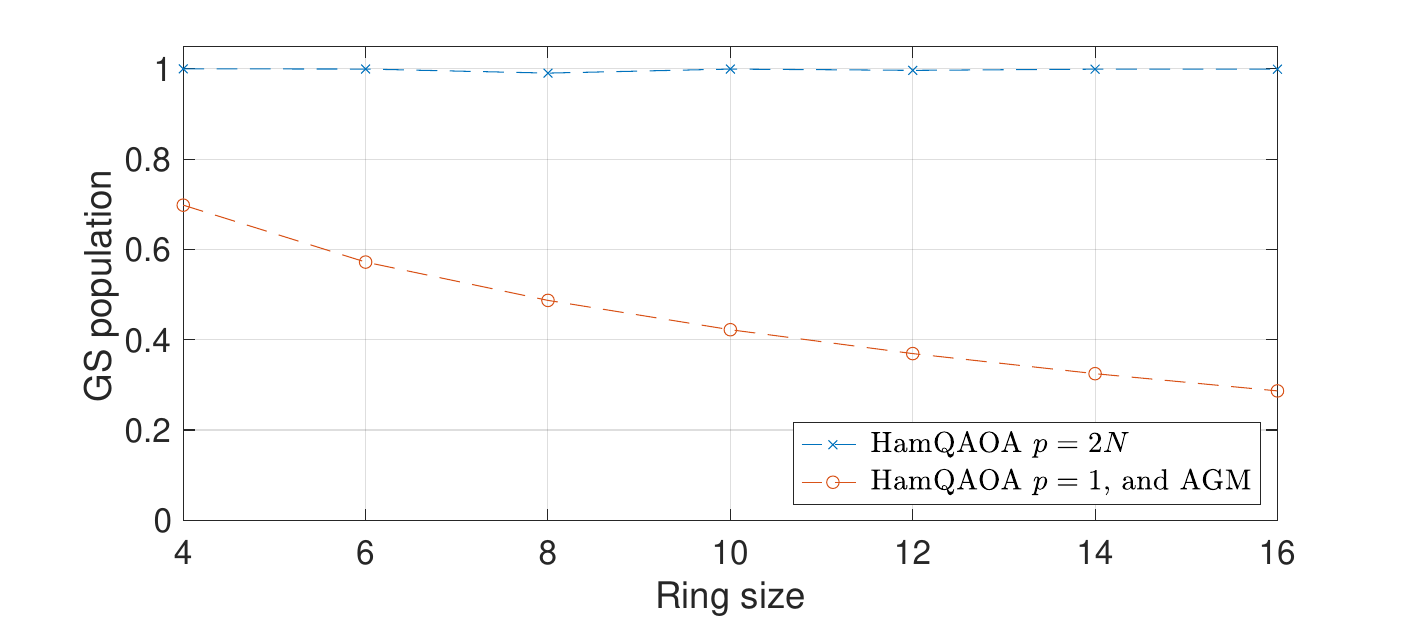}
    \caption{Fidelity of states prepared by the HamQAOA of depth $p=2N$, depth $p=1$, and the AGM algorithm with the maximum-energy eigenstate of the QMC Hamiltonian, denoted as ground state (GS) population. The AGM and $p=1$ HamQAOA yield the same curve, so they have been displayed as a single curve, although it is not obvious that the algorithms are exactly equivalent in this subcase. Whereas the output states from the AGM algorithm have GS population that decays smoothly with ring size, the HamQAOA is capable of preparing state with essentially perfect GS population.
    } 
    \label{fig:ring_GSpop}
\end{figure}

Since the AGM algorithm coincides with the $p=1$ HamQAOA in Figure \ref{fig:Ring_Best_p[1-5]}, one can also interpret Figure \ref{fig:ring_GSpop} as depicting a phase transition in the HamQAOA's computational power to prepare output states. Particularly, there exists some depth $p$ as function of $N$, between any fixed constant and $O(N)$ at which the scaling of ground state fidelity changes from exponetially decaying with $N$ to a constant. At any $p < N/2$, the HamQAOA does not see the full graph, and thus produces locally imperfect approximations. These errors in the fidelity multiply together, and we thus expect exponentially decaying fidelity in the system size until $p=N/2$ and the whole graph is seen. Given the constant-order fidelity observed at $p=2N$ and the required exponentially-decaying fidelity scaling at $p<N/2$, we expect that this transition occurs at $\Theta(N)$.

We also find HamQAOA parameters at $p=4$ which prepare the exact ground state of the $N=4$ ring, and similarly $p=7$ parameters for preparing the exact ground state of the $N=6$ ring. These parameters are are given in Table~\ref{tab:exactparam} Appendix \ref{sec:optparams}. This result is surprising as the QAOA framework is intended to produce exact solutions only in the large-depth limit, and it is an open theoretical question as to why this is possible for systems with particular structure such as the ring. Combining this observation with the results of Figure $\ref{fig:ring_GSpop}$ and the $1/p$ scaling of the approximation error in Figure~\ref{fig:Ring_Best_p[1-5]}, we state our main conjecture:

\begin{conjecture}
\label{conj1}
    For an AFM Heisenberg XXX ring of size $N$, there is a HamQAOA circuit of depth $O(N)$ which prepares its exact ground state.
\end{conjecture}
Proving this conjecture would constitute an exponential improvement in deterministic Bethe state preparation via the HamQAOA, compared to the best existing unitary protocol from \cite{raveh2024deterministicbethestatepreparation}. 

Our optimism for the HamQAOA's performance stems from the demonstrated success of variational algorithms for this task. The Variational Quantum-Classical Simulation (VQCS) algorithm from \cite{vqcs} also achieves $\Omega(1)$ fidelity using a linear-depth quantum circuit, though it does not find exact ground states up to numerical precision at any finite $p$. The similar HVA-based ansatz from \cite{kagome} achieves fidelity greater than $1-10^{-5}$ on rings up to size $16$ with $p$ again scaling linearly in ring size $N$; however, each layer has a number of tunable parameters linear in $N$ to allow for symmetry breaking. Our $D$ driver serves this purpose, but with only a single parameter per layer. Finally, as discussed in Section \ref{sec:otherlhp} below, Conjecture \ref{conj1} likely extend to LHPs beyond the AFM Heisenberg XXX model.

\subsection{Performance on Other LHPs}\label{sec:otherlhp}
While Algorithm \ref{HamQAOA} is tailored for the Heisenberg XXX model (i.e. QMC), we find it highly effective for other physically important 2-local LHPs. Here, we benchmark its performance on the quantum XY model and Heisenberg XXZ model in an external field.

\begin{figure}[h!]
    \centering
    \includegraphics[width=0.75\textwidth]{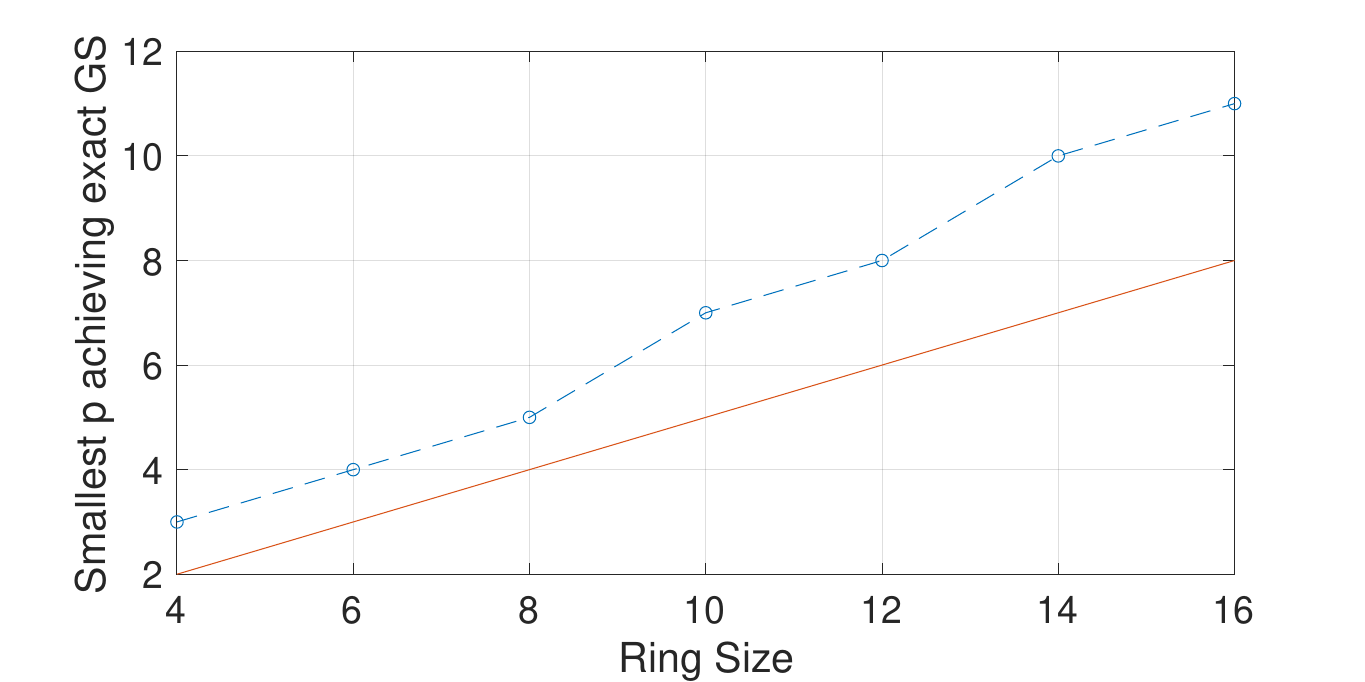}
    \caption{Smallest depth $p$ at which our numerical experiments of the HamQAOA find the exact ground state (up to floating-point precision in the fidelity) of the 1d XY model for different system size. The red line is $p=N/2$, the smallest $p$ at which the HamQAOA sees the whole system, for visual comparison.
    }
    \label{fig:xy}
\end{figure}

The XY model Hamiltonian is
\begin{equation}
    H_{XY} = \sum_{i\sim j} (X_i X_j + Y_i Y_j)
\end{equation}
One of the classic problems in condensed matter physics, the 1-dimensional XY model is exactly solvable by employing the Jordan-Wigner transformation to map the quantum spins to a free theory of spinless fermions \cite{Jordan:1928wi}. Perhaps due to this innate tractability, we find in Figure \ref{fig:xy} that the HamQAOA can exactly prepare the ground state of the 1d XY ring with $p$ less than system size.

\begin{figure}[h!]
    \centering
    \includegraphics[width=0.75\textwidth]{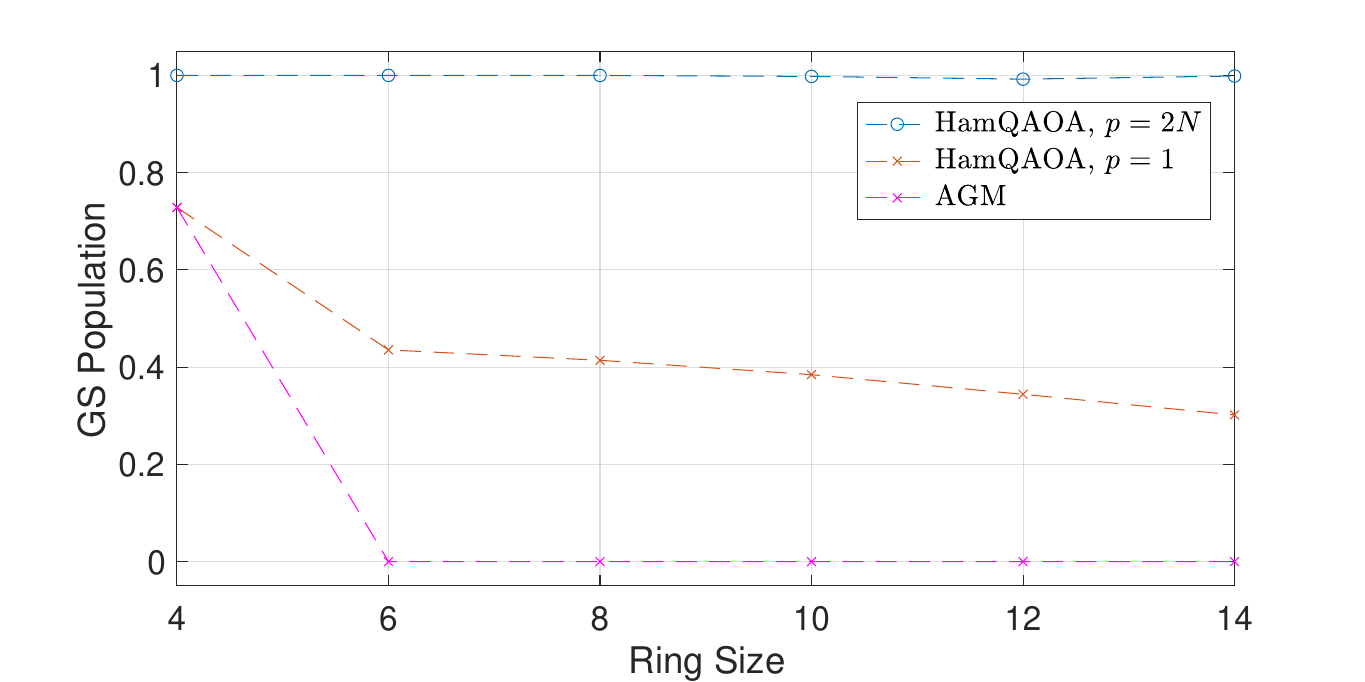}
    \caption{Fidelity of states prepared by HamQAOA of depth $p=2N$ versus AGM with ground state of XXZ model Hamiltonian \eqref{eq:H_XXZ} with $\Delta=1/2, h= 1/2$. The AGM struggles more than for QMC and the $p=1$ HamQAOA behaves less smoothly given the problem anisotropy, but the $p=2N$ HamQAOA again achieves near unity population in the ground state (GS),  similar to Figure \ref{fig:ring_GSpop}.}
    \label{fig:xxz}
\end{figure}

Next, the XXZ model in an external field is a generalization of the XXX model, allowing for anistropy in the 2-local Z term and a single-qubit field-coupling term:
\begin{equation}
\label{eq:H_XXZ}
    H_{XXZ} = \sum_{i\sim j} \big(X_iX_j + Y_iY_j + \Delta Z_i Z_j\big) + h\sum_i Z_i
\end{equation}
As seen in the phase diagram in Fig 1 in \cite{Rakov_2016}, this model permits an XY-like Luttinger-liquid phase, in which the long-range magnetic order observed in the ferromagnetic and AFM phases is lost \cite{Sachdev_2011, tomonaga}. Since QMC already describes an AFM model, we consider here the Luttinger-liquid phase with $\Delta = h = 0.5$. In Figure \ref{fig:xxz} we recreate Figure \ref{fig:ring_GSpop} for the XXZ case.
We are able to prepare exact XXZ ground states of $N=4$ ring at $p=4$, and $N=6$ ring at $p=8$.
This suggests that Conjecture \ref{conj1} extends to more general LHPs with richer ground state properties.

\section{Discussion}

We have proposed the HamQAOA as a simple, flexible, and practical quantum algorithm for LHPs with rigorous performance guarantees.
Our analysis of the HamQAOA's performance shows that it empirically beats the current best approaches for Quantum MaxCut, and provably outperforms the AGM and SDP algorithms on all graphs and on edge-transitive graphs respectively. Furthermore, we show that HamQAOA prepares states with near-unity overlap with the true ground state of the 1-dimensional AFM Heisenberg XXX model and other 2-local LHPs in depth only linear in system size $N$, and exhibits a ``phase transition'' at a depth between $N/2$ and $2N$ where the ground state fidelity goes from exponentially decaying to remaining a constant close to 1. Since we are also able to prepare exact ground states of small 1d systems with low-depth HamQAOA circuits and achieve an approximation ratio on an infinite-size chain that approaches 1 as $1/p$, we conjecture that the HamQAOA can prepare exact ground states of 1d 2-local LHPs with a circuit depth linear in the system size, a potential exponential improvement in Bethe state preparation.

We also provide an iterative formula to compute the expected performance of our HamQAOA when averaged over choice of the symmetry-breaking driver on high-girth regular graphs. This formula also applies to large random regular graphs, which have locally tree-like neighborhoods. Our formula gives a rigorous upper bound on the ground energy density of Heisenberg spin glass models in the thermodynamic limit, and this bound improves with depth.
We further simplify this formula in the infinite-degree limit based on a technical assumption. In the setting of classical MaxCut, a similar infinite-degree formula has yielded a connection between solving MaxCut on high-girth regular graphs and preparing the ground state of the Sherrington-Kirkpatrick (SK) spin glass model on the complete graph, and fitting the infinite-degree energy curve as a function of $p$ provides a very precise estimate of the Parisi value \cite{leo_talk}. For the HamQAOA on Heisenberg spin glasses, we optimize the energy to depth $p=8$, and provide their values in Appendix \ref{sec:nu_p_values}; however, our data so far has not yielded a good fit by the same family of curves. 
Understanding the connection between the high-girth and complete graph cases and studying the scaling of this formula to obtain a conjectured ``quantum Parisi value'' for the ground energy of the quantum Heisenberg model are natural and important next steps stemming from this work.
   
Our study of the HamQAOA framework for LHPs suggests several other avenues for both empirical and theoretical work. First, while our choice of drivers are tailored for 2-local LHPs, the generalization to $k$-local LHPs is immediate: one needs only replace the entangling driver $A$ with a sum over $k$-local Pauli-$Z$-products matching the geometry of the target LHP.
As approximation algorithms outputting entangled states have not been widely applied to practical LHPs that involve finding ground states of relevant many-body systems, exploring new HamQAOA ansatze is an interesting avenue for future work. One can also incorporate further improvements on the base HamQAOA, such as choosing the ansatz adaptively \cite{PhysRevResearch.4.033029} or using a clever warm start such as in \cite{Egger_2021,sridhar2023adaptqaoaclassicallyinspiredinitial}.

Moreover, while we have presented a thorough empirical study of the HamQAOA's optimal performance, we do not fully understand how the algorithm's variational landscape behaves. Although our numerical experiments suggest that it is easy to find useful local minima at large depth (i.e., without encountering barren plateau issues) with random initialization, it is unclear whether there is a depth regime where there are traps in the landscape that make it difficult to find minima within a constant gap of the global minimum. Prior works on the original QAOA have also found heuristic patterns in the optimal sets of parameters, such as a smooth increase in the $\beta$ values over layers while $\gamma$ values decrease smoothly, simulating an quasi-adiabatic evolution \cite{leo_qaoa_performance}; in our case, we have not observed such a trend within the depths that we could numerically access. Furthermore, our Greedy Iterative (GI) strategy is drawn from \cite{Sack_2023}, which has shown rigorous results for a more general parameter-guessing strategy with the QAOA; it may be fruitful to look for similar guarantees for the HamQAOA.
A rigorous understanding of our algorithm's optimization properties would be valuable for future application.

Although our HamQAOA is geared towards practical application, it can also be used to prove worst case guarantees to improve understanding of hardness-of-approximation of QMC and other LHP problems.
Our Lemma~\ref{lemma:contains_prev} shows that the HamQAOA at $p=1$ already inherits the worst-case guarantee of the AGM algorithm~\cite{agm_beyond}, and our empirical results indicate that it has strictly better performance at $p=2$.
It would thus be interesting to derive a worst-case guarantee of HamQAOA at $p>1$, similar to what has been done for the vanilla QAOA in \cite{WurtzLove2021}.
Our iterative formula also gives a rigorous upper bound on the ground energy density of Heisenberg spin glass models, and with modifications to the problem cost function can yield nontrivial bounds for various other many-body models. Thus, it is valuable to prove the correctness of our formula for the infinite-degree limit, which would allow for optimization at higher $p$ to yield better rigorous bounds, in addition to guiding the search for optimal parameters for finite-degree cases at larger depths. 

Finally, our main conjecture suggests potential for an exponential improvement in 1d Bethe state preparation. These states have been studied for decades in physics, but there do not exist efficient deterministic methods to prepare them. Moreover, these Bethe states can encode solutions to classically hard computational problems~\cite{crichigno2024quantumspinchainssymmetric}. Understanding the power of the HamQAOA or other variational frameworks for preparing eigenstates of quantum spin chains, and potentially more general Bethe eigenstates, is thus a promising direction in which to search for quantum advantage in both physics and computer science that arises from this work.

\begin{acknowledgments}
We thank Kunal Marwaha and John Preskill for helpful discussions.
IK was supported by the Caltech Summer Undergraduate Research Fellowship during part of this project.  RK was supported by the Institute for Quantum Information and Matter, an NSF Physics Frontiers Center (NSF Grant PHY-1125565). LZ acknowledges funding from the Walter Burke Institute for Theoretical Physics at Caltech. The authors declare no competing interests.
\end{acknowledgments}

\newpage
\let\oldaddcontentsline\addcontentsline
\bibliography{refs}
\let\addcontentsline\oldaddcontentsline

\appendix
\appendixpage
\section{Algorithm Heuristics and Data for HamQAOA}
\subsection{Choice of driver $D$}
\label{sec:D_Driver_choice}
In the simplified HamQAOA for QMC (Algorithm \ref{HamQAOA}), the $D$ driver is
\begin{equation}
    D = \sum_v s_v X_v
\end{equation}
where each $s_v\in \{\pm 1\}$. We define the string $\vect{s}$, a collection of all $s_v$. Note that $\vect{s}$ can be interpreted a bipartition on the vertices of the graph by grouping all vertices assigned the same $s_v$ together. Therefore, $\sv$ defines a classical cut of the graph given by the edges spanning this bipartition. Different choices of this cut change the HamQAOA's variational landscape, so it is important to understand the connection between the choice of $\vect{s}$ and the ease of finding good local minima. Here we supply numerical evidence that the choice of $\sv$ should coincide as closely as possible with the classical MaxCut of the interaction graph. 

\begin{figure}[h!]
    \centering
    \includegraphics[width=\textwidth]{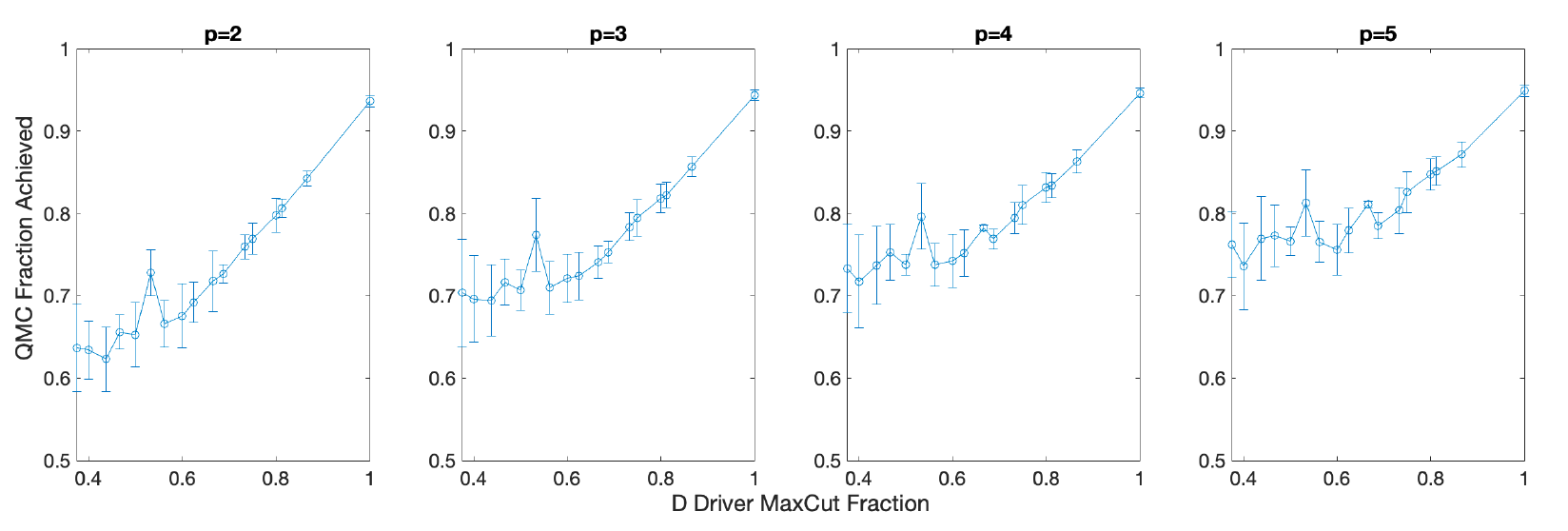}
    \caption{Best QMC approximation ratio achieved by the simplified HamQAOA of depth 2-5 as a function of choice of the classical MaxCut approximation ratio achieved by D-driver signs. To compute each data point corresponding to a choice of $\sv$, we first sample 20 sets of $4p$ parameters, each chosen independently and uniformly from $(-\pi/2, \pi/2]$, and optimize for each of 8 Erdos-Renyi graphs of size 12 with edge probability 0.5. We then choose the best-performing point for each graph, computing the mean and variance over 8 graphs. This is done for each p, over many choices of D-driver signs sampled uniformly from $\{\pm 1\}^{12}$, but always including at least one exact classical MaxCut. The upshot is a clear correlation between the classical MaxCut fraction of the HamQAOA D-driver and algorithm performance, and particularly the ease of finding good local minima in the parameter space when using relatively few parameter initializations.}
    \label{fig:D_Driver_Perf}
\end{figure}
Figure \ref{fig:D_Driver_Perf} demonstrates that choices of $\vect{s}$ with larger approximation ratios to the classical MaxCut of the given graph yield better QMC approximations. 

This also gives us intuition for when a random choice of $\vect{s}$, as opposed to the exact MaxCut or an approximation, is acceptable. For graphs of small vertex degrees, the expected cut fraction of a random cut is small, causing a significant drop in performance when using a random $\vect{s}$ vs. the exact MaxCut. However, in the large-size limit of a dense random graph, random cuts are asymptotically close to the MaxCut. In this case, we find that using a random cut does not perform much worse than using the true MaxCut, as shown in Figure \ref{fig:rand_vs_mc_d}.
\begin{figure}[h!]
    \centering
    \includegraphics[width=0.5\textwidth]{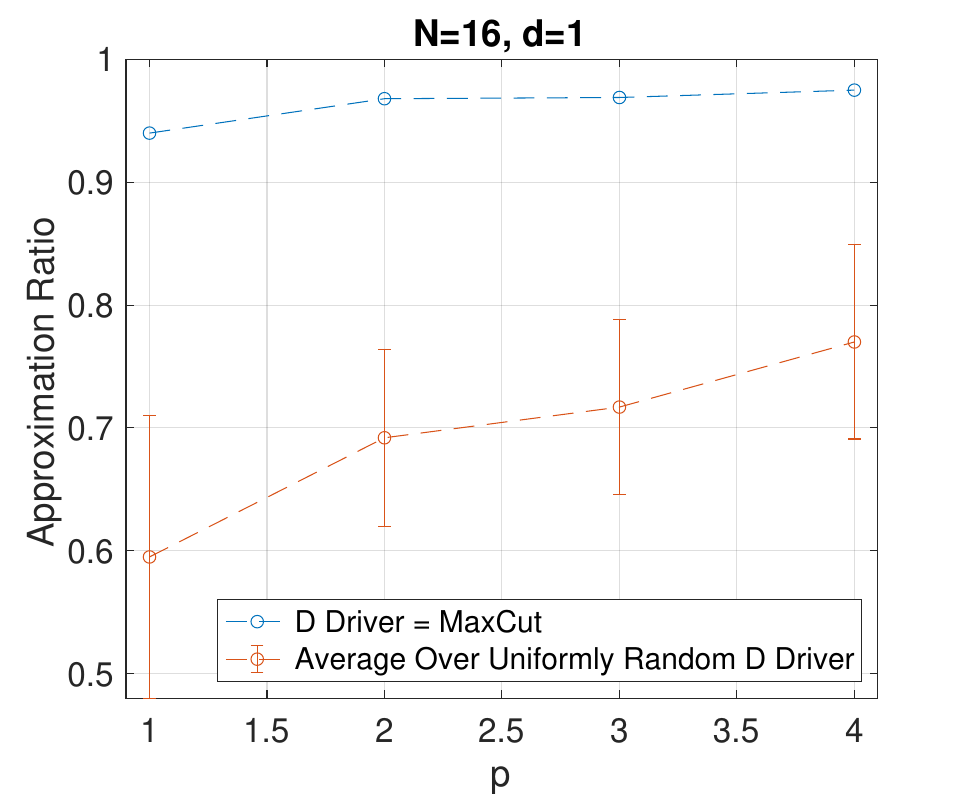}\includegraphics[width=0.5\textwidth]{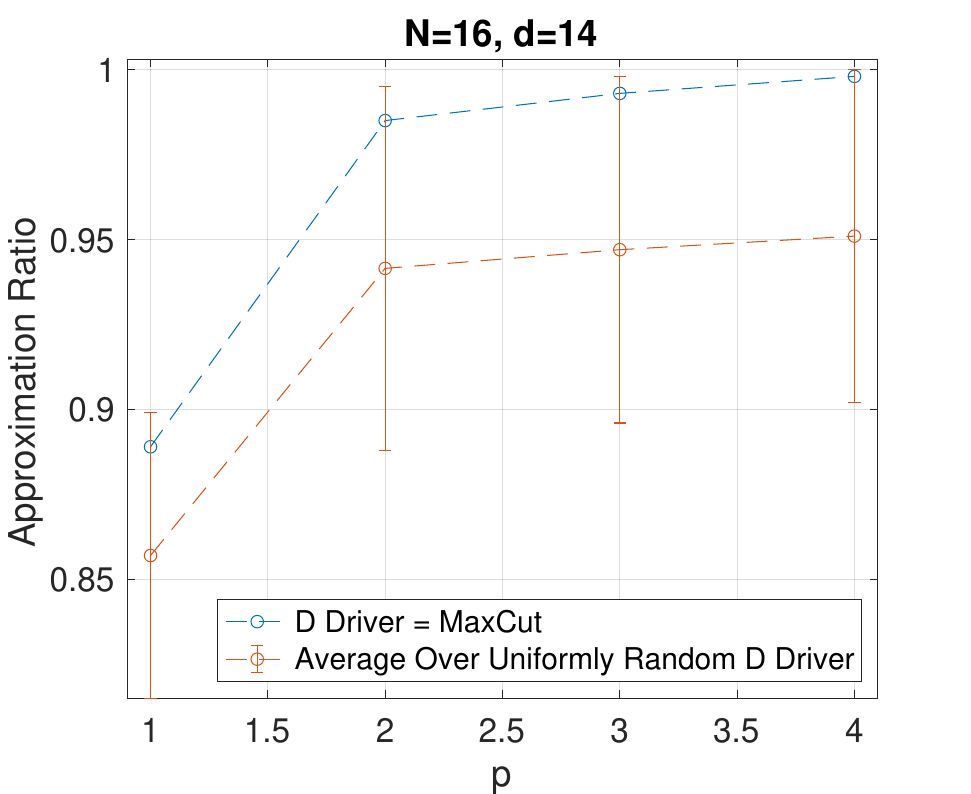}
    \caption{Performance of the simplified HamQAOA on low vs. high degree regular graphs of size 16 using $D$ driver that coincides with the MaxCut vs. randomly chosen $D$ driver. Initial points in parameter optimization are taken from optimizers of the finite-degree formula. Left: ring graph ($d=1$). Right: fully connected graph ($d=14$). For higher degree, the randomly chosen $D$ driver achieves much closer performance to the optimal $D$ driver. We average over $20$ uniformly random $D$ drivers.} 
    \label{fig:rand_vs_mc_d}
\end{figure}
Hence, when tractable, one should use the exact MaxCut of the given graph; this is often simple for practical physical models, like square and hexagonal lattices. Otherwise, an approximation algorithm such as the Goemans-Williamson SDP (\cite{goemans_williamson}) should be used. If the average degree of the graph is large, a random cut may be chosen.
\vspace{-2mm}
\subsection{Iterative Formula Parameter Initialization Strategy} 
\label{sec:iter_strategy}
To optimize the HamQAOA of depth $p$ on a 2-local LHP with an interaction graph that has average degree $d+1$ using the Iterative Formula Parameter (IFP) strategy, one optimizes Formula \ref{alg:fin_iter}, with interaction terms taken from the LHP Hamiltonian, at depth $p$ and degree parameter $d$. One then initializes the HamQAOA with the resulting optimal parameters and optimize, minimizing the output state's energy with respect to the LHP Hamiltonian. 
\begin{figure}[h!]
    \centering
    \includegraphics[width=0.5\textwidth]{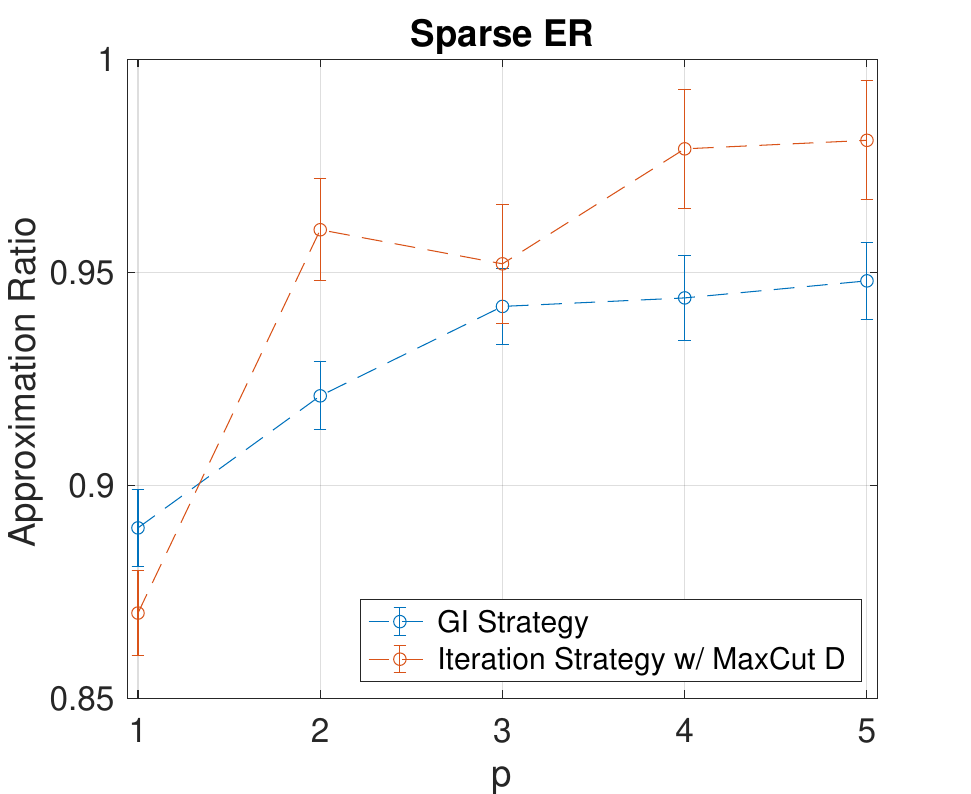}\includegraphics[width=0.5\textwidth]{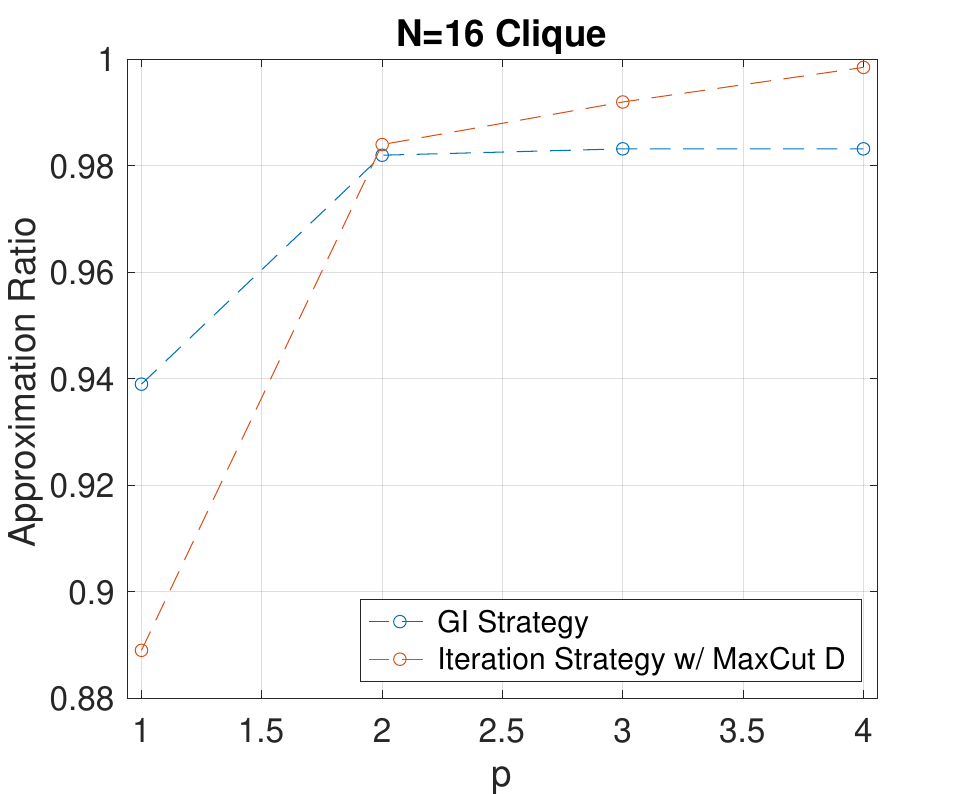}
    \caption{Performance of IFP vs. GI strategies on: (Left) an ensemble of 8 Erd\H{o}s-R\'enyi (ER) graphs with edge probability $2/N$ for $N=16$, and (Right) an all-to-all connected graph of size 16. In both cases, the iterative formula strategy outperforms the GI strategy for $p>1$. At $p=1$, the GI strategy outperforms because the optimal parameters do not have $\beta,\gamma = 0$ on specific graphs, but the formula parameters do. We do not find the iterative formula strategy better on dense ER graphs, where both regularity and high-girth assumptions are violated.}
    \label{fig:iter_vs_gi}
\end{figure}
Note that there may be multiple sets of parameters for the formula which yield the same energy; however, these are often equivalent up to symmetries of the parameter space, and can be collapsed using Protocol 1 (\ref{protocol:Gauge}). Our numerical benchmarks are performed for QMC, but this procedure applies to any LHP satisfying the requirements of Formula \ref{alg:fin_iter}.

The parameters obtained from optimizing Formula \ref{alg:fin_iter} are optimal for high-girth regular graphs in the average case; however, we find that the IFP strategy performs best when the $D$ driver is fixed to the classical MaxCut. Note that when degree $d$ is large, e.g. linear in system size, we expect that the performance of the IFP strategy should not suffer significantly when the $D$ driver is chosen at random vs when it coincides with the exact MaxCut. This expectation is confirmed by Figure \ref{fig:rand_vs_mc_d} (right).

As demonstrated in Figure \ref{fig:iter_vs_gi}, the IFP strategy can outperform the GI strategy on specific instances; particularly, approximately-regular high girth graphs and even low-girth but high degree regular graphs.
\subsection{Greedy Iterative Parameter Initialization Strategy}\label{sec:GI}
In this section, we detail and benchmark the Greedy Iterative optimization heuristic which is inspired by \cite{Sack_2023}. Suppose some set of $4p$ parameters achieves a QMC energy $E$ on a given graph. For each $i\in \{1, 2, \hdots p+1\}$ we can insert a row of parameters $(\alpha, \beta, \gamma, \delta)_i = (0, 0, 0, 0)$ into our original set of parameters, shifting all subsequent layers of the HamQAOA circuit.
\vspace{-2em}
\begin{figure}[h!]
    \centering
    \includegraphics[width=0.85\textwidth]{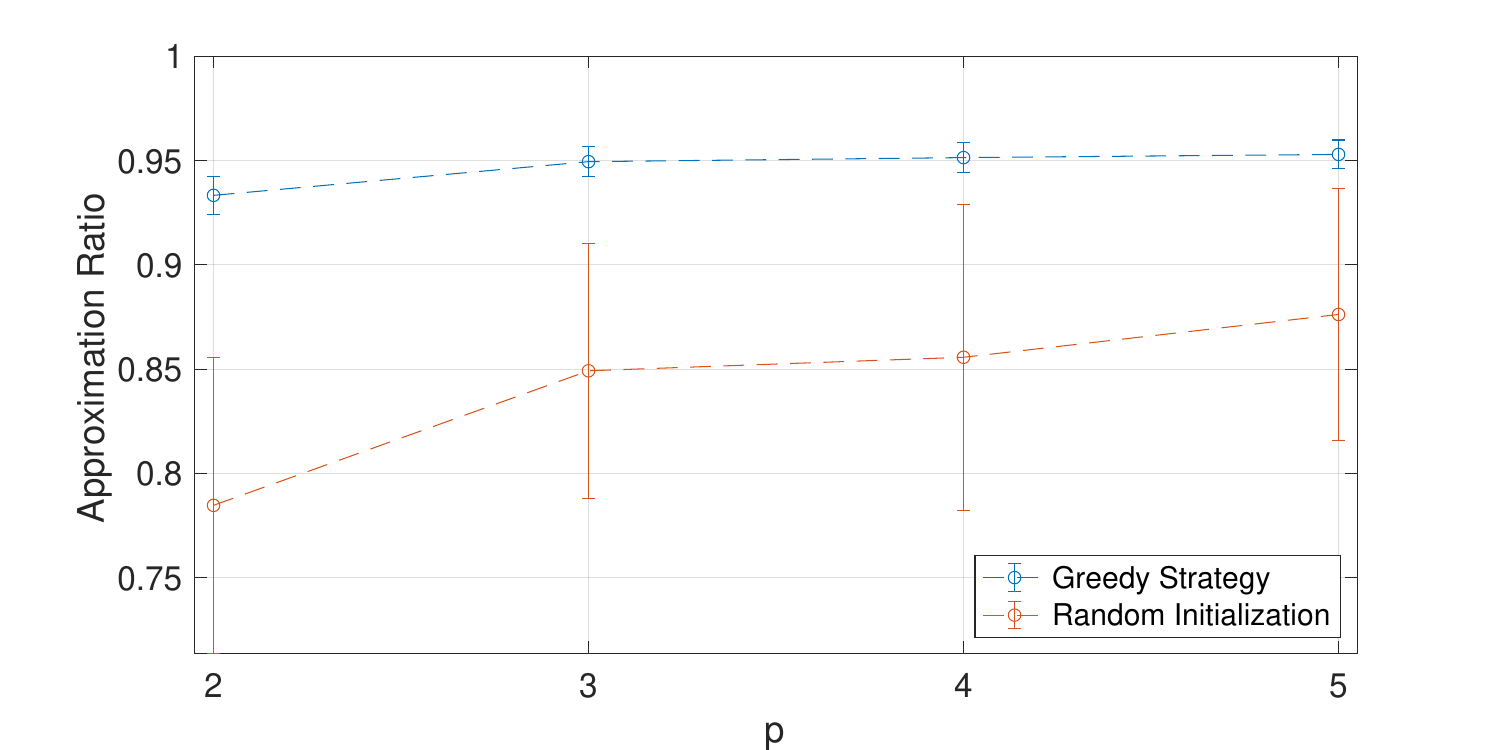}
    \caption{Performance of GI strategy vs. random initialization on ensemble of $8$ ER random graphs of size $14$. For the greedy strategy, we take the $2$ best-performing sets of parameters at each $p$ and append a row of zeroes to initialize at depth $p+1$. We also perform $2$ random initializations for even comparison.}
    \label{fig:enter-label}
\end{figure}
Since this layer of the circuit applies the identity map to our system, the resulting QMC energy is exactly $E$, now for a QAOA of depth $p+1$. Hence, for every $p$, this iterative strategy provides $p+1$ sets of initial parameters from which to optimize the circuit at depth $p+1$ with nondecreasing performance.
\cite{Sack_2023} shows that for the original 2-driver QAOA, these $p+1$ possible parameters constitute so-called \textit{transition states} in the energy landscape, which are connected in the negative curvature direction by a local minima of the depth $p+1$ QAOA, thus guaranteeing that better local minima are found at each successive depth.

This heuristic outperforms random initialization when the number of parameter initializations are limited, making it more efficient on a variety of graphs. However, despite the nondecreasing performance guarantee, it is possible to get stuck in poor local minima when iterating to large $p$, as the greedy path from small-$p$ initializations may not extend well on complex instances. In these instances, we find either the IFP strategy or random initialization more suitable, depending on the regularity and degree of the graph, as well whether the iterative formula is tractably optimized at the desired $p$.

\subsection{Zero Insertion for Greedy Strategy} \label{sec:0_insert}
\vspace{-2em}
\begin{figure}[H]
    \centering
    \includegraphics[width=0.9\textwidth]{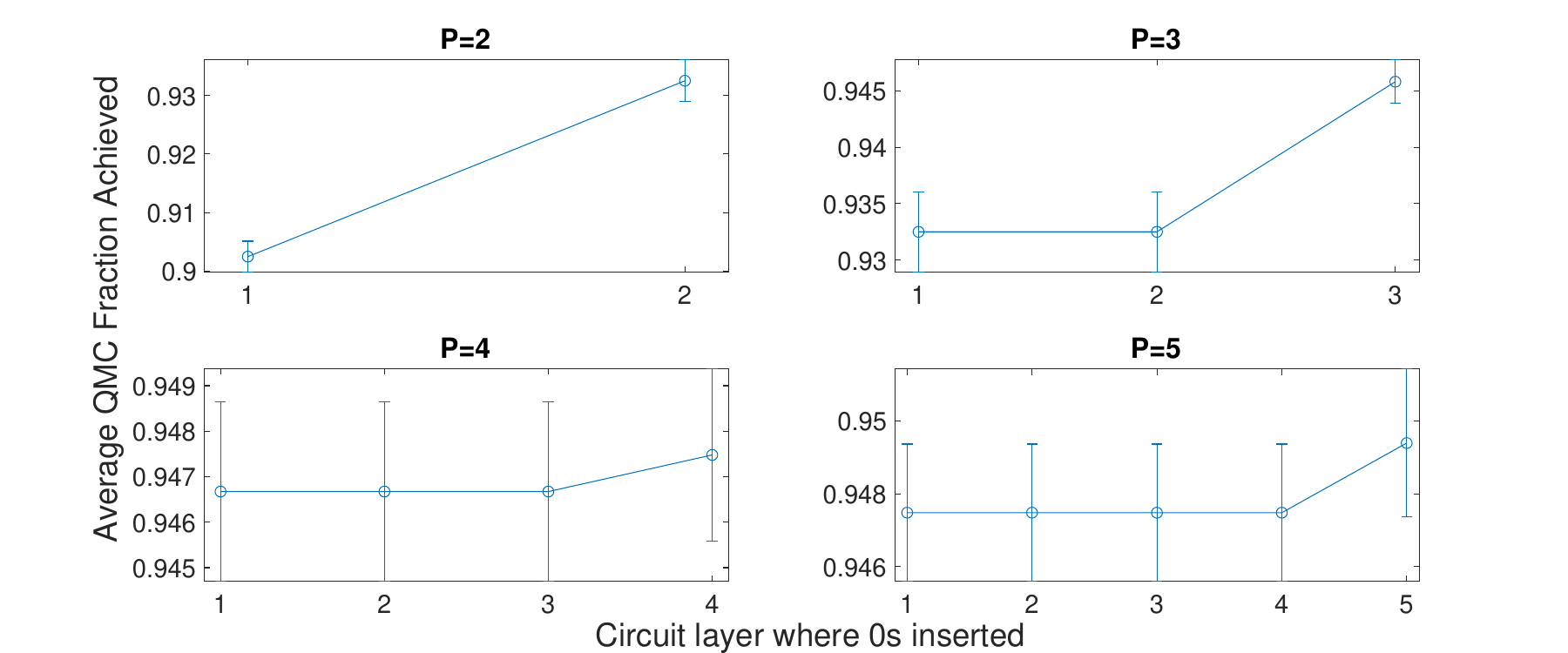}
    \caption{Performance of GI strategy with insertion of $0$ row at different indices. These results are obtained on an ensemble of 8 random 3-regular graphs of size 14. We find that inserting at the end performs best on average.}
    \label{fig:0_insertion}
\end{figure}
One may wish to make the GI optimization yet more efficient without sacrificing performance. We find that, on average, appending the row of zeroes to the end of the previous set of parameters, performs best out of the $p+1$ choices of layer at which the identity layer is inserted.

\subsection{Optimal rescaled $\nu_p$ in the $d\to\infty$ limit}
\label{sec:nu_p_values}
\vspace{-2em}
\begin{table}[h!]
\begin{center}
\begin{tabular}{|p{1cm} |p{1cm}|p{1cm}|p{1cm}|p{1cm}|p{1cm}|p{1cm}|p{1cm}|p{1cm}|}
\hline
$p$ & 1 & 2 & 3 & 4 & 5 & 6 & 7 & 8\\
\hline
$\nu_p$ &
0.3033 & 0.4459
 & 0.5045 &  0.5417
   & 0.5671
    &0.5909
    &0.6034
    &0.6191 \\
\hline
\end{tabular}
\end{center}
\vspace{-1em}
\caption{Numerically optimized values of $\nu_p(\paramv,X,X)+\nu_p(\paramv,Y,Y)+\nu_p(\paramv,Z,Z)$ defined in equation \eqref{eq:nu_p} for $1\le p\le 8$, given to the 4th decimal.}
\end{table}
\vspace{-2em}

\subsection{Optimal Parameters for QMC on Rings}\label{sec:optparams}
Here we give HamQAOA parameters that generate exact ground states of AFM Heisenberg rings of sizes 4 and 6. Parameters are listed in circuit order, i.e. the first row corresponds to the first set of HamQAOA gates applied.

\begin{table}[H]
\centering
\begin{tabular}{ |p{1.2cm}|p{1.5cm}|p{1.5cm}|p{1.5cm}|}
\hline
\multicolumn{4}{|c|}{$N=4, p=4$} \\
\hline
$\vect{\alpha}$& $\vect{\beta}$ &$\vect{\gamma}$& $\vect{\delta}$ \\
\hline
0.2821 & 0 & -1.2707 & -0.6880 \\
0.5697 & 0 & 0.0630 & -0.2841 \\ 
-1.0968 & 0 & -0.8312 & -1.3104 \\
1.1374 & 0 & 0.8710 & 1.4865 \\
\hline
\end{tabular}
\quad
\begin{tabular}{ |p{1.2cm}|p{1.5cm}|p{1.5cm}|p{1.5cm}|}
\hline
\multicolumn{4}{|c|}{$N=6, p=7$} \\
\hline
$\vect{\alpha}$& $\vect{\beta}$ &$\vect{\gamma}$& $\vect{\delta}$ \\
\hline
0.4440  &  0.5794 &  -1.5708 &   0 \\
-0.8367  & -0.7445  & -0.7854  &  0.7854\\
    1.4894  & -1.2421 &   1.1202&   -1.5686\\
    1.5708 &   1.0088  &  1.0335  & -1.9968\\
   -0.4696  & 0 &  -0.0025  & -1.3673\\
   -1.0117  & -0.7854   &-1.5708   & 0.3109\\
   -0.1558  & 0 &  -0.7854  &  0.7854\\
\hline
\end{tabular}
\caption{HamQAOA parameters to prepare ground state of QMC on $N$-vertex ring graphs exactly (up to floating-point precision), given to the 4th decimal. 
\label{tab:exactparam}}
\end{table}

\subsection{Symmetries in optimal parameters for high-girth regular graphs}
\vspace{-1em}
We are interested in optimizing Formula \ref{alg:fin_iter}, which gives the average-case HamQAOA performance on high-girth regular graphs, with respect to HamQAOA parameters. Performing this parameter optimization using random initializations generally yields degenerate sets of optimal parameters which all evaluate to yield the same energy. However, by considering symmetries in the parameter landscape, many of these degeneracies can be collapsed. Observe that adding or subtracting $\pi$ from any angle only contributes a global phase to the HamQAOA state, so our angles are all $\pi$-periodic with respect to the output energy. Furthermore, notice that adding or subtracting $\pi/2$ to any angle $\alpha$ introduces a term $\prod_{i\sim j} Z_jZ_j$ into the HamQAOA circuit \ref{eq:hamqaoa_circuit}. Recall that $d$ is one less than the degree of the regular graph in Formula \ref{alg:fin_iter}. For odd $d$, there are thus an even number of $Z$ terms in this product acting on each site. Then the expression $\prod_{i\sim j} Z_i Z_j$ commutes with all drivers and the QMC Hamiltonian, and does not affect the output state energy. Leveraging these observations, we use Protocol 1 to simplify our optimal iteration parameters.
\SetAlgorithmName{Protocol}{}{}
\label{protocol:Gauge}
\begin{algorithm}
    \caption{Reducing Parameter Degeneracy via ``Gauge-Fixing''}
    \KwIn{HamQAOA depth $p$, degree $d$ as defined in Formula \ref{alg:fin_iter}, set of $4p$ HamQAOA parameters $\paramv$.}
    \KwOut{Transformed set of $4p$ HamQAOA parameters.}
    Define subroutine RestrictToMiddle($\vect{\theta}$, $T$): \\ \Indp
    For all $\theta_i \in \vect{\theta}$: \\
    \Indp If $\theta_i < -T/2$, set $\theta_i = \theta_i + T$ \\
    Else if $\theta_i > T/2$, set $\theta_i = \theta_i - T$ \\
    \Indm Return $\vect{\theta}$.\\
    \Indm If $\alpha_1 < 0$, set $\paramv = -\paramv$. If $d$ is odd, set $\vect{\alpha} = \text{RestrictToMiddle}(\vect{\alpha}, \pi/2)$.\\
    For layer $l$ from 1 to $p-1$:\\
    \Indp Alpha Fixing: If $d$ even and $\alpha_l \lessgtr \mp \pi/4$, set $\alpha_l = \alpha_l \pm \pi/2, \beta_l = -\beta_l, \gamma_l = \text{RestrictToMiddle}(\gamma_l-\pi/2, \pi)$. \\
    Beta Fixing: If $\beta_l < 0$, set $\beta_l = \beta_l + \pi/2$, $\gamma_l = -\gamma_l$, $\beta_{l+1} = \text{RestrictToMiddle}(\beta_{l+1} - \pi/2, \pi)$ \\
    Gamma Fixing: If $\gamma_l < 0$, set $\gamma_l = \gamma_l + \pi/2, \delta_l = \delta_l$. If $d$ even, perform Alpha Fixing on layer $l+1$. If $d$ odd, set $\beta_{l+1} = -\beta_{l+1}, \gamma_{l+1} = \text{RestrictToMiddle}(\gamma_{l+1}-\pi/2, \pi)$. \\
    Delta Fixing: If $\delta_l < 0$, set $\delta_l = \delta_l + \pi/2, \gamma_{l+1} = -\gamma_{l+1}, \delta_{l+1} = \text{RestrictToMiddle}(\delta_{l+1} - \pi/2, \pi)$.\\
    \Indm Return $\paramv$.
\end{algorithm}

Intuitively, we are fixing a ``gauge'' where we force all $\beta,\gamma,\delta$ parameters except those in the last layer to be positive and restricted to $[0, \pi/2)$, and all $\alpha$ as close to $0$ as possible. At each layer, if a parameter is negative, we add $\pi/2$ and commute the induced Paulis to the next layer, subtracting $\pi/2$ from the same driver in the next layer and adjusting the intermediate parameters based on the standard commutation rules. Doing so iteratively, we continue until the second to last layer.
Figure \ref{fig:gauge_symm} demonstrates that this strategy is very effective in reducing the parameters obtained by the optimization.

This strategy may be useful for finding heuristic patterns in the optimal sets of parameters at larger QAOA depths, allowing one to choose good initial points at high $p$ without exhaustively searching the parameter landscape. We apply Protocol 1 to all combinations of $p$ and $d$ listed below, and report the gauge-fixed optimal parameters. 
\begin{center}
\begin{tabular}{ |p{0.8cm}|p{1cm}|p{0.6cm}|p{0.6cm}|p{1cm}|}
\hline
\multicolumn{5}{|c|}{$p=1$} \\
\hline
$d$ & $\vect{\alpha}$& $\vect{\beta}$ &$\vect{\gamma}$& $\vect{\delta}$ \\
\hline
1 & $\pi/8$  &  0 &  0 &   $\pi/8$ \\

3 & $0.2618$  &  0 &  0 &   $\pi/8$ \\
10 & $0.1531$  &  0 &  0 &   $\pi/8$ \\
100 & $0.0498$  &  0 &  0 &   $\pi/8$ \\
\hline
\end{tabular}
\vspace{-0.2cm}
\begin{tabular}{ |p{0.8cm}|p{2.8cm}|p{1.5cm}|p{1cm}|p{2.8cm}|}
\hline
\multicolumn{5}{|c|}{$p=2$} \\
\hline
$d$ & $\vect{\alpha}$& $\vect{\beta}$ &$\vect{\gamma}$& $\vect{\delta}$ \\
\hline
1 & (0.3279, 0.6214)  & (0, 0) &  (0, 0) &   (0.6214, 0.3279) \\

3 & (0.1819, 0.2123)  & ($\pi/4$, 0) &  (0, 0) &   (1.4188, 0.4711) \\
10 & (0.1085, 0.1261)  & ($\pi/4$, 0) &  (0, 0) &   (1.4198, 0.4678) \\
33 & (0.0614, 0.713)  & ($\pi/4$, 0) &  (0, 0) &   (1.4202, 0.4666) \\\hline
\end{tabular}
\end{center}
\begin{center}
\begin{tabular}{ |p{0.8cm}|p{4cm}|p{2cm}|p{2cm}|p{4cm}|}
\hline
\multicolumn{5}{|c|}{$p=3$} \\
\hline
$d$ & $\vect{\alpha}$& $\vect{\beta}$ &$\vect{\gamma}$& $\vect{\delta}$ \\
\hline
1 & $(0.4806, 0.5260, -0.6838)$  & $(0, 0, \pi/4)$ &  $(0, 0, -\pi/4)$ &   $(0.7716, 0.4090, \pi/4)$ \\

3 & $(0.1555, 0.2875, -0.1375)$  & $(\pi/4, 0, 0)$ &  $(0, 0, 0)$ &   $(1.3726, 1.2326, 0.8065)$ \\
10 & $(0.0966, -0.0997, 0.1750)$  & $(\pi/4, \pi/4, 0)$ &  $(0, 0, 0)$ &   $(1.4104, 0.6992, 0.3538)$ \\
33 & $(0.0541, -0.0569, 0.0977)$  & $(\pi/4, 0, 0)$ &  $(0, 0, 0)$ &   $(1.4061, 1.6979, 0.3537)$\\\hline
\end{tabular}
\end{center}
\begin{center}
\begin{tabular}{ |p{0.8cm}|p{5.7cm}|p{2.5cm}|p{1.7cm}|p{5cm}|}
\hline
\multicolumn{5}{|c|}{$p=4$} \\
\hline
$d$ & $\vect{\alpha}$& $\vect{\beta}$ &$\vect{\gamma}$& $\vect{\delta}$ \\
\hline
1 & $(0.6831, -0.2580, 0.6224, -0.5412)$  & $(\pi/4, \pi/4, 0, 0)$ &  $(0, 0, 0, 0)$ &   $(1.2135, 0.4228, 1.0670, 1.2715)$ \\

3 & $(0.2019, 0.1464, 0.1735, -0.1265)$  & $(\pi/4, \pi/4, 0, 0)$ &  $(0, 0, 0, 0)$ &   $(1.3078, 0.1804, 0.6461, 1.0604)$\\
10 & $(0.1100, 0.0823, 0.0921, -0.0641)$  & $(\pi/4, \pi/4, 0, 0)$ &  $(0, 0, 0, 0)$ &   $(0.2664, 0.1838, 0.6389, 1.0826)$ \\\hline
\end{tabular}
\end{center}
\begin{figure}[h!]
    \centering
    \includegraphics[width=0.8\linewidth]{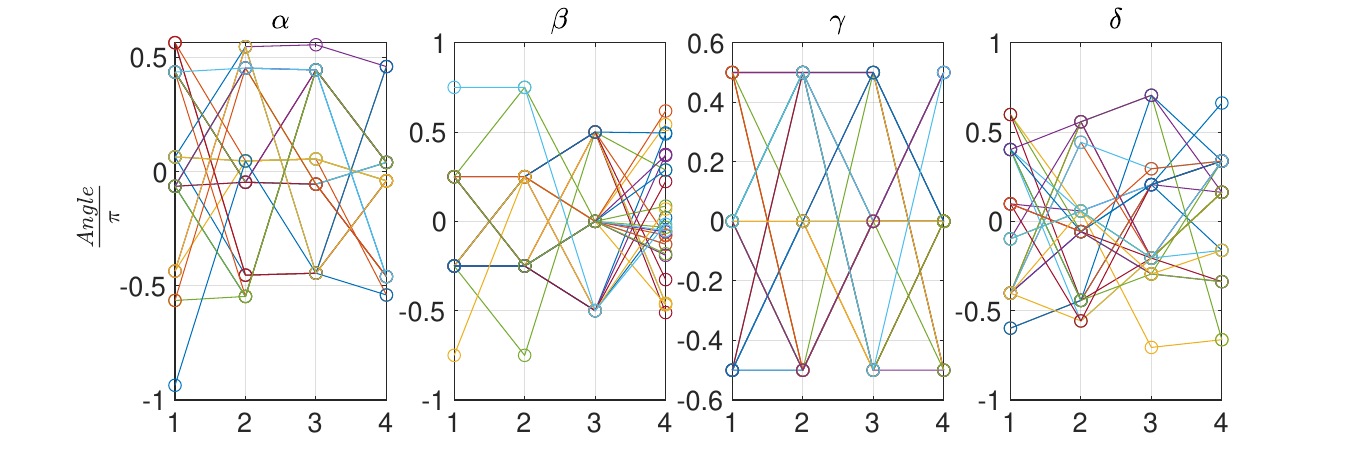}
    \includegraphics[width=0.8\linewidth]{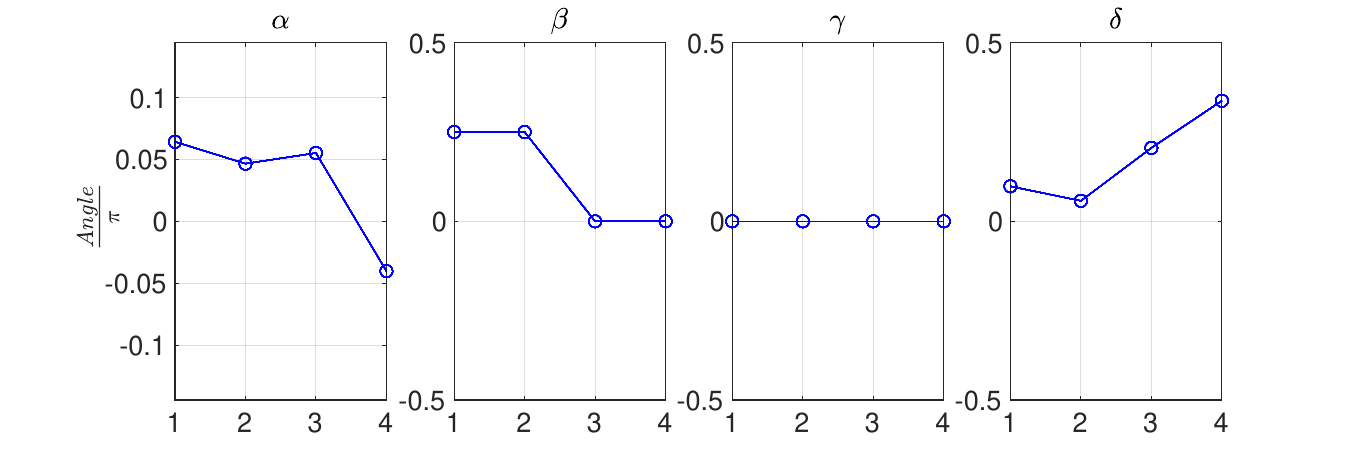}
    \caption{Optimal parameters for Formula \ref{alg:fin_iter} at $p=4$ and $d=3$. Top: Raw sets of optimal parameters yielding the same energy. Bottom: Gauge-symmetry parameter reduction applied to optimal parameters, collapsing all sets to a single one.}
    \label{fig:gauge_symm}
\end{figure}
\vspace{-2em}

\section{Comparison of simplified HamQAOA to Generalized Variants} 
\label{sec:gen_ansatze}
In this section, we benchmark the simplified HamQAOA against two versions of Algorithm \ref{GHamQAOA}.

\subsection{Spherical HamQAOA}
First, we define the Spherical HamQAOA ansatz, which is simply Algorithm \ref{GHamQAOA}, but with $\mv = \nv$ (noting that these vectors can be chosen from the full unit sphere). 
\begin{figure}[h!]
    \centering
    \includegraphics[width=\textwidth]{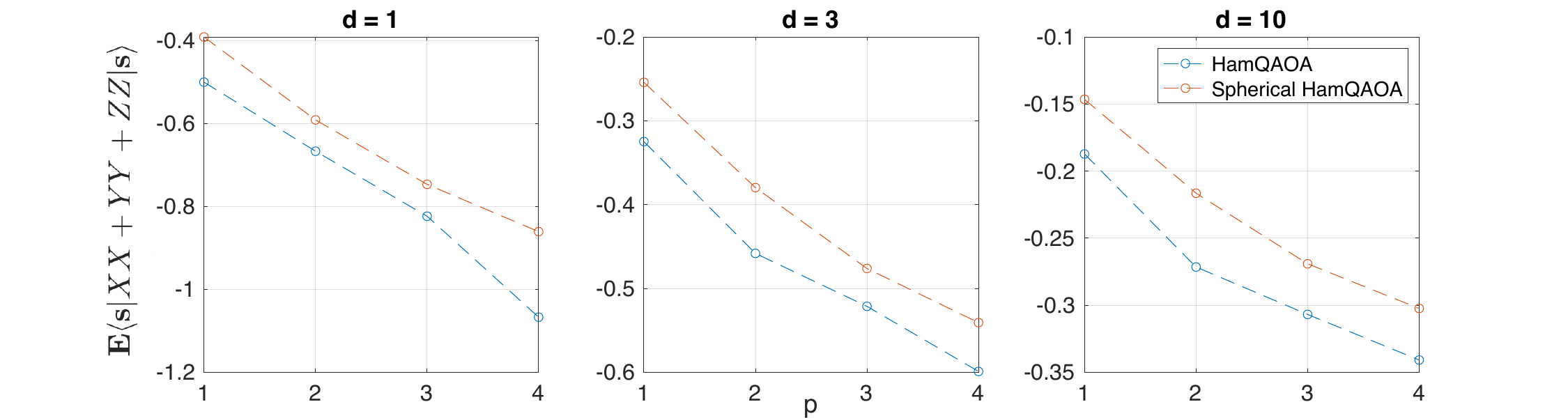}
    \caption{Optimal average-case energy of Spherical vs simplified HamQAOA on high-girth regular graphs. At all $p$ and $d$, the simplified HamQAOA achieves a better energetic minimum on the non-identity part of the QMC Hamiltonian.}
    \label{fig:sp_vs_rx}
\end{figure}
Upon applying this restriction, we can compute the average-over-$\nv$ performance of the Spherical HamQAOA on high-girth regular graphs using Formula \ref{alg:fin_iter}. If we expand
\begin{align}
\ketbra{\nv} &= \frac{1}{2} \left( I + \nv \cdot\bm{\sigma}\right) \\
e^{i\delta \nv\cdot\bm{\sigma}} &= \cos{\delta} \ I + i\sin{\delta} \ \nv\cdot\bm{\sigma}
\end{align}
we can see that $f_{\nv, \mv}(\zv)$ as defined in \ref{alg:fin_iter} is a degree $2p+1$ polynomial in $\nv$ (where $\mv=\nv$. Then we must take the expectation over choice of $\nv$ by
\begin{equation}
\bar{f}^\sigma(\zv) = \mathbb{E}_{\nv} f_{\nv}^\sigma(\zv) = \frac{1}{|S|} \sum_{\bm{s} \in S} f^\sigma_{\bm{v}}(\vect{z})
\end{equation}
where $S \subset \mathbb{S}^2$ is a spherical $(2p+1)$-design \cite{spherical_design}. Using this method of averaging and optimizing for the best Spherical HamQAOA parameters at various $d$ and $p$, we can benchmark this HamQAOA variant against the simplified HamQAOA, Algorithm \ref{HamQAOA}.

Figure \ref{fig:sp_vs_rx} suggests that maximizing the inhomogeneity in rotation between qubits by pushing the $D$ driver to opposing poles of the Bloch sphere is more efficient than using the entire Bloch sphere, as the simplified HamQAOA outperforms the Spherical HamQAOA.  
\subsection{General HamQAOA}
\begin{figure}[h!]
    \centering
    \includegraphics[width=0.8\textwidth]{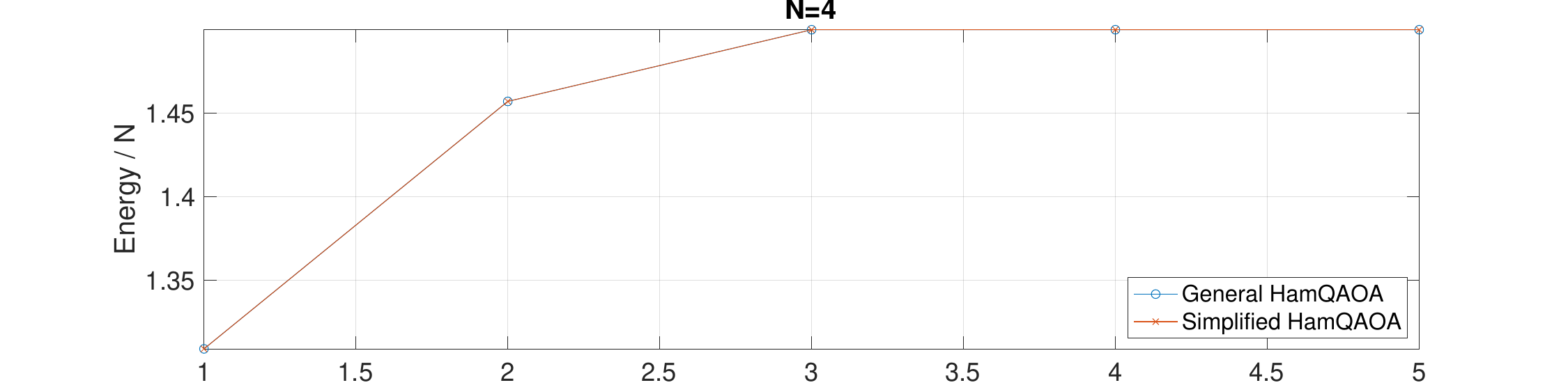}
    \includegraphics[width=0.8\textwidth]{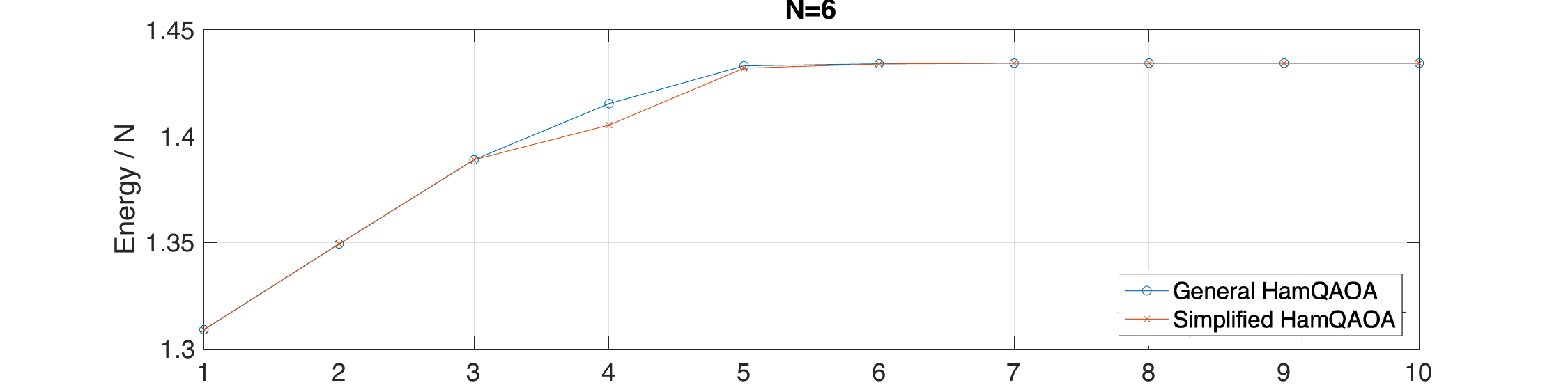}
    \includegraphics[width=0.8\textwidth]{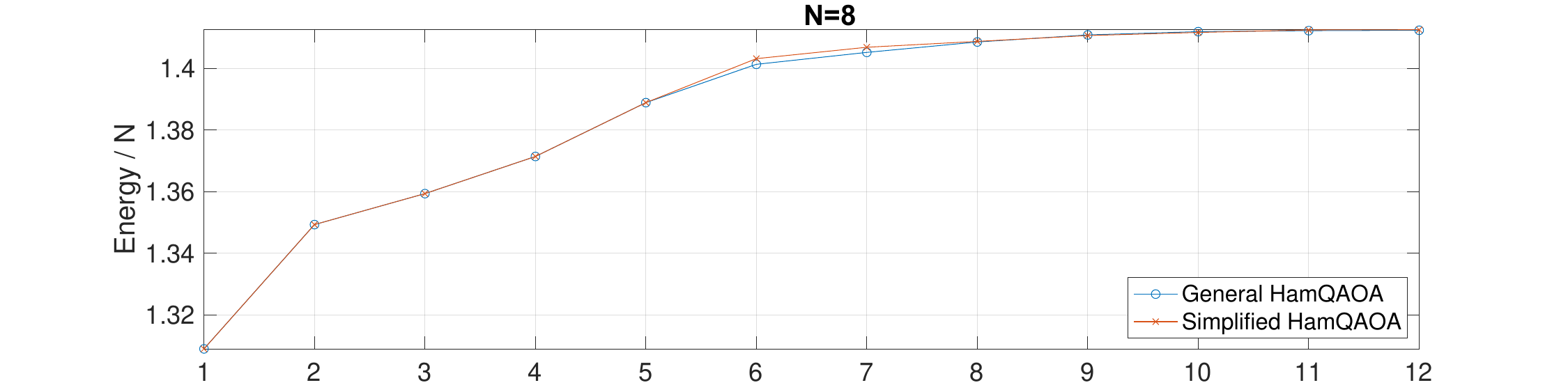}
    \caption{Energy Achieved by optimizing over both the $4p+4|V|$ parameters of the general ansatz and the $4p$ parameters of our simpler choice of HamQAOA ansatz for rings of size $4, 6, 8$ at various $p$. We find that the energy density achieved are very close at all $p$, making it advantageous to use the simplified ansatz for computational savings. At $N=8$, the simplified ansatz slightly outperforms the general one on our experiments as the tradeoff between ease of optimization and algorithm robustness tips toward the simpler ansatz.}
    \label{fig:gen_vs_ham}
\end{figure}
Next, we open up both $\mv$ and $\nv$ degrees of freedom, and consider optimizing over both the choice of Algorithm \ref{GHamQAOA} ansatz and its variational parameters. While the General HamQAOA generalizes all other possible variants, it is challenging to optimize in practice. This is because each $\nv, \mv$ can be described by two angles, leaving the ansatz with $4p + 4|V|$ variational parameters. Though classical optimization of these parameters is unfeasible except for very simple systems, we can numerically compare the General and simplified HamQAOA on small Heisenberg rings. This comparison is displayed in Figure \ref{fig:gen_vs_ham}.

We note that the general ansatz actually finds exact ground states of the $N=4$ ring at $p=3$; however, both ansatze find exact ground states of the $N=6$ ring at $p=7$. Given the optimization cost on more complex instances and the observation in Figure \ref{fig:gen_vs_ham} that the simplified HamQAOA achieves slightly better energies than the General HamQAOA on rings of only size $8$, it is evident that for QMC, our simpler ansatz will suffice, and the challenge of finding good parameters whose number scales with system size will hamper the general ansatz.

We also find that the optimization almost always pushes the vectors $\nv_v$ into the Bloch sphere's $XY$ plane, away from the $Z$-axis of our entangling gate. This supports the choice of our HamQAOA $D$ driver along the $\hat{\vect{x}}$-axis.
\vspace{-1em}

\section{Proofs of Iterative Formulas}
\label{sec:proof_iter}
Here we give the proofs of both finite-degree and infinite-degree formulas, alongside mathematical properties of the functions defined in these formulas.

\subsection{Proof of Finite-Degree Iterative Formula for HamQAOA on High-Girth Regular Graphs}
\begin{figure}[h!]
    \centering
    \includegraphics[scale = 0.35]{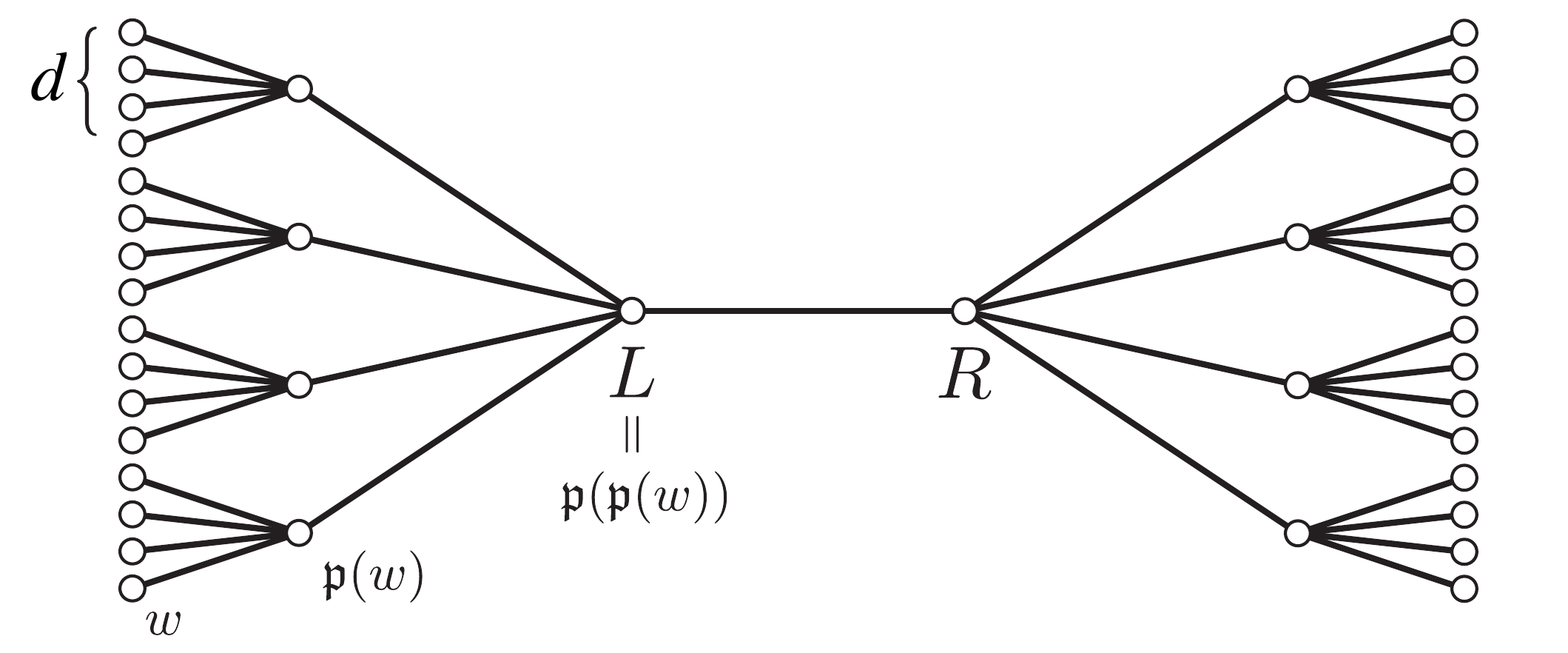}
    \caption{Local subregion of a high-girth $(d+1)$-regular graph, which looks like two depth-p trees glued at their roots. Our formula computes the expected Quantum MaxCut value across edge L-R, averaged over choice of the inhomogeneous $D$-driver. We denote the leaves of each tree as $w$, and for any vertex $v$, its parent as $\mathfrak{p}(v)$, as adapted from \cite{Leo_MaxCut}.}
    \label{fig:D_reg_tree}
\end{figure}

\begin{proof}[Proof of Formula \ref{alg:fin_iter}]
    We shall prove the formula for $p=2$ for simplicity, as it extends directly to all $p$. Since only vertices $\leq p$ edges away from $L$ and $R$ contribute to the energy across edge $LR$ prepared by the QAOA, we may restrict our attention to a subgraph which looks like Figure \ref{fig:D_reg_tree}. This subgraph has $N = 2(d^p+...+d+1)$ vertices. We look to evaluate 
    \begin{equation}
        \mathbb{E}_{\mv, \nv}\left[\bra{\mv}\hat{O}(\sigma_L, \sigma_R)\ket{\mv}\right] = \mathbb{E}\left[\bra{\mv}e^{i\alpha_1 A}E_1e^{i\alpha_2 A}E_2\sigma_L\sigma_R E_{-2}e^{-i\alpha_2 A}E_{-1}e^{-i\alpha_1 A}\ket{\mv}\right]
    \end{equation}
     where
    \begin{equation}
        E_j = \bigotimes_v E_j^{(v)} = \begin{cases}
            \bigotimes_v e^{i\beta_j X_v} e^{i\gamma_j Z_v} e^{i\delta_j s_vX_v}, \qquad &j > 0 \\
            \bigotimes_v e^{-i\delta_j s_vX_v} e^{-i\gamma_j Z_v} e^{-i\beta_j X_v}, &j < 0
        \end{cases}
    \end{equation}
    and where $\ket{\mv} = \bigotimes_v \ket{\mv_v}$, with the expectation taken over all $\mv_v$ and $\nv_v$ drawn i.i.d. from distributions $\cD_m, \cD_n$ over the Bloch sphere, respectively. 
    Since $e^{i\alpha_j A}$ is diagonal in the computational basis, we may replace it with $\sum_z \ket{\zv}e^{i\alpha_jA(z)}\bra{
    \zv
    }$.
    With this substitution, we can insert 6 complete sets in the $Z$-basis:
     \begin{align}
       \bra{\mv}\hat{O}(\sigma_L, \sigma_R)\ket{\mv} = &\sum_{\{\vect{z}^{(\cdot)} \}} \langle{\mv}|\zv^{(1)}\rangle e^{i\alpha_1 A(\zv^{(1)})} \langle\zv^{(1)}|E_1|\zv^{(2)}\rangle e^{i\alpha_2 A(\zv^{(2)})} \langle \zv^{(2)}|E_2|\zv^{(p+1)}\rangle \notag\\&\times\langle \zv^{(p+1)}| \sigma_L \sigma_R |\zv^{-(p+1)}\rangle \langle \zv^{-(p+1)}|E_{-2}|\zv^{(-2)}\rangle e^{-i\alpha_2 A(\zv^{(-2)})} \langle \zv^{(-2)}|E_{-1}|\zv^{(-1)}\rangle\notag\\&\times e^{-i\alpha_1 A(\zv^{(-1)})} \langle\zv^{(-1)}|\mv\rangle
    \end{align}
    where the sum is taken over all combinations of 6 computational basis bitstrings, each on $N$ qubits. We have upper-indexed the different complete sets for clarity in the following step. Grouping terms corresponding to each vertex, 
    \begin{align}
        \bra{\mv}\hat{O}(\sigma_L, \sigma_R)\ket{\mv}
        =& \sum_{\{\vect{z}^{(\cdot)} \}} \langle \zv^{(p+1)}|\sigma_L \sigma_R|\zv^{-(p+1)}\rangle
        \exp\left( i\alpha_1 A(\zv^{(1)}) + i\alpha_2 A(\zv^{(2)}) - i\alpha_2 A(\zv^{(-2)}) - i\alpha_1 A(\zv^{(-1)}) \right) \notag\\&\times
        \prod_v \langle \mv_v| z_v^{(1)}\rangle \langle z_v^{(1)}| E_1^{(v)} |z_v^{(2)}\rangle \langle z_v^{(2)}| E_2^{(v)} |z_v^{(p+1)}\rangle \langle z_v^{-(p+1)}| E_{-2}^{(v)} |z_v^{(-2)}\rangle \\ &\qquad\qquad\qquad\qquad\qquad\qquad\cdot\langle z_v^{(-2)}| E_{-1}^{(v)}|z_v^{(-1)}\rangle \langle z_v^{(-1)} |\mv_v\rangle
    \end{align}
   
    Here we have used the fact that the $B, C$ and $D$ drivers are sums of single-qubit terms and therefore factorize over the vertices. Then substituting our definitions of $\vect{\mathcal{A}}$ and $f$ and using the fact that $\sigma_L$ and $\sigma_R$ act on different qubits, we obtain
    \begin{equation}
        \bra{\mv}\hat{O}(\sigma_L, \sigma_R)\ket{\vect{m}} = \sum_{\{\vect{z}^{(\cdot)}\}} \langle \zv^{(p+1)}|\sigma_L \sigma_R|\zv^{-(p+1)}\rangle \exp\left(i\sum_{j=-2}^2 \mathcal{A}_jA(\zv^{(j)})\right)\prod_v f_{\mv_v, \nv_v}^{\sigma_v}(\vect{z}_v)
    \end{equation}
    where $\sigma_v = I$ for all $v\neq L, R$ and the bitstring $\vect{z}_v = (z_v^{(1)}, z_v^{(2)}, z_v^{(p+1)}, z_v^{-(p+1)}, z_v^{(-2)}, z_v^{(-1)})$ contains the indices of the complete sets which act on vertex $v$. Accordingly, we can rewrite the outer sum over length-$N$ bitstrings corresponding to complete sets on all vertices as an equivalent sum over length-$2p+2$ bitstrings corresponding to all ordered configurations of each vertex. We can further simplify by using the definition of $A$, as we know $A(\vect{z}) = \sum_{u\sim v} z_uz_v$ Since the $\vect{m}, \nv$-dependence above is contained in $f$, we may take the expectation of both sides and write:
    \begin{equation}
    \label{23}
        \mathbb{E}_{\mv, \nv}\left[\bra{\vect{m}}\hat{O}(\sigma_L, \sigma_R)\ket{\vect{m}}\right] = \sum_{\{\vect{z}_v\}} \exp\left( -i \sum_{u \sim w} \vect{\mathcal{A}} \cdot (\vect{z}_u \vect{z}_w) \right)  \prod_v \bar{f}^{\sigma_v}(\zv_v)
    \end{equation}
    Let us consider performing this sum over bit configurations of each vertex by first considering any leaf node of the tree subgraph, $w$, whose parent is $\mathfrak{p}(w)$. Since $w$ is only connected to $\mathfrak{p}(w)$, its contribution to the sum in \eqref{23} is simply
    \begin{equation}
        \sum_{\vect{z}_w} \exp\left( -i \vect{\mathcal{A}} \cdot (\vect{z}_w \vect{z}_{\mathfrak{p}(w)}) \right) \bar{f}^I(\zv_w)
    \end{equation}
    where we sum over the $2^{2p+2} = 64$ possible configurations of $w$. Each of $d$ vertices contributes this term to its parent, so when we sum over all possible configurations of the leaves, we get a term, for a fixed configuration of the parent:
    \begin{equation}
        H^{(1)}_d(\zv_{\mathfrak{p}(w)}) = \left( \sum_{\vect{z}_w} \exp( -i \vect{\mathcal{A}} \cdot (\vect{z}_w \vect{z}_{\mathfrak{p}(w)})) \bar{f}^I(\vect{z}_w)\right)^d
    \end{equation}
    Then the only remaining vertex to which $\mathfrak{p}(w)$ is connected is its parent, for whom the same concept applies. Packaging the contributions from its children into $H^{(1)}_d$, we can sum over configurations of $\mathfrak{p}(w)$ using \eqref{23}:
    \begin{equation}
        \sum_{\vect{z}_{\mathfrak{p}(w)}} \exp( -i \vect{\mathcal{A}} \cdot (\vect{z}_{\mathfrak{p}(w)} \vect{z}_{\mathfrak{p}(\mathfrak{p}(w))})) H_d^{(1)}(\vect{z}_{\mathfrak{p}(w)}) \bar{f}^I(\vect{z}_{\mathfrak{p}(w)}))
    \end{equation}
    Once again, $\mathfrak{p}(\mathfrak{p}(w))$ has $d$ children each yielding the same contribution; hence, for one of its fixed configurations,
    \begin{equation}
        H_d^{(2)}(\zv_{\mathfrak{p}(\mathfrak{p}(w))}) = \left(\sum_{\vect{z}_{\mathfrak{p}(w)}} \exp( -i \vect{\mathcal{A}} \cdot (\vect{z}_{\mathfrak{p}(w)} \vect{z}_{\mathfrak{p}(\mathfrak{p}(w))})) H_d^{(1)}(\vect{z}_{\mathfrak{p}(w)}) \bar{f}^I(\vect{z}_{\mathfrak{p}(w)}))\right)^d
    \end{equation}
    Here we note that at $p=2$, the vertex $\mathfrak{p}(\mathfrak{p}(w))$ is one of our roots — WLOG we let it be $L$. By symmetry, $R$ will have the same contribution. Then to evaluate \eqref{23} we now need only evaluate the sum over the $2p+2 = 6$ bits each of $L$ and $R$, leaving us with
    \begin{equation}
         \mathbb{E}_{\mv, \nv}\left[\bra{\vect{m}}\hat{O}(\sigma_L, \sigma_R)\ket{\vect{m}}\right] = \sum_{\vect{z}_L, \vect{z}_R} \exp( -i \vect{\mathcal{A}} \cdot (\vect{z}_L \vect{z}_R)) H_d^{(2)}(\vect{z}_L) H_d^{(2)}(\vect{z}_R) \bar{f}^{\sigma_L}(\vect{z}_L) \bar{f}^{\sigma_R}(\vect{z}_R)
    \end{equation}
    as our final answer. Naturally, the formula can be extended to any $p$, defined accordingly as 
    \begin{align}
        H_d^{(0)}(\vect{z}) &= 1 \\
        H_d^{(k)}(\vect{z}) &= \left( \sum_{\substack{\vect{x} \\ x^{(p+1)}=x^{-(p+1)}}} \exp( -i \vect{\mathcal{A}} \cdot (\vect{x} \vect{z})) H_d^{(k-1)}(\vect{x}) \bar{f}^I(\vect{x}) \right)^d
    \end{align}
    and with $H_d^{(p)}$ appearing in place of $H_d^{(2)}$ in the final result. Note here that $\bra{x^{(p+1)}}I\ket{x^{-(p+1)}}$, which appears implicitly in $\bar{f}^I(\xv)$, is nonzero only if $x^{(p+1)}=x^{-(p+1)}$, hence the restriction of the above sum. As each of $p$ levels of the formula requires summing over $2^{2p+1}$ bitstrings (with the aforementioned restriction) for each of $2^{2p+1}$ entries which must be explicitly calculated. The final sum, dominating the computation, is over $O(16^p)$ strings , so the time complexity of evaluating this formula is given by $O(16^p)$. By storing only the most recently computed $H_d^{(k)}$ we need memory of order $O(4^p)$.
\end{proof}
\vspace{-2em}
\subsection{General Properties of Iterative Formula Functions}
We now prove a few lemmas that will be instrumental to the proof of the $D\rightarrow\infty$ formula. First, we prove generally useful properties of the functions in Algorithm \ref{alg:fin_iter}.
\begin{lemma}\label{lem:sum_f_1}
    For the function $\Bar{f}^\sigma(\cdot)$ at depth $p$ as defined in Formula \ref{alg:fin_iter},
    \begin{equation}
        \sum_{\vect{z}\in \{\pm 1\}^{2p+2}} \Bar{f}^I(\vect{z}) = 1
    \end{equation}
\end{lemma}
\begin{proof}
    \normalfont{}
    $E_j(s, \beta_j, \gamma_j, \delta_j) = e^{i\beta_j X}e^{i\gamma_j Z}e^{is\delta_j X}$ is unitary by construction. 
    By linearity of expectation,
    \begin{align}
        \sum_{\vect{z}\in \{\pm 1\}^{2p+2}} \Bar{f}^I(\vect{z}) = \mathbb{E}_{\mv_v, \nv_v}\Big{[}&\bra{\mv_v}\sum_{z_1}\ketbra{z_1}E_1\sum_{z_2}\ketbra{z_2}E_2\hdots \\ &E_p\sum_{z_{p+1}}\ketbra{z_{p+1}}I\sum_{z_{-(p+1)}} \ketbra{z_{-(p+1)}}E_p^\dagger\hdots \sum_{z_{-1}} \ket{z_{-1}}\braket{z_{-1}}{\mv_v}\Big{]}
    \end{align}
    Since each sum $\sum_{z_i} \ketbra{z_i} = I$, all such sums drop out, leaving us with
    \begin{equation}
        \sum_{\vect{z}\in \{\pm 1\}^{2p+2}} \Bar{f}^I(\vect{z}) = \mathbb{E}_s\Big{[}\bra{\mv_v}E_1E_2\hdots E_p E_p^\dagger\hdots E_2^\dagger E_1^\dagger\ket{\mv_v}\Big{]} = \mathbb{E}_{\mv_v}[\braket{\mv_v}{\mv_v}] = 1
    \end{equation}
\end{proof}
\vspace{-1em}
\noindent Next we define an important symmetry transformation of $f$ and $H$. Let $B$ be the set of all $(2p+2)$-bit strings, indexed according to our convention \eqref{idx_convention}. 
We define the following subset
\begin{equation}
B_0 = \{\av \in B : a_{-r} = a_r \text{ for all } 1\le r \le p+1\}.
\end{equation}

For any $\av \in B$, we define the $T(\av)$ to be the largest positive index $T$ such that $a_{T} \neq a_{-T}$ if $\av\not \in B_0$, and 0 if $\av \in B_0$.
More formally,
\begin{equation}
    T(\av) = \max \big( \{j :  a_j \neq a_{-j}]\} \cup \{0\}\big).
\end{equation}
For any bit $\av\not\in B_0$, let $\av'\not\in B_0$ be the bitstring with
\begin{equation}
    a'_{\pm j} = \begin{cases}
        a_{\pm j}, & 1 \le j \le T(\av) \\
        -a_{\pm j}, & T(\av) +1 \le j  \le p+1
    \end{cases}
\end{equation}
Note $\av'=\av$ if $T(\av)=p+1$, and $\av'\neq \av$ if $1\le T(\av)\le p$. With these definitions, we obtain a number of symmetries.

\begin{lemma}
\label{lem:prime_op}
For any $\av\not\in B_0$ and $\mv_v \in \mathbb{S}^2$, we have 
    \begin{equation}
    f^I_{\mv_v, \nv_v}(\av') = -f^I_{\mv_v}(\av), \qquad
    f^Z_{\mv_v}(\av') = f^Z_{\mv_v}(\av) 
    \end{equation}
\end{lemma}
\begin{proof}
We denote the combined single-qubit rotation in the $j$-th layer as $U_j = \exp(i\tau_j \hat{r}_j\cdot \vec\sigma)$.
Then
\begin{equation} \label{eq:fprod}
    f_{\vect{m}_v}^\sigma(\av)
    = \braket{\mv_v}{a_1}\braket{a_{-1}}{\mv_v} \braket{a_{p+1}}{\sigma |a_{-(p+1)}}  
    \prod_{j=1}^p \Big[\braket{a_j}{U_j|a_{j+1}} \langle{a_{-(j+1)}}|U_j^\dag |{a_{-j}}\rangle  \Big]
\end{equation}
Note for any $j$, we have
\begin{equation}
\braket{a}{U_j|b} = 
    \begin{cases}
    \cos\tau_j + i br_j^z \sin \tau_j, & a=b \\
    i\sin\tau_j (r_j^x + i b r_j^y), & a\neq b
    \end{cases}
\end{equation}
where  $a,b\in \{+1,-1\}$.

Note we only need to consider $\{\av: 1\le T(\av)\le p\}$, since $\Bar{f}^I(\av)=\Bar{f}^Z(\av)=0$ if $T(\av)=p+1$.
For these $\av$, under $\av\to\av'$, we note that the products of interior terms, corresponding to $j\ge T(\av)+1$ in \eqref{eq:fprod}, remain the same. This is because for any $a\in\{\pm1\}$,
\begin{align*}
\braket{a}{U_j|a}\langle{a}|{U_j^\dag |a}\rangle &= 
    (\cos\tau_j + i ar_j^z \sin \tau_j)(\cos\tau_j - i ar_j^z \sin \tau_j) = \cos^2\tau_j + (r_j^z)^2\sin^2\tau_j\\
    &=
    \braket{{-}a}{U_j|{-}a}\langle{{-}a}|{U_j^\dag |{-}a}\rangle, \\
\langle{a}|{U_j|{-}a}\rangle\langle{{-}a}|{U_j^\dag |a}\rangle &= 
    i\sin\tau_j (r_j^x - i a r_j^y) \times -i\sin\tau_j (r_j^x + i a r_j^y) = \sin^2\tau_j [(r_j^x)^2+(r_j^y)^2]\\
    &=
    \langle{-a}|{U_j|a}\rangle\langle{a}|{U_j^\dag |{-}a}\rangle .
\end{align*}
The product of boundary terms, corresponding to $j=T(\av)$ in \eqref{eq:fprod}, picks up an overall negative sign, since
\begin{align*}
\braket{a}{U_j|a}\langle{a}|{U_j^\dag |{-}a}\rangle &= (\cos\tau_j + ia r_j^z\sin\tau_j)\times [-i\sin\tau_j (r_j^x-iar_j^y)]\\
&= -\braket{a}{U_j|{-}a}\langle{-a}|{U_j^\dag |{-}a}\rangle.
\end{align*}
All other terms remain unchanged under $\av \rightarrow \av'$. Note that the term $\braket{\mv}{a_1}\braket{a_{-1}}{\mv}$ is never changed under this map since $T(\av) \geq 1$ and only indices greater than $T(\av)$ are flipped. 
Since $\braket{a_{p+1}}{I|a_{-(p+1)}} = \braket{-a_{p+1}}{I|{-}a_{-(p+1)}} $, we have
\begin{equation*}
    \Bar{f}^I(\av') = - \Bar{f}^I(\av).
\end{equation*}
Similarly, since $\braket{a_{p+1}}{Z|a_{-(p+1)}} = -\braket{-a_{p+1}}{Z|{-}a_{-(p+1)}} $, we have
\begin{equation*}
    \Bar{f}^Z(\av') = \Bar{f}^Z(\av).
\end{equation*}
\end{proof}
\vspace{-2em}
By employing the antisymmetry of $\Bar{f}^I$ under this transformation, we can prove a symmetry of $H_D$.
\begin{lemma} \label{lem:H-D-prop}
For all $0\le m \le p$, we have
\begin{align}
    H_d^{(m)}(\av) &= 1 &\textnormal{ if } T(\av)=0, \\
    H_d^{(m)}(\av') &= H_d^{(m)}(\av)  &\textnormal{ if } T(\av)>0
\end{align}
\end{lemma}
\begin{proof}
We prove by induction on $m$. The base case $m=0$ follows from the definition $H_d^{(p+1)}(\av)=1$ for all $\av\in B$.
For $m>0$, we can break the sum over $\bv$ in the definition of $H_{d}^{(m)}(\av)$ into pieces as
\begin{align}
     H_d^{(m)}(\av) &= \bigg(\sum_{\bv:~ T(\bv) \le T(\av)} \Bar{f}^I(\bv) H_d^{(m-1)}(\bv) \exp{\Big[ {\textstyle -\frac{i}{\sqrt{d}} }\Av \cdot (\av \bv) \Big]} \nonumber \\
    &\quad \quad + \frac12 \sum_{\bv :~ T(\bv) > T(\av)} \Bar{f}^I(\bv) H_d^{(m-1)}(\bv) \bigg[ \exp{\Big[ {\textstyle -\frac{i}{\sqrt{d}}}\Av \cdot (\av \bv) \Big]} - \exp{\Big[ {\textstyle -\frac{i}{\sqrt{d}} }\Av \cdot (\av \bv') \Big]} \bigg] \bigg)^d
    \label{eq:H_D_expanded}
\end{align}
where in the last line we used the $\Bar{f}^I(\bv')=-\Bar{f}^I(\bv)$ and the inductive assumption that $H_D^{(m-1)}(\bv')=H_D^{(m-1)}(\bv)$.

We first show that the second sum evaluates to zero. To see this, note
\begin{equation*}
    \Av\cdot(\av\bv) = \sum_{r=1}^p \alpha_r (a_r b_r - a_{-r}b_{-r}) = \sum_{r=1}^{\max\{T(\av), T(\bv)\} } \alpha_r (a_r b_r - a_{-r} b_{-r}).
\end{equation*}
When $T(\bv) \ge T(\av)$, we have
\begin{equation} \label{eq:Aabprime-ident}
    \Av\cdot(\av\bv') = \sum_{r=1}^{\max\{T(\av), T(\bv')\} } \alpha_r (a_r b_r' - a_{-r} b_{-r}') = \sum_{r=1}^{\max\{T(\av), T(\bv)\} } \alpha_r (a_r b_r - a_{-r} b_{-r}) =\Av\cdot(\av\bv)
\end{equation}
So indeed the second sum in \eqref{eq:H_D_expanded} vanishes.

This gives us a simplified form of $H_d^{(m)}(\av)$ as
\begin{equation}
    H_d^{(m)}(\av) = \bigg(\sum_{\bv:~ T(\bv) \le T(\av)} \Bar{f}^I(\bv) H_d^{(m-1)}(\bv) \exp{\Big[ {\textstyle -\frac{i}{\sqrt{d}} }\Av \cdot (\av \bv) \Big]}\bigg)^d
\end{equation}
It is easy to see that if $T(\av)=0$, then
\begin{align*}
    H_d^{(m)}(\av) = \bigg(\sum_{\bv:~ T(\bv) =0} \Bar{f}^I(\bv) H_d^{(m-1)}(\bv) \exp{\Big[ {\textstyle -\frac{i}{\sqrt{d}} }\Av \cdot (\av \bv) \Big]}\bigg)^d = \Big[\sum_{\bv: T(\bv)=0} \Bar{f}^I(\bv)\Big]^d = 1
\end{align*}
since $H_D^{(m-1)}(\bv)=1$ when $T(\bv)=0$ by inductive assumption, and $\Av\cdot(\av\bv)=0$ when $T(\av)=T(\bv)=0$.

Furthermore, when $T(\av)>0$, we have
\begin{align*}
    H_d^{(m)}(\av') &= \bigg(\sum_{\bv:~ T(\bv) \le T(\av')} \Bar{f}^I(\bv) H_d^{(m-1)}(\bv) \exp{\Big[ {\textstyle -\frac{i}{\sqrt{d}} }\Av \cdot (\av' \bv) \Big]}\bigg)^d \\
    &= \bigg(\sum_{\bv:~ T(\bv) \le T(\av)} \Bar{f}^I(\bv) H_d^{(m-1)}(\bv) \exp{\Big[ {\textstyle -\frac{i}{\sqrt{d}} }\Av \cdot (\av \bv) \Big]}\bigg)^d  = H_d^{(m)}(\av)
\end{align*}
where we used $\Av\cdot(\av'\bv)=\Av\cdot(\av\bv)$ when $T(\av)\ge T(\bv)$ by switching the roles of $\av$ and $\bv$ in \eqref{eq:Aabprime-ident}.
\end{proof}
\vspace{-1em}
\begin{lemma} \label{lem:fH=1}
\begin{equation}
    \sum_{\bv\in B} \Bar{f}^I(\bv) H_d^{(m)}(\bv) = \sum_{\bv\in B_0} \Bar{f}^I(\bv) H_d^{(m)}(\bv) = 1
\end{equation}
\end{lemma}
\begin{proof}
    We have
    \begin{align}
        \sum_{\av\in B} \Bar{f}^I(\av)H_d^{(m)}(\av) &= \sum_{\av\in B_0} \Bar{f}^I(\av)H_d^{(m)}(\av) + \frac{1}{2}\sum_{\av\notin B_0} \Big[\Bar{f}^I(\av)H_d^{(m)}(\av) + \Bar{f}^I(\av')H_d^{(m)}(\av')\Big]\\
        &= \sum_{\av\in B_0} \Bar{f}^I(\av) +\frac{1}{2}\sum_{\av\notin B_0} \Big[\Bar{f}^I(\av)H_d^{(m)}(\av) - \Bar{f}^I(\av)H_d^{(m)}(\av)\Big] \\
        &= \sum_{\av\in B_0} \Bar{f}^I(\av) = 1
    \end{align}
    where in the second line we use Lemma \ref{lem:H-D-prop} and in the final step we combine Lemma \ref{lem:sum_f_1} and Lemma \ref{lem:prime_op}.
\end{proof}
\subsection{Parameter-Dependent Properties of Iterative Formula Functions}
Next, we note special properties of these functions that appear only at particular values of the QAOA parameters. Here, we utilize the bar notation to differentiate between averaged and unaveraged $f$ functions.
\begin{lemma}\label{lemma:f_symm}
Take the HamQAOA ansatz at any depth $p \in \mathbb{N}$ with all $\gamma$ parameters fixed to 0, and with the initial state $\ket{\mv} = \ket{\sv}$ chosen to be a product state only over $X$-eigenstates as in Algorithm \ref{HamQAOA}. Then for the function $\Bar{f}$ in Formula \ref{alg:fin_iter}, 
\begin{align}
    &\Bar{f}^I(\vect{z}) = \Bar{f}^I(-\vect{z})\\
    &\Bar{f}^X(\vect{z}) = {\Bar{f}}^X(-\vect{z})\\
    &{\Bar{f}}^Y(\vect{z}) = -{\Bar{f}}^Y(-\vect{z})\\
    &{\Bar{f}}^Z(\vect{z}) = -{\Bar{f}}^Z(-\vect{z})
\end{align}
where $\vect{z}\in\{\pm 1\}^{2p+2}$.
\end{lemma}
\begin{proof}
    \normalfont{}
    By definition, we have
    \begin{equation}
        \Bar{f}^I(\vect{z}) = \mathbb{E}_{s, \nv_v}[f_{s,\nv_v}(\vect{z})\delta_{z^{(p+1)}, z^{-(p+1)}}]
    \end{equation}
    Since all $\gamma = 0$, we have
    \begin{align}
        f_{s,\nv_v}(\vect{z}) = \braket{s}{z_1}&\bra{z_1}e^{i\beta_1B}e^{i\delta_1D}\ket{z_2}...\bra{z_{p-1}}e^{i\beta_{p-1}B}e^{i\delta_{p-1}D}\ket{z_p}     \bra{z_p}e^{i\beta_{p}B}e^{i\delta_{p}D}\ket{z_{p+1}}\\
    &\bra{z_{p+1}}\sigma\ket{z_{-(p+1)}}\bra{z_{-(p+1)}}e^{i\beta_{-p}B}e^{i\delta_{-p}D}\ket{z_{-p}}...\bra{z_{-2}}e^{i\beta_{-1}B}e^{i\delta_{-1}D}\ket{z_{-1}}\braket{z_{-1}}{s}
    \end{align}
    Noting that $\ket{s}$ is an $X$-basis state, we have
    \begin{equation}
        \braket{s}{z_1}\braket{z_{-1}}{s} = \braket{s}{-z_1}\braket{-z_{-1}}{s}
    \end{equation}
    For the remaining terms containing drivers, we have
    \begin{equation}
    \bra{z_a}e^{i\beta_{a}B}e^{i\delta_{a}D}\ket{z_b} =  \bra{-z_a}e^{i\beta_{a}B}e^{i\delta_{a}D}\ket{-z_b}
    \end{equation}
    Both of the above two relations can be verified by simply inserting $X^2 = I$ into the appropriate locations and acting on the z-basis states with $X$ to flip signs and pull out coefficients. Applying these relations to all terms in $f_s$, we have
    \begin{equation}
    \label{eq:fs_sym}
        f_{s, \nv_v}(\vect{z}) = f_{s,\nv_v}(-\vect{z}) 
    \end{equation}
    since any $z_1^2 = 1$. Then
    \begin{equation}
         \Bar{f}^I(\vect{z}) = \mathbb{E}_{s,\nv_v}[f_{s,\nv_v}(-\vect{z})\delta_{z^{(p+1)}, z^{-(p+1)}}] = \Bar{f}^I(-\vect{z})
    \end{equation}
    The remaining equalities follow directly from equation \eqref{eq:fs_sym}.
    \begin{align}
        &\Bar{f}^X(\vect{z}) = \mathbb{E}_{s,\nv_v}[f_{s,\nv_v}(-\vect{z})\langle{z^{(p+1)}}|X|{z^{-(p+1)}}\rangle] = \mathbb{E}[f_{s,\nv_v}(-\vect{z})\delta_{z^{(p+1)}, -z^{-(p+1)}}] = \Bar{f}^X(-\zv) \\ 
        &\Bar{f}^Y(\vect{z}) = \mathbb{E}_{s,\nv_v}[f_{s,\nv_v}(-\vect{z})\langle{z^{(p+1)}}|Y|{z^{-(p+1)}}\rangle] = -\mathbb{E}[f_{s,\nv_v}(-\vect{z})\langle{-z^{(p+1)}}|Y|{-z^{-(p+1)}}\rangle] = -\Bar{f}^Y(-\zv)\\
        &\Bar{f}^Z(\vect{z}) = \mathbb{E}_{s,\nv_v}[f_{s,\nv_v}(-\vect{z})\langle{z^{(p+1)}}|Z|{z^{-(p+1)}}\rangle] = -\mathbb{E}[f_{s,\nv_v}(-\vect{z})\langle{-z^{(p+1)}}|Z|{-z^{-(p+1)}}\rangle] = -\Bar{f}^Z(-\zv)
    \end{align}
\end{proof}
\noindent We can use this property to derive a symmetry relation for $H_D$.
\begin{corollary}
    Take the HamQAOA ansatz at depth $p\in\mathbb{N}$, with all $\gamma$ parameters fixed to 0. Consider Algorithm \ref{alg:fin_iter} for a regular graph of degree $d+1 \in \mathbb{N}$ and girth $\geq 2p+2$. Then $\forall k$ such that $0\leq k \leq p$ and all $\vect{z}\in\{\pm 1\}^{2p+2}$, $H_d^{(k)}(\vect{z}) = H_d^{(k)}(-\vect{z})$.
\end{corollary}
\begin{proof}
    \normalfont{}
    We proceed by induction on $k$. For $k = 0$ the equality is trivial as $H_D^{(1)}(\vect{z}) = 1$ for any $\vect{z}$. 

    For $m>1$, we have
    \begin{align}
        H_d^{(k)}(\vect{z}) &= \left( \sum_{\substack{\vect{x} \\ x^{(p+1)}=x^{-(p+1)}}} \exp( -i \vect{\mathcal{A}} \cdot (\vect{x} \vect{z})) H_d^{(k-1)}(\vect{x}) \bar{f}^I(\vect{x}) \right)^d \\ 
        &=\left(\frac{1}{2} \sum_{\substack{\vect{x} \\ x^{(p+1)}=x^{-(p+1)}}} \exp( -i \vect{\mathcal{A}} \cdot (\vect{x} \vect{z})) H_d^{(k-1)}(\vect{x}) \bar{f}^I(\vect{x})  + \exp(i \vect{\mathcal{A}} \cdot (\vect{x} \vect{z})) H_d^{(k-1)}(-\vect{x}) \bar{f}^I(-\vect{x})\right)^d \\
        &= \left( \sum_{\substack{\vect{x} \\ x^{(p+1)}=x^{-(p+1)}}} \cos( \vect{\mathcal{A}} \cdot (\vect{x} \vect{z})) H_d^{(k-1)}(\vect{x}) \bar{f}^I(\vect{x}) \right)^d\label{eq:cos_Hd}
    \end{align}
    In the second line we use Lemma \ref{lemma:f_symm} and the inductive hypothesis of the evenness of $H_d^{(k-1)}$. Since $\vect{z}$ appears only inside the cosine, which is even, the resulting function must be even in $\vect{z}$. Hence $H_d^{(k)}$ is even.
\end{proof}
To derive a formula in the limit of infinite degree, we need to establish that the functions dependent on $d$ indeed exist and are finite in this limit.
\begin{lemma}
    Assume $\mathcal{A} = \Theta(1/\sqrt{d})$. Then the function $H^{(k)}\equiv \lim_{d\rightarrow\infty} H_d^{(k)}$ exists and is finite $\forall k$ with $0\leq k \leq p$ when all parameters $\gamma$ are fixed to 0.
\end{lemma}
\begin{proof}
    \normalfont{}
    We again proceed by induction on $k$. The base case, $k = 0$, is trivial as the function is $1$ everywhere. For $k > 1$, we use equation \eqref{eq:cos_Hd}, but with $\vect{\Tilde{\mathcal{A}}} = \vect{\mathcal{A}}\sqrt{d}$:
    \begin{align}
        \lim_{d\rightarrow\infty} H_d^{(k)}(\vect{z}) &= \lim_{d\rightarrow\infty}\left( \sum_{\substack{\vect{x} \\ x^{(p+1)}=x^{-(p+1)}}} \cos\left( \frac{1}{\sqrt{d}}\vect{\Tilde{\mathcal{A}}} \cdot (\vect{x} \vect{z})\right) H_d^{(k-1)}(\vect{x}) \bar{f}^I(\vect{x}) \right)^d \\ 
        &=\lim_{d\rightarrow\infty}\left( \sum_{\substack{\vect{x} \\ x^{(p+1)}=x^{-(p+1)}}} \left(1 - \frac{1}{2d}(\vect{\Tilde{\mathcal{A}}} \cdot (\vect{x} \vect{z}))^2 + O(\frac{1}{d^2})\right) H_d^{(k-1)}(\vect{x}) \bar{f}^I(\vect{x}) \right)^d \\ 
        &= \exp\left(-\frac{1}{2}\sum_{\substack{\vect{x} \\ x^{(p+1)}=x^{-(p+1)}}} (\vect{\Tilde{\mathcal{A}}} \cdot (\vect{x} \vect{z}))^2 H^{(k-1)}(\vect{x}) \bar{f}^I(\vect{x})\right)
        \label{eq:inf_d_H}
    \end{align}
    In the last step we use Lemma \ref{lem:fH=1}. This is a well-defined function summed over finite variables, and by the inductive hypothesis, all terms are finite. Then the induction is complete.
\end{proof}
Further structure behind our parameters is needed for a sensible expression in the infinite-degree limit. The following assumption has been extensively numerically verified but not proven.
\begin{assumption}\label{lem:fx_H=0}
    When all $\gamma_i = 0$ and all $\beta_i = \frac{n\pi}{4}$, $n\in\mathbb{Z}$ for $i\in \{1, 2, \hdots p\}$:
    \begin{equation}
        \sum_{\av} \Bar{f}^X(\av)H_d^{(k)}(\av) = 0       \end{equation}
    for all $k \in \{1, 2, \hdots p\}$.
\end{assumption}
\vspace{-2em}
\subsection{Proof of $d\rightarrow\infty$ Formula}

\begin{proof}[Proof of Infinite-Degree Formula (Formula \ref{alg:d_inf_iter})]

Consider the HamQAOA with all parameters $\gamma_i = 0, \beta_i = n\pi/4$, for $n\in\mathbb{Z}$ and $i \in \{1, 2, \hdots p\}$. We now rearrange the sum as follows:
\begin{align}
    \nu_{p,d} = \frac{\sqrt{d}}{4}\sum_{\vect{z}_L, \vect{z}_R} \exp\left( -\frac{i}{\sqrt{d}} \Tilde{\vect{\mathcal{A}}} \cdot (\vect{z}_L \vect{z}_R)\right)& H_d^{(p)}(\vect{z}_L) H_d^{(p)}(\vect{z}_R) \Bar{f}^{\sigma_i}(\vect{z}_L) {\Bar{f}}^{\sigma_j}(\vect{z}_R) \\
    &+ \exp\left(\frac{i}{\sqrt{d}} \Tilde{\vect{\mathcal{A}}} \cdot (\vect{z}_L \vect{z}_R)\right) H_d^{(p)}(\vect{z}_L) H_d^{(p)}(-\vect{z}_R) \Bar{f}^{\sigma_i}(\vect{z}_L) {\Bar{f}}^{\sigma_j}(-\vect{z}_R)
\end{align}
Using Lemma \ref{lemma:f_symm}, we can consider cases for the value of $\sigma_i$ as follows.
\begin{enumerate}
    \item $\sigma_j = Y$ or $Z$:
    \begin{equation}
        \nu_{p,d} = \frac{i\sqrt{d}}{2}\sum_{\vect{z}_L, \vect{z}_R} \sin\left(\frac{\Tilde{\vect{\mathcal{A}}}}{\sqrt{d}} \cdot (\vect{z}_L \vect{z}_R)\right) H_d^{(p)}(\vect{z}_L) H_d^{(p)}(\vect{z}_R) \Bar{f}^{\sigma_i}(\vect{z}_L) \Bar{f}^{\sigma_j}(\vect{z}_R)
    \end{equation}
    Taking the infinite-$d$ limit and Taylor expanding,
    \begin{align}
        \lim_{d\rightarrow\infty} \nu_{p,d} &= \lim_{d\rightarrow\infty} \frac{i\sqrt{d}}{2}\sum_{\vect{z}_L, \vect{z}_R}\frac{\vect{\Tilde{\mathcal{A}}}}{\sqrt{d}}(\vect{z}_L \vect{z}_R)H_d^{(p)}(\vect{z}_L) H_d^{(p)}(\vect{z}_R) \Bar{f}^{\sigma_i}(\vect{z}_L) {\Bar{f}}^{\sigma_j}(\vect{z}_R) + O\left(\frac{1}{d}\right)\\
        &= \frac{i}{2}\sum_{\vect{z}_L, \vect{z}_R}\Tilde{\vect{\mathcal{A}}}(\vect{z}_L \vect{z}_R)H^{(p)}(\vect{z}_L) H^{(p)}(\vect{z}_R) \Bar{f}^{\sigma_i}(\vect{z}_L) {\Bar{f}}^{\sigma_j}(\vect{z}_R) \\ 
        &= \frac{i}{2}\sum_{j = -p}^{p}\Tilde{\alpha}_j\left(\sum_{\zv_L}\Bar{f}^{\sigma_i}(\zv_L)H^{(p)}(\zv_L)z_{L,j}\right)\left(\sum_{\zv_R}\Bar{f}^{\sigma_j}(\zv_R)H^{(p)}(\zv_R)z_{R,j}\right) \label{eq:vp_inf}
    \end{align}
    So far we have defined the $d\to\infty$ formula in terms of the limit $H^{(m)}(\av):= \lim_{d\to\infty} H_d^{(m)}(\av)$, which is still a $O(4^p)$-dimensional object.
We can further simplify by defining  $2p\times 2p$-dimensional matrices $G^{(m)}$ as follows:
\begin{equation}
    G^{(m)}_{j, k} = \sum_{\vect{z}} \bar{f}^{I}(\vect{z})H^{(m)}(\vect{z})z_j z_k
\end{equation}
To analyze this, we rewrite equation \eqref{eq:inf_d_H} as follows:
\begin{equation}
    H^{(m)} = \exp\left(-\frac{1}{2}\sum_{j, k = -p}^p \tilde{\alpha_j}\tilde{\alpha_j}a_j a_k (\sum_b H^{m-1}(\vect{b})\bar{f}^I(\vect{b}) b_j b_k)\right)
    \label{eq:inf_h_expr}
\end{equation}
Substituting, we have 
\begin{align}
     G^{(m)}_{j, k} &= \sum_{\vect{z}}\bar{f}^{I}(\vect{z})z_j z_k \left[\exp\left(-\frac{1}{2}\sum_{j', k' = -p}^p \tilde{\alpha}_{j'}\tilde{\alpha}_{k'}z_{j'} z_{k'} (\sum_b H^{m-1}(\vect{b})\bar{f}^I(\vect{b}) b_j b_k)\right)\right] \nonumber\\
     &= \sum_{\vect{z}}\bar{f}^{I}(\vect{z})z_j z_k \left[\exp\left(-\frac{1}{2}\sum_{j', k' = -p}^p \tilde{\alpha}_{j'}\tilde{\alpha}_{k'}z_{j'} z_{k'}G_{j, k}^{(m-1)}\right)\right]
\end{align}
with
\begin{equation}
    G_{j, k}^{(p+1)} = \sum_{\vect{z}} \bar{f}^{I}(\vect{z})z_j z_k
\end{equation}
by definition. Then we may define a vector $K^{\sigma}$ as follows, and substitute \eqref{eq:inf_h_expr}:
\begin{align}
    K_i^{\sigma} &= \sum_{\vect{z}} \bar{f}^{\sigma}(\vect{z})H^{(p)}(\vect{z})z_i \nonumber\\
    &= \sum_{\vect{z}}\bar{f}^{\sigma}(\vect{z})z_i \left[\exp\left(-\frac{1}{2}\sum_{j', k' = -p}^p \tilde{\alpha}_{j'}\tilde{\alpha}_{k'}z_{j'} z_{k'} \left(\sum_b H^{(p-1)}(\vect{b})\bar{f}^I(\vect{b}) b_j b_k\right)\right)\right] \nonumber\\
    &= \sum_{\vect{z}}\bar{f}^{\sigma}(\vect{z})z_i \left[\exp\left(-\frac{1}{2}\sum_{j', k' = -p}^p \tilde{\alpha}_{j'}\tilde{\alpha}_{k'}z_{j'} z_{k'}G_{j, k}^{(p-1)}\right)\right]
\end{align}
Substituting into \eqref{eq:vp_inf}, we are left with
\begin{equation}
\lim_{d\rightarrow\infty}\nu_{p,d}(\paramv, \sigma_i, \sigma_j) = \frac{i}{2}\sum_{l=-p}^p \tilde{\alpha}_l(K_l^{\sigma_i})(K_l^{\sigma_j})
\end{equation}
or in the case where $\sigma_i = \sigma_j$ on different sites as for the Heisenberg model, our final answer of
\begin{equation}
\lim_{d\rightarrow\infty}\nu_{p,d}(\paramv, \sigma, \sigma) = \frac{i}{2}\sum_{l=-p}^p \tilde{\alpha}_l(K_l^{\sigma})^2
\end{equation}
\item $\sigma_j = X$:
Again using Lemma \ref{lemma:f_symm},
\begin{equation}
    \nu_{p,d} = -\frac{\sqrt{d}}{2}\sum_{\vect{z}_L, \vect{z}_R} \cos\left(\frac{\Tilde{\vect{\mathcal{A}}}}{\sqrt{d}} \cdot (\vect{z}_L \vect{z}_R)\right) H_d^{(p)}(\vect{z}_L) H_d^{(p)}(\vect{z}_R) f^{\sigma_i}(\vect{z}_L) {f}^{\sigma_j}(\vect{z}_R)
\end{equation}
Taking the infinite-$D$ limit,
\begin{align}
    \lim_{d\rightarrow\infty} \nu_{p,d} &= \lim_{d\rightarrow\infty} \frac{-\sqrt{d}}{2}\sum_{\vect{z}_L, \vect{z}_R}H_d^{(p)}(\vect{z}_L) H_d^{(p)}(\vect{z}_R) \Bar{f}^{\sigma_i}(\vect{z}_L) {\Bar{f}}^{\sigma_j}(\vect{z}_R) + O\left(\frac{1}{\sqrt{d}}\right)\\ 
    &= \lim_{d\rightarrow\infty} -\frac{\sqrt{d}}{2}\left(\sum_{\zv_L}\Bar{f}^{\sigma_i}(\zv_L)H^{(p)}(\zv_L)\right)\left(\sum_{\zv_R}\Bar{f}^{\sigma_j}(\zv_R)H^{(p)}(\zv_R)\right)\\
    &= 0
\end{align}
where in the final step we use Assumption \ref{lem:fx_H=0}.
\end{enumerate}
To compute the HamQAOA's expected performance across edge $ij$ in this limit, then, we need only compute $\lim_{d\rightarrow\infty}\nu_{p,d}(\paramv, Y_i, Y_j) \text{and} \lim_{d\rightarrow\infty}\nu_{p,d}(\paramv, Z_i, Z_j)$. Computing one of these terms requires computing $p-1$ matrices $G^{(m)}$, each of which has $(2p+2)\times (2p+2)$ entries. Each of these entries requires an outer sum over $2^{2p+1}$ bitstrings (since we can apply the restriction that $z^{(p+1)} = z^{-(p+1)}$ as with $H_d^{(m)}$), and an inner sum over $O(p^2)$ bits. The vector $k$ requires the same time complexity for each entry, but only has $O(p)$ entries. The time complexity of evaluating this formula is then $O(p^34^p)$, a significant speedup over the $O(16^p)$ in the finite-degree case.  
\end{proof}

\section{Proofs of Algorithmic Guarantees}\label{sec:proof_guarantees}
\begin{proof}[Proof of Theorem \ref{Thm:bipartite_guarantee}]
    It is known for the original 2-driver QAOA, whose two drivers (in appearing order) we call $H_{mix}$ and $H_{obj}$, that the asymptotic convergence guarantee holds when the time-dependent Hamiltonian
    \begin{equation}
        H(t) = tH_{mix} + (1-t)H_{obj}
    \end{equation}
    has a unique maximum eigenvalue for all $0\leq t\leq 1$ \cite{farhi2014quantum}. This is due to the alternating structure simulating a Trotter approximation of an adiabatic evolution, and the fact that in the large-p limit, this Trotter approximation converges to an exact adiabatic evolution, which the prepares the exact max-energy state of the objective Hamiltonian if the initial state is the max-energy state of the mixing Hamiltonian and if $H(t)$ has a gapped largest eigenvalue. Hence, we prove Theorem \ref{Thm:bipartite_guarantee} by first showing that the HamQAOA parameters can be chosen to simulate layers of the form
    \begin{equation}
        e^{-i\delta D}e^{-i \theta H_{\rm QMC}}
    \end{equation}
    for any real $\delta, \theta$ and then showing that the time-dependent Hamiltonian
    $H(t) = tD + (1-t)H_{\rm QMC}$ has a unique maximum eigenvalue for all $0\leq t\leq 1$ when the graph is bipartite.

    The multivariate Lie Product Formula immediately gives us
    \begin{equation}
        e^{-i\theta H_{\rm QMC}} = \lim_{p\rightarrow\infty} \left(e^{\frac{-i\theta}{p} \sum_{j\sim k} X_jX_k } e^{\frac{-i\theta}{p} \sum_{j\sim k} Y_jY_k } e^{\frac{-i\theta}{p} \sum_{j\sim k} Z_jZ_k }\right)^p
    \end{equation}
    Moreover, the following are true for any $\theta$:
    \begin{align}
        &e^{i\frac{\pi}{4} \sum_j X_j} e^{-i\theta\sum_{j
        \sim k} Z_jZ_k}e^{-i\frac{\pi}{4} \sum_j X_j} = e^{-i\theta\sum_{j
        \sim k} Y_jY_k} \\
        &e^{i\frac{\pi}{4} \sum_j Z_j} e^{-i\theta\sum_{j
        \sim k} X_jX_k}e^{-i\frac{\pi}{4} \sum_j Z_j} = e^{-i\theta\sum_{j
        \sim k} X_jX_k}
    \end{align}
Then in the limit of large $p$, the HamQAOA circuit can exactly generate $e^{-i\theta H_{\rm QMC}}$. It already includes $e^{-i\delta D}$ as a driver. Hence, in the infinite-depth limit, the HamQAOA can exactly simulate the 2-driver QAOA of any depth.

    Next we will show the maximum-eigenvalue gap of the time-dependent evolution Hamiltonian. First, note that the QMC Hamiltonian can be rewritten as 
    \begin{equation}
        H_{\rm QMC} = 2\sum_{i\sim j}\ketbra{\psi^-}_{ij}, \quad \ket{\psi^-}_{ij} = \frac{1}{\sqrt{2}}(\ket{0_i1_j}-\ket{1_i0_j})
    \end{equation}
    
    Since the graph $G = (V, E)$ is bipartite, its MaxCut coincides with the bipartition of its vertices such that no two vertices in each partition share an edge. Let the two sets of the partition be $V_1$ and $V_2$ with $V_1\cup V_2 = V$ and $V_1\cap V_2 = \emptyset$. Now consider the unitary 
    \begin{equation}
        U(V_1) = \bigotimes_{j\in V_1} Z_j\bigotimes_{k\in V_2} I_k
    \end{equation}
    WLOG we can let the first qubit index in each singlet term in the Hamiltonian be an element of $V_1$. Then 
    \begin{equation}
        U(V_1)\left(tD + (1-t)H_{\rm QMC}\right)U(V_1)^\dagger = t\sum_{j\in V} X_j + 2(1-t)\sum_{i\sim j}\ketbra{\psi^+}_{ij}
    \end{equation}
    where $\ket{\psi^+} = \frac{1}{\sqrt{2}}(\ket{01} +\ket{10})$. Since t is nonnegative and $\leq 1$, the above matrix has only nonnegative entries in the $Z$-basis. Hence, the Perron-Frobenius theorem implies that the gap between the largest and second largest eigenvalues of this matrix is positive for all $t<1$. Since rotation by a unitary does not change the spectrum of the matrix, this is true for the time dependent matrix $tD + (1-t)H_{\rm QMC}$.

    Then the HamQAOA simulates a gapped 2-driver QAOA with mixing Hamiltonian $D$ and objective Hamiltonian $H_{\rm QMC}$ such that the time-dependent linear combination of the two is always gapped; by the adiabatic theorem, the HamQAOA prepares the exact max-energy state of $H_{\rm QMC}$ in the infinite-depth limit.  This proves Theorem \ref{Thm:bipartite_guarantee}.
    
\end{proof}

\begin{proof}[Proof of Lemma \ref{lemma:contains_prev}.]
First we show that HamQAOA at depth 1 can simulate the AGM algorithm for any choice of parameter $\theta$. Let $\paramv = (\theta, 0, 0, \pi/8)$. Then
\begin{align}
    \ket{\paramv} &= e^{-i\frac{\pi}{8}D}e^{-i\theta A}\ket{\sv} \\ 
    &=e^{-i\frac{\pi}{8}D}e^{-i\theta A}e^{i\frac{\pi}{8}D}\ket{\vect{s}}
\end{align}
where we use the fact that $\ket{\vect{s}}$ is an eigenstate of $D$. Now note that 
\begin{equation}
    e^{-is\frac{\pi}{8}X}Ze^{is\frac{\pi}{8}X} = \begin{cases} 
  \frac{Z-Y}{\sqrt{2}} & s = +1 \\
  \frac{Z+Y}{\sqrt{2}} & s = -1 
\end{cases}
\end{equation}
Define 
\begin{equation}
    \tilde{Z} \coloneqq \frac{Z-Y}{\sqrt{2}} \qquad \tilde{Y} \coloneqq \frac{Z+Y}{\sqrt{2}} 
\end{equation}
Then let
\begin{equation}
    \tilde{H}_{ZY} \coloneqq e^{-i\frac{\pi}{8}D}Ae^{i\frac{\pi}{8}D} = \sum_{i\sim j}\tilde{P}(i)\tilde{P}(j),\qquad \tilde{P}(i) = \begin{cases} 
  \tilde{Z} & s_i = +1 \\
  \tilde{Y} & s_i = -1 
\end{cases}
\end{equation}
Let us also restate the AGM driver Hamiltonian
\begin{equation}
    H_{XY} \coloneqq \sum_{i\sim j}P(i)P(j),\qquad P(i) = \begin{cases} 
  Y & b_i = 0 \\
  X & b_i = 1 
\end{cases}
\end{equation}
where $\ket{b_i} \in \{\ket{0}, \ket{1}\}$ denote the computational basis eigenstates.

The AGM algorithm optimizes a real parameter $\theta$ to prepare the state 
\begin{equation}
    \ket{\theta'} = e^{-i\theta' H_{XY}} \ket{\vect{b}},\quad \theta' = \argmax_\theta \bra{\vect{b}}e^{i\theta H_{XY}}H_{\rm QMC}e^{-i\theta H_{XY}}\ket{\vect{b}}
\end{equation}
Since $(\Tilde{Y}, \Tilde{Z}, X)$ and $(X, Y, Z)$ are both orthonormal bases in the Bloch sphere, there is a sequence of SU(2) rotations mapping between them. That is, there exists some SU(2) rotation $R$ such that 
\begin{equation}
    R\left(e^{-i\theta' H_{XY}} \ket{\vect{b}}\right) = e^{-i\theta' \Tilde{H}_{ZY}}\ket{\vect{s}}
\end{equation}
Let this state be denoted by $\ket{\Tilde{\theta'}}$. Then, since $H_{\rm QMC}$ is invariant under conjugation by any uniform $SU(2)$ rotation of all sites, 
\begin{equation}
    \bra{\Tilde{\theta'}} H_{\rm QMC} \ket{\Tilde{\theta'}} = \bra{\Tilde{\theta'}} R^\dagger H_{\rm QMC} R \ket{\Tilde{\theta'}} = \bra{{\theta'}} H_{\rm QMC} \ket{{\theta'}}
\end{equation}
Notice that the state prepared by the HamQAOA of depth 1 with parameters $\paramv$ is exactly $\ket{\theta'}$. So regardless of the value of $\theta'$, the $p=1$ HamQAOA with the chosen $\paramv$ simulates the AGM algorithm on any graph.
To obtain the comparison to \cite{King_2023} on edge-transitive graphs, observe that the SDP in \cite{King_2023} will have the same value on all edges due to the symmetry. In this case, the rounding algorithm in \cite{King_2023} reduces to the AGM algorithm.
\end{proof}

\end{document}